\documentclass[sigconf,nonacm=true,balance=false]{acmart}
\usepackage{popets}

\setcopyright{popets}
\copyrightyear{2024}

\acmYear{2024}
\acmVolume{2024}
\acmNumber{2}

\acmDOI{10.56553/popets-2024-0001}

\acmISBN{}
\acmConference{Proceedings on Privacy Enhancing Technologies}
\usepackage{dashbox}
\usepackage{colortbl}
\usepackage{multirow}
\usepackage{enumitem}
\usepackage{hyperref}
\usepackage{subcaption}
\usepackage{tikz}
\usetikzlibrary{matrix}

\hypersetup{
    colorlinks=true,
    linkcolor=blue,
    filecolor=magenta,      
    urlcolor=cyan,
}
\newtheorem{theorem}{Theorem}
\newtheorem{lemma}{Lemma}
\newtheorem{definition}{Definition}

\providecommand{\citep}{\cite}

\newcommand{\mycomment}[1]{// #1}
\newcommand{\myminorcomment}[1]{{\color{gray}// #1}}
\newcommand{\adv}{\ensuremath{\mathcal{A}}}
\newcommand{\simu}{\ensuremath{\mathcal{S}}}

\newcommand{\mixer}[1]{\ensuremath{\mathcal{M}_{#1}}}
\newcommand{\prover}[1]{\mixer{#1}}              
\newcommand{\verifier}{\ensuremath{\mathcal{V}}}       
\newcommand{\querier}{\ensuremath{\mathcal{Q}}}       

\newcommand{\setup}{\ensuremath{\mathsf{Setup}}}
\newcommand{\keygen}{\ensuremath{\mathsf{Keygen}}}
\newcommand{\enc}{\ensuremath{\mathsf{Enc}}}
\newcommand{\mix}{\ensuremath{\mathsf{Mix}}}

\newcommand{\btrin}{\ensuremath{\mathsf{BTraceIn}}}
\newcommand{\btrout}{\ensuremath{\mathsf{BTraceOut}}}
\newcommand{\renc}{\ensuremath{\mathsf{REnc}}}
\newcommand{\dec}{\ensuremath{\mathsf{Dec}}}
\newcommand{\tdec}{\ensuremath{\mathsf{TDec}}}

\newcommand{\sharealg}{\ensuremath{\mathsf{Share}}}
\newcommand{\recon}{\ensuremath{\mathsf{Recons}}}
\newcommand{\mult}{\ensuremath{\mathsf{Mult}}}

\newcommand{\otrin}{\ensuremath{\mathsf{OTraceIn}}}
\newcommand{\otrout}{\ensuremath{\mathsf{OTraceOut}}}

\newcommand{\Nat}{\mathbb{N}}
\newcommand{\Val}{\mathbb{V}}
\newcommand{\Zq}{\mathbb{Z}_{q}}
\newcommand{\Gone}{\mathbb{G}_{1}}
\newcommand{\Gtwo}{\mathbb{G}_{2}}
\newcommand{\GT}{\mathbb{G}_{T}}

\newcommand{\share}[2]{{#1}^{(#2)}}                       

\newcommand{\vect}[1]{\boldsymbol{#1}}                         
\newcommand{\vectcomp}[2]{{\vect{#1}}_{#2}}                    
\newcommand{\vectshare}[2]{\share{\vect{#1}}{#2}}              
\newcommand{\vectcompshare}[3]{\share{\vect{#1}}{#3}_{#2}}  

\newcommand{\ve}[1]{\vect{#1}}                   
\newcommand{\vc}[2]{\vectcomp{#1}{#2}}           
\newcommand{\vs}[2]{\vectshare{#1}{#2}}          
\newcommand{\vcs}[3]{\vectcompshare{#1}{#2}{#3}} 

\newcommand{\permfunc}[1]{\share{\pi}{#1}}                  

\newcommand{\si}[2]{#1 [\![ #2 ]\!] }   
\newcommand{\so}[2]{#1 [\![ #2 ]\!] }   

\newcommand{\cmpl}{\mathsf{Exp}_{\mathsf{completeness}}}
\newcommand{\snd}{\mathsf{Exp}_{\mathsf{soundness}}}
\newcommand{\secr}{\mathsf{Exp}_{\mathsf{secrecy}}}
\newcommand{\negl}{\mathsf{negl}}

\newcommand{\wit}[1]{\share{\omega}{#1}}                  

\newcommand{\nizkpktxt}{\ensuremath{\mathsf{NIZKPK}}}

\newcommand{\pokcomm}{\ensuremath{\rho_{\gamma}}}

\newcommand{\pkbody}[2]{\ensuremath{\{(#1): #2\}}}
\newcommand{\pk}[2]{\ensuremath{\mathsf{PK}\pkbody{#1}{#2}}}

\newcommand{\nizkpk}[2]{\ensuremath{\nizkpktxt\pkbody{#1}{#2}}}
\newcommand{\nizkver}{\ensuremath{\mathsf{NIZKVer}}}

\newcommand{\zkpsmtxt}{\ensuremath{\mathsf{DB}\text{-}\mathsf{SM}}}
\newcommand{\zkprsmtxt}{\ensuremath{\mathsf{DB}\text{-}\mathsf{RSM}}}
\newcommand{\zkpsm}{\ensuremath{{\zkpsmtxt}}}
\newcommand{\zkprsm}{\ensuremath{{\zkprsmtxt}}}

\newcommand{\encscheme}{\ensuremath{\mathsf{E}}}

\newcommand{\pke}{\ensuremath{\encscheme}}
\newcommand{\pkeg}{\ensuremath{\pke.\keygen}}
\newcommand{\pkee}{\ensuremath{\pke.\enc}}

\newcommand{\pked}{\ensuremath{\pke.\dec}}
\newcommand{\pkepk}[1]{\ensuremath{\share{\mathsf{pk}}{#1}}}
\newcommand{\pkesk}[1]{\ensuremath{\share{\mathsf{sk}}{#1}}}

\newcommand{\mpk}{\ensuremath{\mathsf{mpk}}}
\newcommand{\msk}[1]{\ensuremath{\share{\mathsf{msk}}{#1}}}

\newcommand{\theztxt}{\ensuremath{\mathsf{Pa}}}
\newcommand{\thez}{\ensuremath{\encscheme_{\theztxt}^{\mathsf{th}}}}
\newcommand{\thezg}{\ensuremath{\thez.\keygen}}
\newcommand{\theze}{\ensuremath{\thez.\enc}}
\newcommand{\thezre}{\ensuremath{\thez.\renc}}

\newcommand{\thezd}{\ensuremath{\thez.\tdec}}
\newcommand{\thezpk}{\ensuremath{\mathsf{pk}_{\theztxt}}}
\newcommand{\thezsk}[1]{\ensuremath{ \share{\mathsf{sk}}{#1}_{\theztxt} }}
\newcommand{\thezfunc}{\ensuremath{ \mathcal{F}_{\thezd}^{\mathsf{th}} }}

\newcommand{\thegtxt}{\ensuremath{\mathsf{EG}}}
\newcommand{\theg}{\ensuremath{\encscheme_{\thegtxt}^{\mathsf{th}}}}
\newcommand{\thegg}{\ensuremath{\theg.\keygen}}
\newcommand{\thege}{\ensuremath{\theg.\enc}}
\newcommand{\thegre}{\ensuremath{\theg.\renc}}

\newcommand{\thegd}{\ensuremath{\theg.\tdec}}
\newcommand{\thegpk}{\ensuremath{\mathsf{pk}_{\thegtxt}}}
\newcommand{\thegsk}[1]{\ensuremath{ \share{\mathsf{sk}}{#1}_{\thegtxt} }}
\newcommand{\thegfunc}{\ensuremath{ \mathcal{F}_{\thegd}^{\mathsf{th}} }}

\newcommand{\perm}[1]{\ensuremath{\mathsf{Perm}(#1)}}
\newcommand{\shuffle}{\ensuremath{\mathsf{Shuffle}}}

\begin{document}
\startPage{1}

\title{Traceable mixnets}


\author{Prashant Agrawal}
\affiliation{
  \institution{Department of Computer Science and Engineering, IIT Delhi}
  \city{}
  \country{}
}
\affiliation{
  \institution{Centre for Digitalisation, AI and Society, Ashoka University}
  \city{}
  \country{}
}
\email{prashant@cse.iitd.ac.in}

\author{Abhinav Nakarmi}
\affiliation{
  \institution{Department of Computer Science and Engineering, University of Michigan}
  \city{}
  \country{}
}
\email{nakarmi@umich.edu}
\authornote{Work done while at Ashoka University.}

\author{Mahabir Prasad Jhanwar}
\affiliation{
  \institution{Department of Computer Science and Centre for Digitalisation, AI and Society, Ashoka University}
  \city{}
  \country{}
}
\email{mahavir.jhawar@ashoka.edu.in}

\author{Subodh Vishnu Sharma}
\affiliation{
  \institution{Department of Computer Science and Engineering, IIT Delhi}
  \city{}
  \country{}
}
\email{svs@cse.iitd.ac.in}

\author{Subhashis Banerjee}
\affiliation{
  \institution{Department of Computer Science and Engineering, IIT Delhi}
  \city{}
  \country{}
}
\affiliation{
  \institution{Department of Computer Science and Centre for Digitalisation, AI and Society, Ashoka University}
  \city{}
  \country{}
}
\email{suban@ashoka.edu.in}

\renewcommand{\shortauthors}{Agrawal et al.}

\begin{abstract}
	We introduce the notion of \emph{traceable mixnets}. In a traditional mixnet, multiple mix-servers jointly permute and decrypt a list of ciphertexts to produce a list of plaintexts, along with a proof of correctness, such that the association between individual ciphertexts and plaintexts remains completely hidden. However, in many applications, the privacy-utility tradeoff requires answering some specific queries about this association, without revealing any information beyond the query result. We consider queries of the following types: a) given a ciphertext in the mixnet input list, whether it encrypts one of a given subset of plaintexts in the output list, and b) given a plaintext in the mixnet output list, whether it is a decryption of one of a given subset of ciphertexts in the input list. Traceable mixnets allow the mix-servers to jointly prove answers to the above queries to a querier such that neither the querier nor a threshold number of mix-servers learn any information beyond the query result. Further, if the querier is not corrupted, the corrupted mix-servers do not even learn the query result. We first comprehensively formalise these security properties of traceable mixnets and then propose a construction of traceable mixnets using novel distributed zero-knowledge proofs (ZKPs) of set membership and of a statement we call reverse set membership. Although set membership has been studied in the single-prover setting, the main challenge in our distributed setting lies in making sure that none of the mix-servers learn the association between ciphertexts and plaintexts during the proof. We implement our distributed ZKPs and show that they are faster than state-of-the-art by at least one order of magnitude.
\end{abstract}

\keywords{verifiable mixnets; traceability; distributed zero-knowledge proofs; set membership; reverse set membership}

\maketitle

\section{Introduction}
\label{intro}

A mixnet is a cryptographic primitive used for anonymous messaging, specifically for unlinking the identity information of data providers with the sensitive data they provide. Let $\vc{u}{i}$ denote the identity information of the $i^{\text{th}}$ individual and $\vc{v}{i}$ denote the sensitive data they contribute. For example, in a secure electronic voting context where mixnets are commonly used, $\vc{u}{i}$ may contain a voter id and $\vc{v}{i}$ may contain the vote. The link between $\vc{u}{i}$ and $\vc{v}{i}$ is hidden by encrypting $\vc{v}{i}$ to obtain ciphertexts (encrypted votes) $\vc{c}{i}$, uploading $\vc{u}{i}$ along with $\vc{c}{i}$ to an input list, and feeding the list of ciphertexts $\ve{c}$ as the mixnet input. The mixnet, which consists of a series of \emph{mix-servers}, processes these ciphertexts and outputs a list $\ve{v}'$ of decryptions of the ciphertexts in a randomly permuted order \cite{Chaum-mixnet}. The secret permutation that links the ciphertexts with their corresponding plaintexts is shared among the mix-servers and remains hidden unless a threshold number of mix-servers are corrupted, thus completely hiding which ciphertext or voter id corresponds to which vote (see Figure \ref{fig:mixnet}). Using \emph{verifiable mixnets} \citep{mixnet-sok}, the mix-servers can also \emph{prove} to a verifier that the output list is obtained correctly by permuting and decrypting each element of the input list, while keeping the linkages between the ciphertexts and plaintexts completely hidden.

\begin{figure}
	\centering
	\begin{subfigure}[b]{\linewidth}
		\centering
		\scalebox{0.75}{
			\begin{tabular}{c@{\hskip 0.1cm} c@{\hskip 0.1cm} c@{\hskip 0.1cm} c@{\hskip 0.2cm}}
				\begin{tabular}{|c|}
					\hline
					\cellcolor{gray!20}Input \\
					\hline
					$\vc{u}{1}, \vc{c}{1}$ \\
					\hline
					$\vc{u}{2}, \vc{c}{2}$ \\
					\hline
					$\vc{u}{3}, \vc{c}{3}$ \\
					\hline
					$\vc{u}{4}, \vc{c}{4}$ \\
					\hline
					$\vc{u}{5}, \vc{c}{5}$ \\
					\hline
				\end{tabular} & 
				\begin{tabular}{c|c|}
					\multicolumn{2}{r}{$\mixer{1}$} \\
					\cline{2-2}
					$\rightarrow$ & \multirow{5}{*}{$\permfunc{1}$} \\
					$\vdots$ &  \\
					$\vdots$ & \\
					$\rightarrow$ & \\
					\cline{2-2}
				\end{tabular} &
				\begin{tabular}{c|c|c}
					\multicolumn{3}{c}{$\mixer{2}$} \\
					\cline{2-2}
					$\rightarrow$ & \multirow{5}{*}{$\permfunc{2}$} & $\rightarrow$\\
					$\vdots$ &  & $\vdots$ \\
					$\vdots$ & & $\vdots$\\
					$\rightarrow$ & & $\rightarrow$\\
					\cline{2-2}
				\end{tabular} &
				\begin{tabular}{|c|}
					\hline
					\multicolumn{1}{|c|}{\cellcolor{gray!20}Output} \\
					\hline
					$\vc{v}{3}$ \\
					\hline
					$\vc{v}{5}$ \\
					\hline
					$\vc{v}{2}$ \\
					\hline
					$\vc{v}{1}$ \\
					\hline
					$\vc{v}{4}$ \\
					\hline
				\end{tabular}
			\end{tabular}
		}
		\caption{{\small A traditional verifiable mixnet.}}
		\label{fig:mixnet}
	\end{subfigure}
	\hfill
	\begin{subfigure}[b]{\linewidth}
		\centering
		\scalebox{0.75}{
			\begin{tabular}{c@{\hskip 0.1cm} c@{\hskip 0.1cm} c@{\hskip 0.1cm} c@{\hskip 0.2cm}}
				\begin{tabular}{|l|}
					\hline
					\cellcolor{gray!20}Input \\
					\hline
					$\vc{u}{1}, \vc{c}{1}$ \\
					\hline
					$\vc{u}{2}, \vc{c}{2} \enskip \times$ \\
					\hline
					$\vc{u}{3}, \vc{c}{3}$ \\
					\hline
					$\vc{u}{4}, \vc{c}{4}$ \\
					\hline
					$\vc{u}{5}, \vc{c}{5}$ \\
					\hline
				\end{tabular} & 
				\begin{tabular}{c|c|}
					\multicolumn{2}{r}{$\mixer{1}$} \\
					\cline{2-2}
					$\rightarrow$ & \multirow{5}{*}{$\permfunc{1}$} \\
					$\vdots$ &  \\
					$\vdots$ & \\
					$\rightarrow$ & \\
					\cline{2-2}
				\end{tabular} &
				\begin{tabular}{c|c|c}
					\multicolumn{3}{c}{$\mixer{2}$} \\
					\cline{2-2}
					$\rightarrow$ & \multirow{5}{*}{$\permfunc{2}$} & $\rightarrow$\\
					$\vdots$ & & $\vdots$ \\
					$\vdots$ & & $\vdots$\\
					$\rightarrow$ & & $\rightarrow$\\
					\cline{2-2}
				\end{tabular} &
				\begin{tabular}{|c|}
					\hline
					\multicolumn{1}{|c|}{\cellcolor{gray!20}Output} \\
					\hline
					$\vc{v}{3}$ \\
					\hline
					\cellcolor{green!50}$\vc{v}{5}$ \\
					\hline
					\cellcolor{green!50}$\vc{v}{2}$ \\
					\hline
					\cellcolor{green!50}$\vc{v}{1}$ \\
					\hline
					$\vc{v}{4}$ \\
					\hline
				\end{tabular}
			\end{tabular} 
		}
		\caption{{\small A $\text{TraceIn}(\vc{c}{2}, \{\vc{v}{5}, \vc{v}{2}, \vc{v}{1}\})$ query in a traceable mixnet: whether $\vc{c}{2}$ encrypted one of $\{\vc{v}{5}, \vc{v}{2}, \vc{v}{1}\}$.}}
		\label{fig:trin}
	\end{subfigure}
	\hfill
	\begin{subfigure}[b]{\linewidth}
		\centering
		\scalebox{0.75}{
			\begin{tabular}{c@{\hskip 0.1cm} c@{\hskip 0.1cm} c@{\hskip 0.1cm} c@{\hskip 0.2cm}}
				\begin{tabular}{|c|}
					\hline
					\cellcolor{gray!20}Input \\
					\hline
					\cellcolor{green!50}$\vc{u}{1}, \vc{c}{1}$ \\
					\hline
					\cellcolor{green!50}$\vc{u}{2}, \vc{c}{2}$ \\
					\hline
					\cellcolor{green!50}$\vc{u}{3}, \vc{c}{3}$ \\
					\hline
					$\vc{u}{4}, \vc{c}{4}$ \\
					\hline
					$\vc{u}{5}, \vc{c}{5}$ \\
					\hline
				\end{tabular} & 
				\begin{tabular}{c|c|}
					\multicolumn{2}{r}{$\mixer{1}$} \\
					\cline{2-2}
					$\rightarrow$ & \multirow{5}{*}{$\permfunc{1}$} \\
					$\vdots$ &  \\
					$\vdots$ & \\
					$\rightarrow$ & \\
					\cline{2-2}
				\end{tabular} &
				\begin{tabular}{c|c|c}
					\multicolumn{3}{c}{$\mixer{2}$} \\
					\cline{2-2}
					$\rightarrow$ & \multirow{5}{*}{$\permfunc{2}$} & $\rightarrow$\\
					$\vdots$ & & $\vdots$ \\
					$\vdots$ & & $\vdots$\\
					$\rightarrow$ & & $\rightarrow$\\
					\cline{2-2}
				\end{tabular} &
				\begin{tabular}{|l|}
					\hline
					\multicolumn{1}{|c|}{\cellcolor{gray!20}Output} \\
					\hline
					$\vc{v}{3}$ \\
					\hline
					$\vc{v}{5}$ \\
					\hline
					$\vc{v}{2}$ \\
					\hline
					$\vc{v}{1} \enskip \times$ \\
					\hline
					$\vc{v}{4}$ \\
					\hline
				\end{tabular}
			\end{tabular} 
		}
		\caption{{\small A $\text{TraceOut}(\{\vc{c}{1}, \vc{c}{2}, \vc{c}{3}\}, \vc{v}{1})$ query in a traceable mixnet: whether $\vc{v}{1}$ was encrypted in one of $\{\vc{c}{1}, \vc{c}{2}, \vc{c}{3}\}$.}}
		\label{fig:trout}
	\end{subfigure}
	\caption{{\small Traditional and traceable mixnets. $\vc{u}{i}$ denotes the $i^{\text{th}}$ individual's identity information and $\vc{v}{i}$ denotes their sensitive data; $\vc{c}{i}$s encrypt $\vc{v}{i}$s and are passed as input to the mixnet; the mixnet consists of mix-servers $\mixer{1}$ and $\mixer{2}$ that jointly decrypt and permute input list $\ve{c}$ to output a plaintext list $\ve{v}':=(\ve{v}_{\pi(i)})_{i=1}^{5}$, where $\pi$ is composed of secret permutations $\permfunc{1}$ and $\permfunc{2}$ of $\mixer{1}$ and $\mixer{2}$.
	}}
	\label{fig:mixnet-sm-overview}
\end{figure}

However, sometimes, it is necessary to reveal specific partial information about the association between $\vc{c}{i}$ and $\vc{v}{i}$, while still preventing any additional information leakage. For example, consider a large-scale public election running across multiple polling booths. Further, assume a dual voting setup \citep{benaloh-dual-voting,pretavoter-hrpat} where voters cast their vote both electronically and on paper: the electronic system produces an encrypted vote as the voter receipt and processes these encrypted votes via a mixnet, whereas the paper votes are collected to form a physical audit trail. Such dual voting systems aim to improve the overall robustness and transparency of electoral processes. In such systems, revealing partial information about the encrypted and decrypted votes allows graceful recovery from disputes without re-running the entire election \citep{openvoting}. For example, if there is a mismatch between a plaintext vote in the mixnet output list and its corresponding paper record, the ability to pinpoint which specific booth the disputed vote came from enables a \emph{localised recovery} of the election by selectively rerunning the election only at the corrupted booths. Yet, it is important to not reveal additional information such as vote counts of all the polling booths or, worse, votes of individual voters. Further, any such partial information revealed should be \emph{provable}, to ensure that recovery steps lead to the correct election outcome. A traditional mixnet does not allow for the release of such controlled partial information because votes cast at all the booths are anonymised together before decryption. 

Conversely, consider a voter claiming that their published encrypted vote does not match the one in their (possibly fake) receipt. If it can be shown that their published encrypted vote decrypted to a plaintext vote that mismatches with its corresponding paper record, it supports the voter's claim because an incorrectly uploaded encrypted vote must also mismatch with the paper vote on decryption. As above, this revealed information should be provable and should not leak any additional information, e.g., the voter's specific vote. Traditional mixnets do not allow this either. 

To address this gap, we introduce the notion of traceable mixnets. Let $\vc{c}{I}:=\{\vc{c}{i} \mid i \in I\}$ for some index set $I$ denote a subset of ciphertexts in the input list and $\vc{v}{J}':=\{\vc{v}{j}' \mid j \in J\}$ for some index set $J$ denote a subset of plaintexts in the output list. A traceable mixnet allows its mix-servers to jointly and provably answer the following queries to an interested and authorised querier:
\begin{itemize}[leftmargin=*]
	\item[-] $\textbf{TraceIn}(\vc{c}{i}, \vc{v}{J}')$: Does the ciphertext $\vc{c}{i}$ in the input list encrypt a plaintext in set $\vc{v}{J}'$ (Figure \ref{fig:trin})? Given an encrypted vote $\vc{c}{i}$ mismatching a voter receipt and a set $\vc{v}{J}'$ of plaintext votes mismatching with their paper records, this query answers if $\vc{c}{i}$ decrypted to a mismatching plaintext vote and thus was incorrect.  
	\item[-] $\textbf{TraceOut}(\vc{c}{I}, \vc{v}{j}')$: Is the plaintext $\vc{v}{j}'$ in the output list encrypted in a ciphertext in set $\vc{c}{I}$ (Figure \ref{fig:trout})? Given an output plaintext vote $\vc{v}{j}'$ mismatching with its paper vote and a set $\vc{c}{I}$ of encrypted votes cast at a given polling booth $B$, this query answers if the disputed vote $\vc{v}{j}'$ came from booth $B$. 
\end{itemize}
The queries are answered such that no \emph{additional information} beyond the query output is leaked to an adversary controlling the querier and less than a threshold number of mix-servers. Also, to prevent mix-servers from accumulating query responses issued to different queriers over time, the mix-servers are not allowed to even learn the output of a query if they do not control the querier. Such controlled querying mechanism is useful in voting as well as several other applications where privacy-preserving data sharing with guaranteed correctness is required.

Note, though, that the query outputs may themselves leak sensitive information, especially when multiple queries are combined. Thus, arbitrary queries cannot be allowed in any application. Deciding what queries to allow requires a privacy risk analysis of the overall information leaked by the queries, followed by an analysis of whether the leakage is acceptable for the application's privacy requirements. For example, in the voting application, the allowed queries are decided so that the recovery process does not leak the vote counts of any booth except the corrupted booths \citep{openvoting}. Once a policy detailing the allowed queries is decided, possibly allowing different queries to different queriers, the application layer must also ensure that the mix-servers comply with the policy. The traceable mixnet guarantees that if an honest mix-server follows the policy, the adversary learns nothing about the output of the queries disallowed by the policy. Without such analysis and policy compliance though, a traceable mixnet solution should not be deployed.

We emphasise that a traceable mixnet's secrecy requirements are more stringent than what existing proof-of-shuffle and verifiable decryption techniques in the area of verifiable mixnets \citep{mixnet-sok} can provide. For example, given a mixnet input list $(\vc{c}{1}, \dots, \vc{c}{5})$ and output list $(\vc{v}{1}',\dots,\vc{v}{5}')$, a TraceIn$(\vc{c}{4},$ $\{\vc{v}{2}',$ $\vc{v}{3}',$ $\vc{v}{4}' \})$ query requires proving in ZK that $\vc{c}{4}$ decrypts to one of $\vc{v}{2}'$, $\vc{v}{3}'$ and $\vc{v}{4}'$. Using these techniques, the mix-servers can prove this by, e.g., verifiably revealing the subset $C:=\{\vc{c}{1}, \vc{c}{2}, \vc{c}{4}\}$ of the input ciphertexts that decrypt to the set of plaintexts $\{\vc{v}{2}', \vc{v}{3}', \vc{v}{4}' \}$ and letting the querier verify that $\vc{c}{4} \in C$ (see Figure \ref{fig:mixnet-vs-sm}b). However, this is not ZK: it reveals, e.g., that ciphertexts $\vc{c}{1}, \vc{c}{2}$ decrypt to one of $\vc{v}{2}'$, $\vc{v}{3}'$ and $\vc{v}{4}'$ (and not $\vc{v}{1}'$ or $\vc{v}{5}'$). Figure \ref{fig:mixnet-vs-sm} shows other similar leakage scenarios. A traceable mixnet does not reveal any such intermediate information.

Finally, our focus is on offline batch processing, i.e., efficiently answering multiple TraceIn/TraceOut queries against the same set. Thus, we define the following batched queries: $a)$ $\textbf{BTraceIn}(\vc{c}{I}, \vc{v}{J}')$: which ciphertexts in set $\vc{c}{I}$ encrypt a plaintext in set $\vc{v}{J}'$; and $b)$ $\textbf{BTraceOut}(\vc{c}{I}, \vc{v}{J}')$: which plaintexts in set $\vc{v}{J}'$ are encrypted in a ciphertext in set $\vc{c}{I}$? We require the amortised time for batched queries to be linear in the input list size to enable practical applications like dispute resolution in elections with millions of votes. Note that offline batch processing in verifiable mixnets is distinct from mixnets in real-time anonymous communication networks \citep{tor}. Thus, we do not aim to answer the queries in real time.

\begin{figure}
	\centering
	\scalebox{0.75}{
		\begin{tabular}{l@{\hskip 0.1cm}l@{\hskip 0.15cm}}
			\begin{tabular}{c@{\hskip 0.1cm} c@{\hskip 0.15cm}}
				\begin{tabular}{|l|}
					\hline
					$\vc{c}{1}$ \\
					\hline
					$\vc{c}{2}$ \\
					\hline
					$\vc{c}{3}$ \\
					\hline
					$\vc{c}{4}$ $\times$ \\
					\hline
					$\vc{c}{5}$ \\
					\hline
				\end{tabular} & 
				\begin{tabular}{|l|}
					\hline
					$\vc{v}{1}'$ \\
					\hline
					\cellcolor{green!50}$\vc{v}{2}'$ \\
					\hline
					\cellcolor{green!50}$\vc{v}{3}'$ {\color{red}$\bullet$} \\
					\hline
					\cellcolor{green!50}$\vc{v}{4}'$ \\
					\hline
					$\vc{v}{5}'$ \\
					\hline
				\end{tabular} \\
				& \\
				\multicolumn{2}{c}{$(a)$} \\
			\end{tabular}
			
			\begin{tabular}{c@{\hskip 0.1cm} c@{\hskip 0.2cm}}
				\begin{tabular}{|l|}
					\hline
					\cellcolor{red!50}$\vc{c}{1}$ \\
					\hline
					\cellcolor{red!50}$\vc{c}{2}$ \\
					\hline
					$\vc{c}{3}$ \\
					\hline
					\cellcolor{red!50}$\vc{c}{4}$  $\times$ \\
					\hline
					$\vc{c}{5}$ \\
					\hline
				\end{tabular} & 
				\begin{tabular}{|l|}
					\hline
					$\vc{v}{1}'$ \\
					\hline
					\cellcolor{green!50}$\vc{v}{2}'$ \\
					\hline
					\cellcolor{green!50}$\vc{v}{3}'$ \\
					\hline
					\cellcolor{green!50}$\vc{v}{4}'$ \\
					\hline
					$\vc{v}{5}'$ \\
					\hline
				\end{tabular} \\
				& \\
				\multicolumn{2}{c}{$(b)$} \\
			\end{tabular}
			
			\begin{tabular}{c@{\hskip 0.1cm} c@{\hskip 0.15cm}}
				\begin{tabular}{|l|}
					\hline
					$\vc{c}{1}$ \\
					\hline
					$\vc{c}{2}$ \\
					\hline
					\cellcolor{green!50}$\vc{c}{3}$ {\color{red}$\bullet$} \\
					\hline
					\cellcolor{green!50}$\vc{c}{4}$ \\
					\hline
					\cellcolor{green!50}$\vc{c}{5}$ \\
					\hline
				\end{tabular} & 
				\begin{tabular}{|l|}
					\hline
					$\vc{v}{1}'$ \\
					\hline
					$\vc{v}{2}'$ $\times$ \\
					\hline
					$\vc{v}{3}'$ \\
					\hline
					$\vc{v}{4}'$ \\
					\hline
					$\vc{v}{5}'$ \\
					\hline
				\end{tabular} \\
				& \\
				\multicolumn{2}{c}{$(c)$} \\
			\end{tabular}
			
			\begin{tabular}{c@{\hskip 0.1cm} c@{\hskip 0.2cm}}
				\begin{tabular}{|l|}
					\hline
					$\vc{c}{1}$ \\
					\hline
					$\vc{c}{2}$ \\
					\hline
					\cellcolor{green!50}$\vc{c}{3}$ \\
					\hline
					\cellcolor{green!50}$\vc{c}{4}$ \\
					\hline
					\cellcolor{green!50}$\vc{c}{5}$ \\
					\hline
				\end{tabular} & 
				\begin{tabular}{|l|}
					\hline
					\cellcolor{red!50}$\vc{v}{1}'$ \\
					\hline
					\cellcolor{red!50}$\vc{v}{2}'$ $\times$ \\
					\hline
					$\vc{v}{3}'$ \\
					\hline
					\cellcolor{red!50}$\vc{v}{4}'$ \\
					\hline
					$\vc{v}{5}'$ \\
					\hline
				\end{tabular} \\
				& \\
				\multicolumn{2}{c}{$(d)$} \\	
			\end{tabular}
		\end{tabular}
	}
	\caption{{\small Subfigures $(a)$ and $(b)$ are for a TraceIn$(\vc{c}{4}, V)$ query, where $V=\{\vc{v}{2}', \vc{v}{3}', \vc{v}{4}'\}$. With proofs-of-shuffle and verifiable decryption, the mix-servers can either $a)$ verifiably reveal the plaintext encrypted by $\vc{c}{4}$, say $\vc{v}{3}'$, and let the querier check if $\vc{v}{3}' \in V$, or $b)$ verifiably reveal the set of ciphertexts encrypting set $V$, say $C:=\{\vc{c}{1}, \vc{c}{2}, \vc{c}{4}\}$, and let the querier check if $\vc{c}{4} \in C$. Subfigures $(c)$ and $(d)$ are for a TraceOut$(C, \vc{v}{2}')$ query, where $C=\{\vc{c}{3}, \vc{c}{4}, \vc{c}{5}\}$. The mix-servers can either $c)$ verifiably reveal the ciphertext encrypting $\vc{v}{2}'$, say $\vc{c}{3}$, and let the querier check if $\vc{c}{3} \in C$, or $d)$ verifiably reveal the set of plaintexts that set $C$ decrypts to, say $V:=\{\vc{v}{1}', \vc{v}{2}', \vc{v}{4}'\}$, and let the querier check that $\vc{v}{2}' \in V$. With traceable mixnets, they do not need to reveal any such intermediate information.}}
	\label{fig:mixnet-vs-sm}
\end{figure}

\subsection{Our contributions} 

Our main contributions are the following. First, we introduce the notion of \emph{traceable mixnets} and formalise their completeness, soundness and secrecy requirements (Section \ref{formalism}).

Second, we propose a construction of traceable mixnets (Section \ref{construction}) using novel distributed ZKPs of \emph{set membership} \citep{Camenisch} and a novel primitive called \emph{reverse set membership}. Given a commitment scheme \citep{Pedersen} $\mathsf{comm}$ with commitment space $\Gamma$, message space $V$ and randomness space $R$, a ZKP of set membership for a commitment $\gamma \in \Gamma$ and a set $\phi$ of values from $V$ proves that $\gamma$ commits a member of $V$. Formally, we denote this as $\rho_{\mathsf{SM}}(\gamma, \phi):= \pk{v,r}{\gamma = \mathsf{comm}(v;r) \wedge v \in \phi}$, a proof of knowledge of a member $v$ of set $\phi$ and a randomness $r$ such that $\gamma$ commits $v$ with randomness $r$. A ZKP of reverse set membership for a value $v \in V$ and a set $\Phi$ of commitments from $\Gamma$ proves that $v$ \emph{is committed by} a member of $\Phi$. Formally, we denote this as $\rho_{\mathsf{RSM}}(\Phi, v):=\pk{r}{\gamma = \mathsf{comm}(v;r) \wedge \gamma \in \Phi}$, a proof of knowledge of a randomness $r$ such that some member $\gamma$ of $\Phi$ commits $v$ with randomness $r$. Prior work has mainly focused on the ZKP of set membership and the single-prover case, whereas our ZKPs (both for set membership and the novel reverse set membership) work in a distributed setting where the mixnet's mix-servers jointly act as the provers. The mix-servers need to carry out the proof without themselves learning any information about either the commitment openings or the association between commitments and  plaintext values. Our ZKPs are interactive, which is acceptable for our inherently interactive use-case.

Third, we provide detailed security analysis of our construction and formal rules for privacy risk analysis of a given set of allowed TraceIn/TraceOut queries (Section \ref{analysis} and Appendices \ref{appendix:proofs}-\ref{privacy-risk-analysis}).

Fourth, we provide a comprehensive implementation of our proposal (Section \ref{practicalities}). Our construction has linear time complexity in the size of the mixnet input list for batched queries and greatly outperforms even single-prover existing techniques. Specifically, our distributed ZKPs of set membership and reverse set membership (which enable BTraceIn and BTraceOut respectively) have per-prover proving times respectively 43x and 9x faster than single-prover zkSNARK-plus-Merkle tree based proofs. By conservative estimates, this makes them at least 86x and 18x faster than the state-of-the-art collaborative zkSNARKs \citep{collaborative-zksnarks} in the distributed setting. Our implementation is open source and available at \citep{tm-impl}.

\subsection{Related work}
\label{review}

\subsubsection{Controlled information release in statistical databases}

A common approach for controlled information release for statistical analytics is by \emph{anonymising} (releasing a noisy version of) the dataset. Another approach is \emph{differential privacy} (DP) \citep{dp-orig}: interactively providing noisy answers to analytics queries such that the answer distribution is insensitive to any given individual's data. However, anonymisation is a known poor safeguard against re-identification attacks \cite{robust-deanon,provable-deanonymisation}. Further, because of correlations in different individuals' data items, differentially private mechanisms can still reveal arbitrary partial information about the link between users' identity information and sensitive data to an adversary running arbitrary queries \cite{bayesian-dp-correlated-data}. In comparison, traceable mixnets reveal answers to only pre-approved queries and otherwise keep users' identity information and sensitive data unlinkable. Also, their distributed setting naturally fits into privacy-sensitive applications, whereas most anonymisation/DP solutions assume a trusted data curator.

\subsubsection{Existing traceability notions}
\label{traceability-notions}

Many works \citep{revocable-anonymity,selectively-traceable-anonymity,backref,amr,traceback,hecate} aim to balance anonymity with accountability by letting users communicate anonymously by default but allowing a trusted third party to revoke the anonymity of users that misbehave by sending illegal, misinformative or offensive messages. This is done by tracing the exact senders of the offending messages, in contrast to our generalised set-based and bidirectional notion of traceability.

Group/ring signatures \citep{groupsign,ringsign} let a verifier verify that a message was sent by a member of a group, without learning which member. This resembles our TraceOut query if we map the group of senders with the subset of input ciphertexts. However, these signatures require active involvement of the senders and are not logistically suitable for backend analytics applications. Likewise, anonymous credentials \citep{anoncred-camenisch} let individuals prove in ZK that they satisfy some eligibility criteria, resembling our TraceIn queries, but they also require individuals' active involvement in securing their anonymity, whereas this is done by distributed mix-servers in traceable mixnets.

\subsubsection{Existing verifiable mixnets}
\label{review-mixnets}

Haines and M{\"u}ller \citep{mixnet-sok} review and identify the following existing techniques for building verifiable mixnets: \emph{message tracing} \cite{sender-verifiable-mixnet,sElect}, \emph{verification codes} \cite{schneier-applied-crypto,sElect}, \emph{trip wires} \cite{khazaei-trip-wire,boyen-trip-wire}, \emph{message replication} \cite{khazaei-trip-wire}, \emph{randomised partial checking (RPC)} \cite{rpc,kusters-secanalysis-rpc-dmn,kusters-secanalysis-rpc-rmn,rpc-revisited} and \emph{proofs of shuffle} 
\cite{neff-shuffle,terelius-restricted-shuffles,commitment-consistent-proof-shuffle}. These techniques only verify that the mixnet output list was a decryption and permutation of its input ciphertext list and do not support the fine-grained TraceIn/TraceOut queries. \emph{Message tracing}  and \emph{verification codes} provide limited traceability by letting senders verify the processing of their \emph{own} ciphertexts, but one who does not hold the ciphertext secrets cannot perform this verification. 

\emph{Proofs of shuffle} are the state-of-the-art verifiable mixnet techniques. A proof-of-shuffle proves in ZK that two ciphertext lists are permutations and re-encryptions of each other. This, combined with verifiable decryption techniques \citep{furukawa-shuffle-verifiable-decryption}, provides a ZKP that a list of plaintexts is a decryption and permutation of a list of ciphertexts. These approaches can also prove that a \emph{sublist} of the plaintext list is a decryption and permutation of its corresponding sublist in the input ciphertext list. However, as shown in Figure \ref{fig:mixnet-vs-sm}, they leak extra information beyond TraceIn/TraceOut query outputs. 

Note that unlike non-interactive verifiable mixnets, an interactive verifiable/traceable mixnet necessarily requires the mix-servers to store their permutations. Thus, forward secrecy of query results is not maintained under future compromise of stored permutations.

\subsubsection{Set membership proofs}
\label{ss:review-sm}

Below we review set membership and reverse set membership ZKP techniques, first for a single prover.
 
\begin{paragraph}{Techniques with quadratic complexity}
	Cramer et al. \citep{cramer-or-composition} propose a generic technique to create a zero-knowledge $\Sigma$-protocol for the OR composition of two statements, given $\Sigma$-protocols for each individual statement. Both ZKPs of set membership and reverse set membership can be constructed using this technique: $\rho_{\mathsf{SM}}(\gamma,\phi):=\pk{r}{\bigvee_{v \in \phi} \gamma = \mathsf{comm}(v;r)}$ and $\rho_{\mathsf{RSM}}(\Phi, v):=\pk{r}{\bigvee_{\gamma \in \Phi}  \gamma = \mathsf{comm}(v;r)}$. However, for proving $\rho_{\mathsf{SM}}(\gamma,\phi)$ for multiple $\gamma$s against the same set $\phi$ or $\rho_{\mathsf{RSM}}(\Phi, v)$ for multiple $v$s against the same set $\Phi$ (the ``batched'' queries), it results in an overall $O(n^2)$ complexity, where $n \approx |\Phi|,|\phi|$. Groth and Kohlweiss \citep{groth-kohlweiss} propose a ZKP of knowledge of the form $\rho_{\mathsf{RSM}\text{-}0}:=\pk{r}{\bigvee_{\gamma \in \Phi}  \gamma = \mathsf{comm}(0;r)}$, i.e., a proof that one of the commitments in a set commits to $0$, with $O(\log(n))$ communication complexity. Interestingly, $\rho_{\mathsf{RSM}\text{-}0}$ can be used to prove both $\rho_{\mathsf{SM}}$ and $\rho_{\mathsf{RSM}}$ \citep{groth-kohlweiss}, resulting in an $O(n\log(n))$ communication complexity for the batched queries. However, the \emph{computational} complexity (for both the prover and the verifier) remains $O(n^2)$.
\end{paragraph}

\begin{paragraph}{Accumulator-based techniques}
	Cryptographic accumulators \citep{merkle-accumulator,benaloh-demare-rsa-accumulator,nguyen-accumulator} enable efficient ZKPs of set membership. An accumulator scheme $(\mathsf{Acc}, \mathsf{GenWitness}, \mathsf{AccVer})$ allows computing a short digest $A_{\phi}$ to a large set $\phi$ as $A_{\phi} \leftarrow \mathsf{Acc}(\phi)$ and a short membership witness $w_v$ for a member $v \in \phi$ as $w_v \leftarrow \mathsf{GenWitness}(A_{\phi}, v)$ such that $\mathsf{AccVer}(A_{\phi}, v, w_{v})=1$ is a proof that $v \in \phi$. A ZKP of set membership can thus be constructed by requiring both prover and verifier to compute $A_{\phi}$, the prover to compute $w_v$ and both to then engage in $\rho_{\mathsf{SM}\text{-}\mathsf{Acc}}(\gamma,A_{\phi}):=\pk{v,r,w_v}{\gamma = \mathsf{comm}(v;r) \wedge \mathsf{AccVer}(A_{\phi}, v, w_v)=1}$. If the accumulator scheme allows commitments to be set members, a ZKP of reverse set membership can be similarly constructed using $\rho_{\mathsf{RSM}\text{-}\mathsf{Acc}}(v,A_{\Phi}):=\pk{r,w}{\mathsf{AccVer}(A_{\Phi}, \mathsf{comm}(v;r), w)=1}$, where $w$ denotes the membership witness of the commitment $\mathsf{comm}(v;r) \in \Phi$. 

	A popular approach for accumulator-based ZKPs of set membership involves Merkle accumulators \citep{merkle-accumulator} as the accumulator scheme and zkSNARKs \citep{groth16,plonk,marlin} as the ZK proof system. With Merkle accumulators, both $\mathsf{GenWitness}$ and $ \mathsf{AccVer}$ take $O(\log n)$ time, where $n=|\phi|$, which allows $n$ set membership and reverse set membership proofs in $O(n \log (n))$ time. This approach can also generically support $\rho_{\mathsf{RSM}\text{-}\mathsf{Acc}}$ by computing Merkle accumulators for the set of commitments. However, it involves expensive hash computations inside the zkSNARK circuit. Benarroch et al. \citep{modular-zkp-sm} present a non-generic ZKP of set membership using RSA accumulators, avoiding these expensive hash computations. However, the technique does not support reverse set membership. Also, computing witness $w_v$ for RSA accumulators takes $O(n)$ time, which makes it $O(n^2)$ for batched queries. Techniques to efficiently batch multiple membership proofs together also exist \citep{boneh-batching,campanelli-batching}. 	
\end{paragraph}

\begin{paragraph}{Extending to the distributed setting}
	All the above techniques work on the single-prover case. Extending them to our distributed setting where none of the provers (the mix-servers) know either the commitment openings or the permutation between the commitments and the plaintexts is non-trivial. The batching techniques \citep{boneh-batching,campanelli-batching} also fail in the distributed setting since they require the prover to know upfront which entries pass the membership proof. 

	Collaborative zkSNARKs \citep{collaborative-zksnarks} allow a distributed set of provers holding secret shares of a SNARK witness to prove joint knowledge of the same. They add roughly 2x overhead in per-prover proving time over the standard zkSNARKs \citep{collaborative-zksnarks}. DPZKs \citep{dpzk} provide similar guarantees. However, even to securely obtain shares of the SNARK witness, an extra MPC protocol is likely required.	
\end{paragraph}

\begin{paragraph}{Signature-based set membership}
	Camenisch et al. \citep{Camenisch} initiated a different approach to ZKPs of set membership: the verifier provides to the prover signatures on all members of the set under a fresh signing key generated by the verifier, and the prover proves knowledge of a signature on the committed value in ZK. This proves set membership because if the commitment commits a value outside the set, the prover does not obtain a signature on it and must forge it. This approach gives $O(n)$ batched query complexity because verifier signatures can be reused. We extend these ZKPs by a signature-based ZKP of \emph{reverse set membership} and by distributed signature-based ZKPs for both set membership and reverse set membership.
\end{paragraph}

\section{Formal definitions}
\label{formalism}

We now formalise traceable mixnets. We directly present the batched BTraceIn/BTraceOut protocols as that is our main focus (the TraceIn/TraceOut protocols are trivial special cases). Also, for simplicity, we present the special case when all $t=m$ mix-servers are required to decrypt the ciphertexts but any set of less than $t$ mix-servers cannot break secrecy. Extension to a general $t$ is possible with standard threshold cryptography techniques \citep{threshold-cryptography}. Finally, we assume an authenticated broadcast channel that ensures authenticity, availability and non-repudiability of all published messages.

\begin{paragraph}{Notation}
Given a positive integer $n$, we denote the set $\{1, \dots, n\}$ by $[n]$. We let (boldface) $\ve{x} \in X^n$ denote an $n$-length vector of values drawn from a set $X$, $\vc{x}{i}$ denote the $i^{\text{th}}$ component of $\ve{x}$, and, given an index set $I \subseteq [n]$, $\vc{x}{I}$ denote the set $\{\vc{x}{i} \mid i \in I \}$. For any scalar binary operation $\odot: X \times Y \rightarrow Z$ and vectors $\ve{x} \in X^n$, $\ve{y} \in Y^n$, we let $\ve{x} \odot \ve{y}$ denote the vector $(\vc{x}{i} \odot \vc{y}{i})_{i \in [n]}$, i.e., the vector obtained by component-wise application of $\odot$. We let $\ve{u}^{\ve{v}}$ denote $(\vc{u}{i}^{\vc{v}{i}})_{i \in [n]}$, $\ve{f}(\ve{v})$ denote $(f(\vc{v}{i}))_{i \in [n]}$, $\ve{f}(x, \ve{v})$ denote $(f(x,\vc{v}{i}))_{i \in [n]}$, etc. 

We denote a multiparty computation protocol $P$ between parties $\mathcal{P}_1,\dots,\mathcal{P}_m$ where the common input of each party is $ci$, $\mathcal{P}_k$'s secret input is $si_k$, the common output is $co$ and $\mathcal{P}_k$'s secret output is $so_k$ as follows: $co, (\so{\mathcal{P}_k}{so_k})_{k\in[m]} \leftarrow P(ci, (\si{\mathcal{P}_k}{si_k})_{k \in [m]})$. When a party does not have a secret output, we drop it from the left-hand side. In security experiments where the experimenter plays the role of honest parties $(\mathcal{P}_k)_{k \in H}$ for some $H \subset [m]$ and an adversary $\adv$ plays the role of $(\mathcal{P}_{k})_{k \in [m] \setminus H}$, we indicate it as $co, (\so{\mathcal{P}_k}{so_k})_{k \in H} \leftarrow P(ci, (\so{\mathcal{P}_k}{si_k})_{k \in H}, ({\mathcal{P}_k}^{\adv})_{k \in [m] \setminus H})$. We let $\share{x}{k}$ denote a component (typically, secret-share) of $x$ designated for $\mathcal{P}_k$. We call a function $f$ \emph{negligible} if for any polynomial $p$, there exists an $N \in \Nat$ such that $f(x) < 1/p(x)$ for all $x > N$. 
\end{paragraph}

\begin{definition}[Traceable mixnets]
\label{def:traceable-mixnets}
A \emph{traceable mixnet} is a tuple of protocols/algorithms $(\keygen,$ $\enc,$ $\mix,$ $\btrin,$ $\btrout)$ between $n$ senders $(S_i)_{i \in [n]}$, $m$ mix-servers $(\mixer{k})_{k \in [m]}$ and a querier or verifier $\querier$ such that:
\begin{itemize}[leftmargin=*]
    \item[-] $\mpk, (\so{\mixer{k}}{\msk{k}})_{k \in [m]} \leftarrow \keygen(1^{\lambda}, (\so{\mixer{k}}{})_{k \in [m]})$ is a key generation protocol between $(\mixer{k})_{k \in [m]}$, where $\lambda$ is a security parameter (given in unary), individual mix-servers do not have any secret input, the common output is a mixnet public key $\mpk$ and each $\mixer{k}$'s secret output is a secret key $\msk{k}$.

    \item[-] $\vc{c}{i} \leftarrow \enc(\mpk, \vc{v}{i})$ is an algorithm run by sender $S_i$, where $\mpk$ is the mixnet public key and $\vc{v}{i}$ is $S_i$'s sensitive input drawn from some plaintext space $\Val$. Output $\vc{c}{i}$ is a ciphertext ``encrypting'' $\vc{v}{i}$.

    \item[-] $\ve{v}', (\so{\mixer{k}}{\wit{k}})_{k \in [m]}$ $\leftarrow$ $\mix(\mpk, \ve{c},$ $(\si{\mixer{k}}{\msk{k}})_{k \in [m]})$ is a mixing protocol between $(\mixer{k})_{k \in [m]}$, where $\mpk$ is the mixnet public key and $\ve{c} \leftarrow (\enc(\mpk, \vc{v}{i}))_{i \in [n]}$ is a vector of ciphertexts encrypting the senders' plaintexts $(\vc{v}{i})_{i \in [n]}$. Each $\mixer{k}$'s secret input is its secret key $\msk{k}$. The common output is a vector $\ve{v}'$ of plaintext values obtained after permuting and decrypting $\ve{c}$ (thus $\ve{v}' = (\vc{v}{\pi(i)})_{i \in [n]}$ for some permutation $\pi$). Each $\mixer{k}$'s secret output is a witness $\wit{k}$ to be used in proving correctness of the $\btrin/\btrout$ outputs (see below).

    \item[-] $\so{\querier}{\vc{c}{I^*}} \leftarrow \btrin(\mpk,$ $\ve{c},$ $\ve{v}',$ $I,$ $J,$ $(\si{\mixer{k}}{\msk{k},$ $\wit{k}})_{k \in [m]},$ $\si{\querier}{})$ is a protocol between $(\mixer{k})_{k \in [m]}$ and querier $\querier$, where $\ve{v}', \so{\mixer{k}}{\wit{k}}$ $\leftarrow$ $\mix(\mpk,$ $\ve{c},$ $\si{\mixer{k}}{\msk{k}})$, $\ve{c}$ $\leftarrow$ $(\enc(\mpk,$ $\vc{v}{i}))_{i \in [n]}$ and $I,J \subseteq [n]$. $\querier$'s secret output is a set of ciphertexts $\vc{c}{I^*} = \{\vc{c}{i} \in \vc{c}{I} \mid \vc{v}{i} \in \vc{v}{J}' \}$; $\querier$ may abort if it is not convinced about the correctness of $\vc{c}{I^*}$. 

    \item[-] $\so{\querier}{\vc{v}{J^*}'} \leftarrow \btrout(\mpk,$ $\ve{c},$ $\ve{v}',$ $I,$ $J,$ $(\si{\mixer{k}}{\msk{k},$ $\wit{k}})_{k \in [m]},$ $\si{\querier}{})$ is a protocol between $(\mixer{k})_{k \in [m]}$ and $\querier$, where all inputs are exactly the same as $\btrin$ and $\querier$'s secret output is a set of plaintexts $\vc{v}{J^*}' = \{\vc{v}{j}' \in \vc{v}{J}' \mid \vc{v}{j}' \in \vc{v}{I} \}$; $\querier$ may abort if it is not convinced about the correctness of $\vc{v}{J^*}'$.
\end{itemize}
\end{definition}

\subsection{Completeness}

The completeness definition for traceable mixnets (Definition \ref{def:compl} and Figure \ref{exp-completeness}) models that when all the parties are honest then $a)$ $\querier$'s output $\vc{c}{I^*}$ on a $\btrin$ query on $(\ve{c},\ve{v}', I, J)$ is exactly the set of ciphertexts in $\vc{c}{I}$ that encrypted some plaintext in $\vc{v}{J}'$ (line 7), and $b)$ its output $\vc{v}{J^*}'$ on a $\btrout$ query on $(\ve{c},\ve{v}', I, J)$ is exactly the set of plaintexts in $\vc{v}{J}'$ that were encrypted by some ciphertext in $\vc{c}{I}$ (line 8). Note that we only consider the case of distinct values. The case of repeated values is trivially reducible to this case if $S_i$ prefixes $\vc{v}{i}$ with a nonce drawn uniformly from a large set.

\begin{definition}[Completeness]
	\label{def:compl}
	A traceable mixnet is \emph{complete} if for each security parameter $\lambda \in \Nat$, number of ciphertexts $n \in \Nat$, vector $\ve{v} \in \Val^n$ of distinct plaintext values and index sets $I,J \subseteq [n]$, there exists a negligible function $\negl$ such that 
	\begin{equation}
		\label{eq:completeness}
		\Pr[\cmpl(1^{\lambda}, n, \ve{v}, I, J) = 1] \geq 1-\negl(\lambda),
	\end{equation} 
	where $\cmpl$ is as defined in Figure \ref{exp-completeness}.
\end{definition}

\begin{figure}
	\scalebox{0.8}{
		\begin{tabular}{ @{}rl@{} } 
			& $\underline{\cmpl(1^{\lambda}, n, \ve{v}, I, J)}:$ \\
			1 & $ \mpk, (\so{\mixer{k}}{\msk{k}})_{k \in [m]}) \leftarrow \keygen(1^{\lambda}, (\si{\mixer{k}}{})_{k \in [m]})$ \\
			2 & $\ve{c}:=(\vc{c}{i})_{i \in [n]} \leftarrow (\enc(\mpk, \vc{v}{i}))_{i \in [n]}$ \quad \mycomment{$\vc{c}{i}$ encrypts $\vc{v}{i}$} \\
			3 & $ \ve{v}', (\so{\mixer{k}}{\wit{k}})_{k \in [m]} \leftarrow \mix(\mpk, \ve{c}, (\si{\mixer{k}}{\msk{k}})_{k \in [m]})$ \\
			4 & $ \so{\querier}{\vc{c}{I^*}} \leftarrow \btrin(\mpk, \ve{c}, \ve{v}', I, J, (\si{\mixer{k}}{\msk{k}, \wit{k}})_{k \in [m]}, \si{\querier}{}) $ \\
			5 & $ \so{\querier}{\vc{v}{J^*}'} \leftarrow \btrout(\mpk, \ve{c}, \ve{v}', I, J, (\si{\mixer{k}}{\msk{k}, \wit{k}})_{k \in [m]}, \si{\querier}{}) $ \\
			6 & \textbf{return} 1 if all $\vc{c}{i}$s and all $\vc{v}{j}'$s are distinct, and \\
			7 & \quad 1) $\vc{c}{I^*} = \{ \vc{c}{i} \in \vc{c}{I} \mid \vc{v}{i} \in \vc{v}{J}' \}$; and \\
			8 & \quad 2) $\vc{v}{J^*}' = \{ \vc{v}{j}' \in \vc{v}{J}' \mid \vc{v}{j}' \in \vc{v}{I} \}$ \quad \mycomment{$\vc{c}{I} = \{\enc(pk, v) \mid v \in \vc{v}{I}\}$ }
		\end{tabular}
	}
	\caption{Completeness experiment}
	\label{exp-completeness}
\end{figure}

\subsection{Soundness}

The soundness definition (Definition \ref{def:snd}) models that as long as input ciphertexts are well-formed, even if all the mix-servers are dishonest, sets $\vc{c}{I^*}$ and $\vc{v}{J^*}'$ output by $\querier$ in $\btrin$ and $\btrout$ respectively must be ``correct,'' where the correctness of $\vc{c}{I^*}$ and $\vc{v}{J^*}'$ is exactly as defined in $\cmpl$. Thus, we do not allow the cheating mix-servers to force $\querier$ to produce incorrect output ($\querier$ may abort, though). We only consider well-formed input ciphertexts because ``correct'' processing of ill-formed inputs is undefined. Nevertheless, proofs of well-formedness of inputs are generally required for application-level correctness and these can be constructed by the senders at the time of uploading their ciphertexts. 

The experiment (Figure \ref{exp-soundness}) begins with the key generation protocol where the adversary $\adv$ controlling all the mix-servers $(\mixer{k})_{k \in [m]}$ provides the mixnet public key $\mpk$ (line 1). We let $\adv$ supply the plaintexts but create ciphertexts from them honestly (lines 2-3), to model that $\adv$ can supply plaintexts of its choice to cheat but the ciphertexts must be well-formed. We then allow $\adv$ to run the $\mix$ protocol and produce the output list $\ve{v}'$ (line 4). $\adv$ then outputs index sets $I$ and $J$ on which it wants to break the subsequent $\btrin/\btrout$ queries (line 3). During these queries, $\querier$ outputs $\vc{c}{I^*}$ and $\vc{v}{J^*}'$ respectively (lines 6-7). $\adv$ wins if it supplies distinct entries and $\querier$ produces valid outputs $\vc{c}{I^*}, \vc{v}{J^*}'$ (does not abort) but at least one of them is incorrect (lines 8-10).

\begin{definition}[Soundness]
	\label{def:snd}
	A traceable mixnet is \emph{sound} if for each PPT adversary $\adv$ and security parameter $\lambda \in \Nat$, there exists a negligible function $\negl$ such that 
	\begin{equation}
		\label{eq:soundness}
		\Pr[\snd^{\adv}(1^{\lambda}) = 1] \leq \negl(\lambda),
	\end{equation} 
	where $\snd^{\adv}$ is as defined in Figure \ref{exp-soundness}.
\end{definition}

\begin{figure}
	\scalebox{0.8}{
		\begin{tabular}{ @{}rl@{} } 
			& $\underline{\snd^{\adv}(1^{\lambda})}:$ \\
			1 & $ \mpk \leftarrow \keygen(1^{\lambda}, (\mixer{k}^{\adv})_{k \in [m]})$ \\
			2 & $\ve{v}:=(\vc{v}{i})_{i \in [n]} \leftarrow \adv()$ \\
			3 & $\ve{c}:=(\vc{c}{i})_{i \in [n]} \leftarrow (\enc(\mpk, \vc{v}{i}))_{i \in [n]}$ \\
			4 & $ \ve{v}' \leftarrow \mix(\mpk, \ve{c}, (\mixer{k}^{\adv})_{k \in [m]})$ \\
			5 & $I, J \leftarrow \adv(\ve{v}')$ \\
			6 & $\so{\querier}{\vc{c}{I^*}} \leftarrow \btrin(\mpk, \ve{c}, \ve{v}', I, J, (\mixer{k}^{\adv})_{k \in [m]}, \si{\querier}{})$ \\
			7 & $\so{\querier}{\vc{v}{J^*}'} \leftarrow \btrout(\mpk, \ve{c}, \ve{v}', I, J, (\mixer{k}^{\adv})_{k \in [m]}, \si{\querier}{})$ \\
			8 & \textbf{return} 1 if all $\vc{v}{i}$s and all $\vc{v}{j}'$s are distinct, and either \\
			9 & \quad 1) $\vc{c}{I^*} \neq \{\vc{c}{i} \in \vc{c}{I} \mid \vc{v}{i} \in \vc{v}{J}' \}$; or \\
			10 & \quad 2) $\vc{v}{J^*}' \neq \{\vc{v}{j}' \in \vc{v}{J}' \mid \vc{v}{j}' \in \vc{v}{I} \}$ \quad \mycomment{$\vc{c}{I} = \{\enc(pk, v) \mid v \in \vc{v}{I}\}$ }
		\end{tabular}
	}
	\caption{Soundness experiment}
	\label{exp-soundness}
\end{figure}

\subsection{Secrecy}

The secrecy definition (Definition \ref{def:secr}) extends a standard anonymity property \citep{benaloh-swap-privacy1,benaloh-swap-privacy2,delaune-swap-privacy,bpriv} to the case when the $\btrin/\btrout$ queries are also allowed. The standard anonymity property can be stated for our setting as follows: an adversary controlling all-but-two senders, the querier and any set of less than $m$ mix-servers should not be able to distinguish between a world where ciphertexts $(c_0, c_1)$ sent by the two honest senders encrypt values $(v_0, v_1)$ (world 0) and the world where they encrypt $(v_1, v_0)$ (world 1). When $\btrin/\btrout$ queries are allowed, distinguishing between the two worlds is trivial because of the query outputs (e.g., if $I,J$ given to a $\btrin$ query are such that $c_0,c_1 \in \vc{c}{I}$ and $v_0 \in \vc{v}{J}'$ but $v_1 \not\in \vc{v}{J}'$ then $\adv$ immediately knows it is in world $0$ if output $\vc{c}{I^*}$ includes $c_0$). Thus, we require that $a)$ in all the $\btrin$ queries either both $v_0,v_1 \in \vc{v}{J}'$ or both $v_0,v_1 \not\in \vc{v}{J}'$ and $b)$ in all the $\btrout$ queries either both $c_0,c_1 \in \vc{c}{I}$ or both $c_0,c_1 \not\in \vc{c}{I}$. 

In more detail, in this experiment (Figure \ref{exp-secrecy}), adversary $\adv$ engages in the key generation protocol where it controls all the mix-servers except one, i.e., $\mixer{k^*}$ (line 1). It then supplies input ciphertexts for all the senders except the two that it does not control, say $S_{i_0}$ and $S_{i_1}$. For these senders, it supplies the values $v_0,v_1$ (line 2). In world 0 ($b=0$), $S_{i_0}$'s ciphertext $\vc{c}{i_0}$ encrypts $v_0$ and $S_{i_1}$'s ciphertext $\vc{c}{i_1}$ encrypts $v_1$; in world 1 ($b=1$), this order is reversed (lines 3-4). The ciphertext list thus formed is processed through the $\mix$ protocol, where $\adv$ controls all mix-servers except $\mixer{k^*}$ and produces an output plaintext list $\ve{v}'$ (line 5). Then, $\adv$ obtains access to oracles $\otrin, \otrout$ that let it choose $I$, $J$, control $\querier$ and all mix-servers except $\mixer{k^*}$ and interact with $\mixer{k^*}$ in the $\btrin, \btrout$ protocols (lines 6-14). $\adv$ is required to respect the condition of including either both or none of the honest senders' ciphertexts/plaintexts in its oracle calls (lines 12 and 18). Finally, $\adv$ outputs a bit $b'$ as its guess of the bit $b$ (line 5) and wins if its advantage in making the correct guess is non-negligible.

\begin{definition}[Secrecy]
	\label{def:secr}
	A traceable mixnet protects \emph{secrecy} if for each PPT adversary $\adv$, security parameter $\lambda \in \Nat$, $k^* \in [m]$, and $i_0, i_1 \in [n]$, there exists a negligible function $\negl$ such that 
	\begin{align}
		\label{eq:secrecy}
		\begin{split}
			&|\Pr[\secr^{\adv}(1^{\lambda}, k^*, i_0, i_1, 0) = 1] - \\
			&\quad \Pr[\secr^{\adv}(1^{\lambda}, k^*, i_0, i_1, 1) = 1]| \leq \negl(\lambda),
		\end{split}
	\end{align} 
	where $\secr^{\adv}$ is as defined in Figure \ref{exp-secrecy}.
\end{definition}

Definition \ref{def:secr-mixers} below models that when $\adv$ does not control $\querier$, it should not even learn the query outputs.

\begin{definition}[Output secrecy]
	\label{def:secr-mixers}
	A traceable mixnet protects \emph{output secrecy} if it protects secrecy as per Definition \ref{def:secr} except that in experiment $\secr^{\adv}$ (Figure \ref{exp-secrecy}), $\adv$ does not control $\querier$ during the $\btrin/\btrout$ calls (lines 10 and 14) and the constraints on lines 9 and 13 are removed.
\end{definition}

Note that Definition \ref{def:secr} also models that if the honest mix-server follows a query policy then the adversary gains no information about the output of queries disallowed by the policy, since this case corresponds to an adversary that simply does not call the $\otrin$ or $\otrout$ oracles for the disallowed queries. The privacy risk associated with the outputs of allowed queries is outside the scope of formal security requirements of traceable mixnets, but we provide a mechanism to analyse this risk in Appendix \ref{privacy-risk-analysis}.

\begin{figure}
	\scalebox{0.8}{
		\begin{tabular}{ @{}rl@{} } 
			& $\underline{\secr^{\adv}(1^{\lambda}, k^*, i_0, i_1, b)}:$ \\
			1 & $\mpk, \so{\mixer{k^*}}{\msk{k^*}} \leftarrow \keygen(1^{\lambda}, \si{\mixer{k^*}}{}, (\mixer{k}^{\adv})_{k \neq k^* })$ \\
			2 & $v_0,v_1, (\vc{c}{i})_{i \in [n] \setminus \{i_0,i_1\}} \leftarrow \adv(\mpk)$ \\
			3 & \textbf{if $b=0$:} $\vc{c}{i_0} \leftarrow \enc(\mpk, v_0)$; $\vc{c}{i_1} \leftarrow \enc(\mpk, v_1)$ \\
			4 & \textbf{else:} $\vc{c}{i_0} \leftarrow \enc(\mpk, v_1)$; $\vc{c}{i_1} \leftarrow \enc(\mpk, v_0)$ \\
			5 & $\ve{v}', \so{\mixer{k^*}}{\wit{k^*}} \leftarrow \mix(\mpk, (\vc{c}{i})_{i \in [n]}, \si{\mixer{k^*}}{\msk{k^*}}, (\mixer{k}^{\adv})_{k \neq k^* })$ \\
			6 & \textbf{return} $b' \leftarrow \adv^{\otrin, \otrout}(\ve{v}')$ \\
			7 & ~ \\
			
			8 & $\underline{\otrin(I, J):}$ \\
			9 & \quad \textbf{assert $(v_0 \in \vc{v}{J}' \iff v_1 \in \vc{v}{J}')$} \\
			10 & \quad $\btrin(\mpk, \ve{c}, \ve{v}', I, J, \si{\mixer{k^*}}{\msk{k^*}, \wit{k^*}}, (\mixer{k}^{\adv})_{k \neq k^* }, \querier^{\adv})$ \\
			11 & ~ \\
			
			12 & $\underline{\otrout(I, J):}$ \\
			13 & \quad \textbf{assert $(\vc{c}{i_0} \in \vc{c}{I} \iff \vc{c}{i_1} \in \vc{c}{I})$} \\
			14 & \quad $\btrout(\mpk, \ve{c}, \ve{v}', I, J, \si{\mixer{k^*}}{\msk{k^*}, \wit{k^*}}, (\mixer{k}^{\adv})_{k \neq k^* }, \querier^{\adv})$ \\
		\end{tabular}
	}
	\caption{Secrecy experiment}
	\label{exp-secrecy}
\end{figure}

\section{Preliminaries}
\label{preliminaries}

\begin{paragraph}{Setup}
	We assume that the output of the following $\setup$ algorithm is implicitly available to all the parties: $(q,$ $\Gone,$ $\Gtwo,$ $\GT,$ $e,$ $f_1,$ $g_1,$ $h_1,$ $f_2,$ $g_2,$ $f_T)$ $\leftarrow$ $\setup(1^{\lambda}, m, n)$. $\setup$ takes as input a security parameter $\lambda \in \mathbb{N}$, integers $m$ and $n$ ($m$ represents the number of mix-servers and $n$ represents the number of input ciphertexts) and outputs the following setup parameters: a large prime number $q$ ($q \gg m,n$), cyclic groups $\Gone, \Gtwo, \GT$ of order $q$, generators $\{f_1, g_1, h_1\}$, $\{f_2, g_2\}$ and $f_T$ of groups $\Gone$, $\Gtwo$ and $\GT$ respectively, and an efficiently computable bilinear map $e: \Gone \times \Gtwo \rightarrow \GT$ \footnote{For all $a,b \in \Zq$ and generators $g_1,g_2$ of $\Gone$ and $\Gtwo$ respectively, $e(g_1^a, g_2^b)=e(g_1,g_2)^{ab}$ and $e(g_1,g_2)\neq 1_{\GT}$, where $1_{\GT}$ denotes the identity element of $\GT$.}. We assume that the $n$-Strong Diffie Hellman (SDH) assumption \citep{boneh-boyen} holds in groups $(\Gone, \Gtwo)$ and that the decisional Diffie-Hellman (DDH) and discrete logarithm (DL) problems are hard in $\Gone$. We assume that all generators are randomly generated, e.g., as the output of a hash function modelled as a random oracle.
\end{paragraph}

\subsection{Key cryptographic primitives}

\subsubsection{ (Basic) Boneh-Boyen (BB) signatures (Section 3.1; \cite{boneh-boyen})}
\label{sss:bb-sig}

In this signature scheme, the signer chooses its secret key (SK) as $x \xleftarrow{\$} \Zq$ and verification key (VK) as $y \leftarrow g_2^x$. To sign a message $m \in \Zq$, it computes $\sigma \leftarrow g_1^{\frac{1}{m+x}}$. The signature is verified if $e(\sigma,y g_2^m) \stackrel{?}= e(g_1,g_2)$. This scheme is unforgeable against weak chosen message attacks under the $n$-SDH assumption \cite{boneh-boyen}.

\subsubsection{BBS+ signatures \cite{bbsplus-sig}}
\label{sss:bbsplus-sig}

In this signature scheme, the signer chooses its SK as $x \xleftarrow{\$} \Zq^*$ and VK as $y \leftarrow f_2^x$. To sign a message $m \in \Zq$, it computes $c,r \xleftarrow{\$} \Zq$ and $S \leftarrow (f_1g_1^m h_1^r)^{\frac{1}{c+x}}$ and outputs $\sigma := (S, c, r)$. The signature is verified if $e(S,y f_2^c) \stackrel{?}= e(f_1g_1^m h_1^r,f_2)$. The scheme is unforgeable against adaptively chosen message attacks under the $n$-SDH assumption \cite{bbsplus-sig}.

\subsubsection{Signatures on committed values} 
\label{sss:sig-committed-value}

BBS+ signatures also let one reveal only a Pedersen commitment \citep{Pedersen} $\gamma=g_1^vh_1^r$ to a signer (and a PoK of $v,r$) and obtain a signature on the \emph{committed value} $v$: 
\begin{itemize}[leftmargin=*]
    \item[-] The requester sends $\gamma$ and $\rho_{\gamma} \leftarrow \pk{v,r}{\gamma=g_1^vh_1^r}$ to the signer. The signer verifies $\rho_{\gamma}$.
    \item[-] The signer (with SK $x$ and VK $y= f_2^x$) computes a \emph{quasi BBS+ signature} by choosing $c,\hat{r} \xleftarrow{\$} \Zq$ and computing $S \leftarrow (f_1h_1^{\hat{r}}\gamma)^{\frac{1}{c+ x }}$. It sends $\hat{\sigma}:= (S, c, \hat{r})$ to the requester.
    \item[-] The requester computes $\sigma \leftarrow (S, c, \hat{r} + r)$. ($\sigma$ $=$ $((f_1h_1^{\hat{r}+r}g_1^v)^{\frac{1}{c+x}},$ $c,$ $\hat{r} + r)$ is a valid BBS+ signature on message $v$ under VK $y$.) 
\end{itemize}
$\hat{\sigma}:=(S,c,\hat{r})$ can be verified by checking if $e(S,yf_2^c) \stackrel{?}= e(f_1h_1^{\hat{r}}\gamma, f_2)$.

\subsubsection{$(m,m)$-threshold secret sharing}
\label{ss:secret-sharing}

We consider the following standard $(m,m)$-threshold secret sharing scheme where $\mathsf{Share}$ allows sharing a secret $x$ among $m$ parties and $\mathsf{Recons}$ allows its reconstruction by all of them (fewer parties do not learn $x$):
\begin{itemize}[leftmargin=*]
	\item[-] $(\share{x}{k})_{k \in [m]} \leftarrow \sharealg_{m,m}(x \in \Zq):$ $(\share{x}{k})_{k \in [m-1]}$ $\xleftarrow{\$}$ $\Zq^{m-1};$ $\share{x}{m}$ $\leftarrow$ $(x - \sum_{k \in [m-1]} \share{x}{k}) \pmod q$
	\item[-] $x \leftarrow \recon((\share{x}{k})_{k \in [m]}): x \leftarrow \sum\limits_{k \in [m]} \share{x}{k} \pmod q$
\end{itemize} 
A secret sharing scheme is called \emph{additive} (resp. \emph{multiplicative}) if $\mathcal{P}_k$ on input its shares $\share{x}{k}$, $\share{y}{k}$ of secrets $x$ and $y$ respectively can obtain its share of $x+y$ (resp. $xy$) without any additional interaction with other parties. The above scheme is clearly additive. It can also be made multiplicative using Beaver's trick  \citep{mult-secret-sharing}, which employs an input-independent precomputation step and an algorithm $\mult$ such that $\mathcal{P}_k$ can obtain its share of $xy$ as $\mult(\share{x}{k}, \share{y}{k})$.

\subsubsection{$(m,m)$-threshold proofs of knowledge}
\label{sss:distr-pok}

An $(m,m)$-threshold proof of knowledge (also called a \emph{distributed proof of knowledge} or a DPK) is a protocol between provers $(\mathcal{P}_k)_{k \in [m]}$ and a verifier $\verifier$ that convinces $\verifier$ that for a given common input $x$, the provers know secret shares $\share{\omega}{k}$ of a secret $\omega$ such that a predicate $p(x,\omega)$ is true. We denote these DPKs as $\so{\verifier}{res}$ $\leftarrow$ $\mathsf{DPK}(x$, $p$, $(\si{\mathcal{P}_k}{\share{\omega}{k}})_{k \in [m]}$, $\si{\verifier}{})$, where $res=1$ means that $\verifier$ accepted the proof. The secrecy guarantee is that an adversary $\adv$ controlling $\verifier$ and all $(\mathcal{P}_{k})_{k \neq k^*}$ for some $k^*$ cannot learn anything about $\share{\omega}{k^*}$ (and thus $\omega$) \citep{distr-pok}.

We use DPKs where the predicate $p$ is of the form $\bigwedge_{i \in [\ell]} y_i = \prod_{j \in [\ell']} g_{ij}^{\omega_j}$ for $\ell,\ell' \in \Nat$, public values $y_i,g_{ij}\in \Gone$, $\Gtwo$ or $\GT$ and $\omega_j \in \Zq$. These DPKs can be constructed using standard $\Sigma$-protocol techniques \citep{desmedt-threshold-cryptosystems,distr-pok,camenisch-thesis} if each prover $\mathcal{P}_k$ knows share $\share{\omega}{k}_j$ of each $\omega_j$. We use the following NIZK variant obtained using the Fiat-Shamir heuristic \citep{fiat-shamir}:
\begin{itemize}[leftmargin=*]
	\item[-] $(\mathcal{P}_k)_{k \in [m]}$: Publish $\share{a}{k}_i \leftarrow \prod_{j \in [\ell']} g_{ij}^{\share{r}{k}_j}$, where $\share{r}{k}_j \xleftarrow{\$} \Zq$.
	\item[-] $(\mathcal{P}_k)_{k \in [m]}$: Compute $a_i \leftarrow \prod_{k \in [m]} \share{a}{k}_i$; $c \leftarrow H(p \Vert (a_i)_{i \in [\ell]})$, where $H$ is a cryptographic hash function modelled as a random oracle; $\share{z}{k}_j \leftarrow \share{r}{k}_j - c \share{\omega}{k}_j \pmod{q}$. Send $\share{z}{k}_j$ to $\verifier$.
	\item[-] $\verifier$: Obtain $((\share{a}{k}_i)_{i \in [\ell],k\in[m]},$ $c,$ $(\share{z}{k}_j)_{j \in [\ell'],k \in [m]})$ from $(\mathcal{P}_k)_{k \in [m]}$. Compute $a_i \leftarrow \prod_{k \in [m]} \share{a}{k}_i$; $z_j \leftarrow \sum_{k \in [m]}\share{z}{k}_j \pmod{q}$. Check $c \stackrel{?}= H(p \Vert (a_i)_{i \in [\ell]})$ and $\bigwedge_{i \in [\ell]} a_i \stackrel{?}= y_i^c\prod_{j \in [\ell']} g_{ij}^{z_j}$.
\end{itemize} 
In this variant, if $\adv$ controls $(\mathcal{P}_{k})_{k \neq k^*}$ but not $\verifier$, it does not even learn whether the statement was proved successfully or not, because it only sees $\share{a}{k^*}_i \leftarrow \prod_{j\in [\ell']} g_{ij}^{\share{r}{k^*}_j}$ and not $\share{z}{k^*}_j$.

\subsubsection{$(m,m)$-threshold homomorphic encryption}
\label{sss:thenc}

An $(m,m)$-threshold encryption scheme $\mathsf{E}^{\mathsf{th}}$ between parties $(\mathcal{P}_k)_{k \in [m]}$ with plaintext space $\mathbb{M}(\mathsf{E}^{\mathsf{th}})$ and ciphertext space $\mathbb{C}(\mathsf{E}^{\mathsf{th}})$ is a tuple $(\keygen,$ $\enc,$ $\tdec)$, where $\keygen$ is a key generation protocol, $\enc$ is an encryption algorithm and $\tdec$ is a threshold decryption protocol, such that for all $x \in \mathbb{M}(\mathsf{E}^{\mathsf{th}})$, security parameters $\lambda \in \Nat$, $(\mathsf{pk},$ $(\so{\mathcal{P}_k}{\share{\mathsf{sk}}{k}})_{k \in [m]})$ $\leftarrow$ $\keygen(1^{\lambda},$ $\si{\mathcal{P}_k}{})$, $\tdec(\thege(\mathsf{pk},$ $x),$ $(\si{\mathcal{P}_k}{\share{\mathsf{sk}}{k}})_{k \in [m]}) = x$. IND-CPA security of these schemes is analogous to the IND-CPA security of vanilla public-key encryption schemes \citep{semantic-security}, where the adversary controls less than $m$ parties.

We use the following threshold encryption schemes: $a)$ $\theg$: the threshold El Gamal encryption scheme \citep{threshold-cryptography} where $\mathbb{M}(\theg) = \Gone$ and $\mathbb{C}(\theg) = \Gone \times \Gone$, and $b)$ $\thez$: an optimised threshold Paillier encryption scheme from Damg{\aa}rd et al. \citep{dj-generalisation} where $\mathbb{M}(\thez) = \mathbb{Z}_N$ for an RSA modulus $N$ and $\mathbb{C}(\thez) = \mathbb{Z}_{N^2}^*$. $\theg$ is multiplicatively homomorphic in $\Gone$: for any two ciphertexts $c_1,c_2 \in \Gone \times \Gone$ encrypting messages $m_1,m_2 \in \Gone$, $c_1c_2$ (their component-wise multiplication in $\Gone$) decrypts to the message $m_1m_2$ (group multiplication in $\Gone$). $\thez$ is additively homomorphic in $\mathbb{Z}_N$: for any two ciphertexts $c_1,c_2 \in \mathbb{Z}_{N^2}^*$ encrypting messages $m_1,m_2 \in \mathbb{Z}_N$, $c_1c_2 \bmod N^2$ decrypts to the message $m_1+m_2 \mod N$. However, we require additive homomorphism in $\Zq$ (a prime order group). Thus, we let $N > q$, intepret messages in $\Zq$ as messages in $\mathbb{Z}_N$ and carefully use $\thez$'s homomorphic addition modulo $N$ to obtain homomorphic addition modulo $q$ (see Section \ref{sec:tech-details}). 

Both $\theg$ and $\thez$ are IND-CPA secure, respectively under the DDH assumption in $\Gone$ \citep{elgamal} and the decisional composite residuosity (DCR) assumption \citep{paillier}. Also, both support distributed key generation protocols. For $\theg$, each party already generates its key shares independently; for $\thez$, a secure key generation protocol can be designed \citep{damgaard-secure-keygen}. Further, the $\tdec$ protocols of both schemes provide \emph{simulation security}, i.e, the adversary's view in the $\tdec$ protocol can be simulated given access to a decryption oracle.

We also use a standard (non-threshold) IND-CPA secure public-key encryption scheme $\pke$ on message space $\Zq$.

\subsubsection{Shuffles} 

Let $\encscheme^{\mathsf{th}}$ be an $(m,m)$ threshold homomorphic encryption scheme, $\mathsf{pk}$ be a public key under $\encscheme^{\mathsf{th}}$, $\ve{\epsilon}$ be an $n$-length vector of ciphertexts against $\mathsf{pk}$, $\permfunc{1},\dots,\permfunc{m} \in \mathsf{Perm}(n)$ be secret permutations of parties $\mathcal{P}_{1},\dots,\mathcal{P}_{m}$, where $\mathsf{Perm}(n)$ denotes the space of permutation functions, and $\encscheme^{\mathsf{th}}.\renc(\mathsf{pk}, c)$ be re-encryption under $\encscheme^{\mathsf{th}}$ of ciphertext $c$ with fresh randomness. We let $\ve{\epsilon'} \leftarrow \shuffle(\encscheme^{\mathsf{th}}, \mathsf{pk}, \ve{\epsilon}, \si{\mathcal{P}_{1}}{\permfunc{1}}, \dots, \si{\mathcal{P}_{m}}{\permfunc{m}})$ be a shorthand for repeated re-encryption and permutation of $\ve{\epsilon}$ by each of $(\mathcal{P}_{k})_{k \in [m]}$ in sequence, such that for all $j \in [n]$, $\vc{\epsilon}{j}'$ $=$ $\encscheme^{\mathsf{th}}.\renc(\mathsf{pk}$, $\vc{\epsilon}{\pi(j)})$, where $\pi = \permfunc{m} \circ \dots \circ \permfunc{1}$. The order of parties in the $\shuffle$ protocol denotes that first $\mathcal{P}_{1}$ runs, then $\mathcal{P}_{2}$, etc.

\section{Our construction}
\label{construction}

Our traceable mixnet construction extends Camenisch et al.'s single-prover ZKPs of set membership \citep{Camenisch} and our novel ZKP of reverse set membership (Section \ref{sss:zkp-rev-membership-single}). We explain these protocols first.

\subsection{Single prover case}

\subsubsection{ZKP of set membership \citep{Camenisch}}
\label{sss:zkp-sm-single}

In this ZKP, given a Pedersen commitment \citep{Pedersen} $\gamma$ and a set of values $\phi$, a single prover proves knowledge of $v,r$ such that $\gamma = g_1^{v}h_1^r$ and $v \in \phi$. The main idea is that the verifier generates a \emph{fresh} BB signature key pair $x \xleftarrow{\$} \Zq, y \leftarrow g_2^x$ and sends to the prover the verification key $y$ and signatures $\sigma_{v'} \leftarrow g_1^{\frac{1}{x+v'}}$ on each $v' \in \phi$. The prover chooses a blinding factor $b \xleftarrow{\$} \Zq$ and sends to the verifier a blinded version $\tilde{\sigma}_v$ of the signature on the value $v$ committed by $\gamma$, as $\tilde{\sigma}_v \leftarrow \sigma_v^b = g_1^{\frac{b}{x+v}}$. Both then engage in a ZKP of knowledge $\pk{v,r,b}{\gamma=g_1^{v}h_1^{r} \wedge e(\tilde{\sigma}_v, y) = e(g_1,g_2)^{b}e(\tilde{\sigma}_v, g_2)^{-v}}$, which proves knowledge of a valid signature $(\tilde{\sigma}_v)^{1/b}$ on the value committed by $\gamma$ (see Section \ref{sss:bb-sig}). This is a proof of set membership because if $\gamma$ does not commit a member of $\phi$ then the proof fails since the prover does not obtain signatures on non-members of the set and cannot forge them. The scheme is an honest-verifier ZKP of set membership if $|\phi|$-Strong Diffie Hellman assumption holds in $(\Gone,\Gtwo)$ \citep{Camenisch}. 

A nice property of the scheme is that multiple proofs for a set $\Phi$ of commitments against the same set $\phi$ of values can be given efficiently by reusing verifier signatures. After obtaining signatures $(\sigma_v)_{v \in \phi}$, the prover can precompute $(\tilde{\sigma}_v)_{v \in \phi}$ in a stage 1. In stage 2, for each commitment $\gamma \in \Phi$ committing a value $v \in \phi$, the corresponding $\tilde{\sigma}_v$ can be looked up and the ZKP of knowledge can be constructed in $O(1)$ time. This results in an $O(|\phi|+|\Phi|)$ amortised complexity for proving set membership for $|\Phi|$ commitments.

\subsubsection{ZKP of reverse set membership}
\label{sss:zkp-rev-membership-single}

Now we show how to extend this idea to prove reverse set membership. Note that the BB signatures used above require messages to be in group $\Zq$. Since commitments are members of $\Gone$, one cannot use BB signatures to sign members of the set $\Phi$ of commitments for the reverse set membership proof. Recall, however, that the BBS+ signature scheme \citep{bbsplus-sig} lets one present a commitment $\gamma = g_1^vh_1^r$ along with $\rho_{\gamma}:=\nizkpk{v,r}{\gamma=g_1^vh_1^r}$ to the signer and obtain a BBS+ signature on the value $v$, without leaking $v$ to the signer. We exploit this property for the reverse set membership proof.

Our reverse set membership verifier generates fresh BBS+ signature key pairs $x \xleftarrow{\$} \Zq$, $y \leftarrow f_2^x$ and sends quasi-signatures $\hat{\sigma}_{\gamma} := (S, c, \hat{r}) \leftarrow ((f_1h_1^{\hat{r}}\gamma)^{\frac{1}{c+x}}, c, \hat{r})$ for each $\gamma = g_1^vh_1^r \in \Phi$, after verifying $\rho_{\gamma}$. If the prover knows commitment randomness $r$ for each $\gamma \in \Phi$, it can use $\hat{\sigma}_{\gamma}$ to derive a valid BBS+ signature $\sigma_v \leftarrow$ $(S, c, \mathsf{r}:=\hat{r}+r) =$ $((f_1g_1^vh_1^{\hat{r}+r})^{\frac{1}{c+x}}, c, \hat{r}+r)$ on the committed value $v$ and store $\sigma_v$ indexed by $v$. To prove that a given $v$ is committed by some commitment in $\Phi$, the prover looks up $\sigma_v$ in O(1) time, blinds each component of $\sigma_v$ to obtain a blinded signature $\tilde{\sigma}_v$ and proves knowledge of a BBS+ signature on $v$ by revealing only $\tilde{\sigma}_v$ to the verifier. This is a proof of reverse set membership because the prover can obtain valid BBS+ signatures only on values committed by $\gamma \in \Phi$ and cannot forge it for $v$ if no $\gamma \in \Phi$ committed $v$. This protocol also enjoys $O(|\phi|+|\Phi|)$ amortised complexity for multiple proofs for each $v \in \phi$ against the same set of commitments $\Phi$. See Appendix \ref{appendix:rsm-single} for the detailed protocol.

\subsection{Overview of our construction}
\label{construction:traceable-mixnets}

\begin{figure*}
	\centering
	\begin{subfigure}[t]{0.33\textwidth}
		\centering
		\scalebox{0.75}{
			\begin{tabular}{c@{\hskip 0.1cm} c@{\hskip 0.1cm} c@{\hskip 0.1cm} c@{\hskip 0.6cm}|}
				\begin{tabular}{|c|}
					\hline
					\cellcolor{gray!20}Input \\
					\hline
					$\vc{\gamma}{1}$ \\
					\hline
					$\vc{\gamma}{2}$ \\
					\hline
					$\vc{\gamma}{3}$ \\
					\hline
					$\vc{\gamma}{4}$ \\
					\hline
					$\vc{\gamma}{5}$ \\
					\hline
				\end{tabular} & 
				\begin{tabular}{c|c|}
					\multicolumn{2}{r}{$\mixer{1}$} \\
					\cline{2-2}
					$\rightarrow$ & \multirow{5}{*}{$\permfunc{1}$} \\
					$\vdots$ &  \\
					$\vdots$ & \\
					$\rightarrow$ & \\
					\cline{2-2}
				\end{tabular} &
				\begin{tabular}{c|c|c}
					\multicolumn{3}{c}{$\mixer{2}$} \\
					\cline{2-2}
					$\rightarrow$ & \multirow{5}{*}{$\permfunc{2}$} & $\rightarrow$\\
					$\vdots$ &  & $\vdots$ \\
					$\vdots$ & & $\vdots$\\
					$\rightarrow$ & & $\rightarrow$\\
					\cline{2-2}
				\end{tabular} &
				\begin{tabular}{|c|}
					\hline
					\multicolumn{1}{|c|}{\cellcolor{gray!20}Output} \\
					\hline
					$\vc{v}{1}'=\vc{v}{3}$ \\
					\hline
					$\vc{v}{2}'=\vc{v}{5}$ \\
					\hline
					$\vc{v}{3}'=\vc{v}{2}$ \\
					\hline
					$\vc{v}{4}'=\vc{v}{1}$ \\
					\hline
					$\vc{v}{5}'=\vc{v}{4}$ \\
					\hline
				\end{tabular}
			\end{tabular}
		}
		\caption{The mixing process}
		\label{fig:mixing}
	\end{subfigure}%
	~
	\begin{subfigure}[t]{0.33\textwidth}
		\centering
		\scalebox{0.75}{
			\begin{tabular}{c@{\hskip 0.1cm} c@{\hskip 0.1cm} c@{\hskip 0.1cm} c@{\hskip 0.6cm}|}
				\begin{tabular}{|l|}
					\hline
					\cellcolor{gray!20}Input \\
					\hline
					\cellcolor{green!30}$\vc{\gamma}{1},\tilde{\sigma}_1$ \\
					\hline
					\cellcolor{green!30}$\vc{\gamma}{2},\tilde{\sigma}_2$ \\
					\hline
					\cellcolor{green!30}$\vc{\gamma}{3},\color{red} \tilde{\sigma}_3$ \\
					\hline
					$\vc{\gamma}{4},\color{red} \tilde{\sigma}_4$ \\
					\hline
					$\vc{\gamma}{5},\tilde{\sigma}_5$ \\
					\hline
				\end{tabular} & 
				\begin{tabular}{c|c|}
					\multicolumn{2}{r}{$\mixer{1}$} \\
					\cline{2-2}
					$\leftarrow$ & \multirow{5}{*}{$\pi^{(1)^{-1}}$} \\
					$\vdots$ &  \\
					$\vdots$ & \\
					$\leftarrow$ & \\
					\cline{2-2}
				\end{tabular} &
				\begin{tabular}{c|c|c}
					\multicolumn{3}{c}{$\mixer{2}$} \\
					\cline{2-2}
					$\leftarrow$ & \multirow{5}{*}{$\pi^{(2)^{-1}}$} & $\leftarrow$\\
					$\vdots$ & & $\vdots$ \\
					$\vdots$ & & $\vdots$\\
					$\leftarrow$ & & $\leftarrow$\\
					\cline{2-2}
				\end{tabular} &
				\begin{tabular}{|c|}
					\hline
					\multicolumn{1}{|c|}{\cellcolor{gray!20}Output} \\
					\hline
					$\vc{v}{1}', \color{red} \sigma'_1$ \\
					\hline
					\cellcolor{green!30}$\vc{v}{2}',\sigma'_2$ \\
					\hline
					\cellcolor{green!30}$\vc{v}{3}',\sigma'_3$ \\
					\hline
					\cellcolor{green!30}$\vc{v}{4}',\sigma'_4$ \\
					\hline
					$\vc{v}{5}',\color{red} \sigma'_5$ \\
					\hline
				\end{tabular}
			\end{tabular} 
		}
		\caption{Protocol $\zkpsm$ (Figure \ref{fig:distr-set-membership})}
		\label{fig:summary-zkpsm}
	\end{subfigure}%
	~
	\begin{subfigure}[t]{0.33\textwidth}
		\centering
		\scalebox{0.75}{
			\begin{tabular}{c@{\hskip 0.1cm} c@{\hskip 0.1cm} c@{\hskip 0.1cm} c@{\hskip 0.2cm}}
				\begin{tabular}{|l|}
					\hline
					\cellcolor{gray!20}Input \\
					\hline
					\cellcolor{green!30}$\vc{\gamma}{1},\hat{\sigma}_1$ \\
					\hline
					\cellcolor{green!30}$\vc{\gamma}{2},\hat{\sigma}_2$ \\
					\hline
					\cellcolor{green!30}$\vc{\gamma}{3},\hat{\sigma}_3$ \\
					\hline
					$\vc{\gamma}{4},\color{red} \hat{\sigma}_4$ \\
					\hline
					$\vc{\gamma}{5},\color{red} \hat{\sigma}_5$ \\
					\hline
				\end{tabular} & 
				\begin{tabular}{c|c|}
					\multicolumn{2}{r}{$\mixer{1}$} \\
					\cline{2-2}
					$\rightarrow$ & \multirow{5}{*}{$\permfunc{1}$} \\
					$\vdots$ & \\
					$\vdots$ & \\
					$\rightarrow$ & \\
					\cline{2-2}
				\end{tabular} &
				\begin{tabular}{c|c|c}
					\multicolumn{3}{c}{$\mixer{2}$} \\
					\cline{2-2}
					$\rightarrow$ & \multirow{5}{*}{$\permfunc{2}$} & $\rightarrow$\\
					$\vdots$ & & $\vdots$ \\
					$\vdots$ & & $\vdots$\\
					$\rightarrow$ & & $\rightarrow$\\
					\cline{2-2}
				\end{tabular} &
				\begin{tabular}{|c|}
					\hline
					\multicolumn{1}{|c|}{\cellcolor{gray!20}Output} \\
					\hline
					$\vc{v}{1}', \tilde{\sigma}'_1$ \\
					\hline
					\cellcolor{green!30}$\vc{v}{2}',\color{red} \tilde{\sigma}'_2$ \\
					\hline
					\cellcolor{green!30}$\vc{v}{3}',\tilde{\sigma}'_3$ \\
					\hline
					\cellcolor{green!30}$\vc{v}{4}',\tilde{\sigma}'_4$ \\
					\hline
					$\vc{v}{5}',\color{red} \tilde{\sigma}'_5$ \\
					\hline
				\end{tabular}
			\end{tabular} 
		}
		\caption{Protocol $\zkprsm$ (Figure \ref{fig:distr-rev-set-membership})}
		\label{fig:summary-zkprsm}
	\end{subfigure}%
	\caption{{\small A summary of our construction: a) the mixing process with two mix-servers; b) a $\zkpsm$ call for index sets $I,J$ represented by the green entries in the mixnet input and output lists: in stage 1, the querier prepares valid BB signatures $\vc{\sigma}{j}'$ for $(\vc{v}{j}')_{j \in J}$ and invalid ones (marked red) for $(\vc{v}{j}')_{j \not\in J}$; blinded versions $(\vc{\tilde{\sigma}}{i})_{i \in [n]}$ of these signatures appear alongside commitments $(\vc{\gamma}{i})_{i \in [n]}$; the DPKs of stage 2 pass only for commitments $(\vc{\gamma}{i})_{i \in I}$ whose corresponding blinded signature was valid (not marked red); c) a $\zkprsm$ call: the querier prepares valid BBS+ quasi-signatures $(\vc{\hat{\sigma}}{i})_{i \in I}$ for commitments $(\vc{\gamma}{i})_{i \in I}$ and invalid ones for $(\gamma_i)_{i \not\in I}$; blinded versions $(\vc{\tilde{\sigma}}{j}')_{j \in [n]}$ of signatures on the respective committed values appear alongside the values $(\vc{v}{j}')_{j \in [n]}$; the DPKs of stage 2 pass only for values $(\vc{v}{j}')_{j \in J}$ whose blinded signature was valid.}}
	\label{fig:construction-sketch}
\end{figure*}

We now give an overview of our traceable mixnet construction (see Section \ref{sec:tech-details} for detailed protocol steps). In our construction, senders send threshold encryptions of their sensitive values as input ciphertexts which get shuffled by a series of mix-servers and eventually decrypted just like a standard re-encryption mixnet \citep{mixnet-sok}. However, in addition, the senders also upload Pedersen commitments to the encrypted values along with the input ciphertexts and secret-share the commitment openings among the mix-servers (see Figure \ref{fig:mixing}). The mix-servers use these shares to distributedly answer BTraceIn and BTraceOut queries via our distributed and batched ZKPs of set membership and reverse set membership ($\zkpsm$ and $\zkprsm$).

A $\zkpsm$ protocol for index sets $I,J$ allows the mix-servers to prove for each uploaded commitment at an index $i \in I$ that it commits a value in the set of output plaintexts at indices $J$; the querier's output is the indices $I^* \subseteq I$ where the statement holds. As in the single-prover ZKP of set membership, the querier is asked to provide BB signatures on the set of plaintexts at indices $J$. Note, however, that in the single prover case, the prover knows the commitment's committed value $v$, which allows it to look-up its blinded signature $\tilde{\sigma}_v$ in $O(1)$ time. In the distributed mixnet setting, no prover (mix-server) knows the committed value, randomness or the permutation between the list of commitments and plaintexts. The challenge is to efficiently identify the blinded signature on a given commitment's committed value without letting any set of less than $m$ mix-servers or the querier learn these secrets.

To solve this challenge, the querier signatures are encrypted and shuffled by the mix-servers \emph{in the reverse direction} as the forward mixnet shuffle --- following the inverse of the mixnet permutation -- and are homomorphically blinded by random blinding factors before decryption (see Figure \ref{fig:summary-zkpsm}). This process produces the blinded signatures next to the corresponding input commitments. To prove set membership, the mix-servers use the blinding factors and the commitment opening shares sent by the senders to jointly prove knowledge of a BB signature on the committed value via a DPK. 

A technical complication in the above outline is that in addition to supplying BB signatures for each plaintext at an index $j \in J$, the querier must also provide invalid signatures for plaintexts at indices $j \not\in J$. These invalid (``fake'') signatures should also be encrypted, reverse-shuffled and blinded in the same way as the valid signatures. Without these invalid signatures, the blinded signatures would appear against exactly the input list commitments that committed a plaintext in set $\vc{v}{J}'$. This would reveal even for commitments outside set $\vc{c}{I}$ whether they committed a plaintext in set $\vc{v}{J}'$ or not, violating our secrecy definition. With these invalid signatures, the DPKs pass only for commitments in $\vc{c}{I}$ that commit a value in set $\vc{v}{J}'$ but no information is revealed for commitments outside $\vc{c}{I}$.

A $\zkprsm$ protocol for index sets $I,J$ allows the mix-servers to prove for each output plaintext at an index $j \in J$ that it is committed by a member of the set of commitments at indices $I$; the querier's output is the indices $J^* \subseteq J$ where the statement holds. As in the single-prover ZKP of reverse set membership, the querier is asked to provide BBS+ quasi-signatures on the set of commitments at indices $I$ (and invalid quasi-signatures for indices $[n] \setminus I$). These quasi-signatures are encrypted, homomorphically converted to encrypted BBS+ signatures and then shuffled by the mix-servers \emph{in the forward direction} --- following the mixnet permutation --- to obtain encrypted BBS+ signatures next to the corresponding plaintexts (see Figure \ref{fig:summary-zkprsm}). These encrypted signatures are then homomorphically blinded before decryption and the mix-servers use the blinding factors to provide a DPK of a BBS+ signature on the corresponding plaintext.

Even when all the mix-servers are cheating, they cannot make make the proofs pass for an incorrect entry and include, e.g., a commitment that did not commit a value in set $\vc{v}{J}'$ in $\zkpsm$ output. However, this does not prevent them from deliberately failing proofs for commitments that actually committed a value in $\vc{v}{J}'$, producing a smaller-than-correct output set and violating Definition \ref{def:snd}. Thus, we run $\zkpsm$ against both $J$ and $[n] \setminus J$ in a $\btrin$ call and make the querier abort if proofs against both the runs failed for some commitment (similarly for $\btrout$).

\subsection{Technical details}
\label{sec:tech-details}

Figure \ref{fig:construction} shows our traceable mixnet construction in detail (this construction preserves secrecy only in the honest-but-curious (HBC) setting; see Section \ref{sec:hbc-to-malicious} for malicious security). We use the threshold $\thez$ scheme to create ciphertexts $\vc{\epsilon}{i}$, the $\theg$ scheme to reverse-shuffle BB signatures in $\zkpsm$ and both $\theg$ and $\thez$ to forward-shuffle BBS+ signatures in $\zkprsm$. Further, we use a standard public-key encryption scheme $\pke$ for securely sending shares of commitment openings to the individual mix-servers. The $\keygen$ step creates public/private keys for all these schemes. Secret keys for $\theg$ and $\thez$ are shared among the mix-servers and secret keys for $\pke$ for each public key is held by individual mix-servers. Via the $\enc$ algorithm, senders upload ciphertexts $\vc{\epsilon}{i}$ encrypting their secret values, along with commitments $\vc{\gamma}{i}$ and encryptions of secret shares of the commitment openings for each mix-server. They also upload proofs of knowledge $\vc{\pokcomm}{i}$ of the commitment openings and encryptions $\vc{\epsilon}{r_i}$ of commitment randomnesses to enable the $\zkprsm$ proofs. During the $\mix$ protocol, the mix-servers shuffle and threshold-decrypt $\vc{\epsilon}{i}$ to produce permuted plaintexts $\vc{v}{j}'=(\vc{v}{\pi(j)})_{j \in [n]}$, where $\pi$ is composed of secret permutations $\permfunc{k}$ of each mix-server $\mixer{k}$. Each $\mixer{k}$ stores $\permfunc{k}$ and decryptions of commitment opening shares to jointly answer $\btrin$/$\btrout$ queries via $\zkpsm$/$\zkprsm$. 

\begin{figure}
	\centering
	\scalebox{0.8}{
		\begin{tabular}{@{}|l|@{}}
			\hline
			$\underline{\keygen(1^{\lambda}, (\si{\mixer{k}}{})_{k \in [m]})}$: \\       
	        \enskip $(\mixer{k})_{k\in[m]}$: \\
	        \enskip \quad $\pkepk{k}, \pkesk{k} \leftarrow \pkeg(1^{\lambda})$; publish $\pkepk{k}$ \\
	        \enskip $\thegpk, (\si{\mixer{k}}{\thegsk{k}})_{k \in [m]} \leftarrow \thegg(1^{\lambda}, (\si{\mixer{k}}{})_{k \in [m]})$ \\		
			\enskip $\thezpk, (\so{\mixer{k}}{\thezsk{k}})_{k \in [m]} \leftarrow \thezg(1^{\lambda}, (\si{\mixer{k}}{})_{k\in[m]})$ \\
			\enskip \textbf{output} $\mpk:= ((\pkepk{k})_{k \in [m]}, \thegpk, \thezpk),$ \\
			\quad $(\so{\mixer{k}}{\msk{k} := (\pkesk{k}, \thegsk{k}, \thezsk{k})})_{k \in [m] }$ \\ 
			\\
			
			$\underline{\enc(\mpk:=((\pkepk{k})_{k \in [m]}, \cdot, \thezpk), v \in \Zq)}$: \quad \mycomment{run by $S_1, \dots, S_n$} \\
			\enskip $\epsilon \leftarrow \theze(\thezpk, v)$ \mycomment{interpret $v$ as an element of $\mathbb{Z}_N$} \\
			\enskip $r \xleftarrow{\$} \Zq$; $\gamma \leftarrow g_1^{v}h_1^{r}$; $\pokcomm \leftarrow \nizkpk{(v,r)}{\gamma = g_1^vh_1^r}$ \\
			\enskip $\epsilon_r \leftarrow \theze(\thezpk, r)$ \mycomment{interpret $r$ as an element of $\mathbb{Z}_N$} \\
			\enskip $(\share{v}{k})_{k \in [m]} \leftarrow \sharealg_{m,m}(v)$; $(\share{r}{k})_{k \in [m]} \leftarrow \sharealg_{m,m}(r)$ \\
			\enskip $(\share{\mathsf{ev}}{k})_{k \in [m]} \leftarrow (\pkee(\pkepk{k}, \share{v}{k}))_{k \in [m]}$ \\	
			\enskip $(\share{\mathsf{er}}{k})_{k \in [m]} \leftarrow (\pkee(\pkepk{k}, \share{r}{k}))_{k \in [m]}$ \\
			\enskip \textbf{output} $c := (\epsilon, \gamma, (\share{\mathsf{ev}}{k},\share{\mathsf{er}}{k})_{k \in [m]}, \pokcomm, \epsilon_r)$ \\ 
			\\
			
			$\mix(\mpk:=(\cdot, \cdot, \thezpk), \ve{c}:=(\vc{\epsilon}{i}, \cdot, (\vcs{\mathsf{ev}}{i}{k},\vcs{\mathsf{er}}{i}{k})_{k\in[m]}, \cdot, \cdot)_{i \in [n]},$ \\
			$\underline{\enskip(\si{\mixer{k}}{\msk{k}:=(\pkesk{k}, \cdot, \thezsk{k})})_{k \in [m]}):\quad\quad\quad\quad\quad\quad\quad\quad\quad}$ \\
			\enskip $(\mixer{k})_{k \in [m]}$: \quad $\permfunc{k} \xleftarrow{\$} \perm{n}$ \\ 
			\enskip $(\vc{\epsilon}{j}')_{j \in [n]} \leftarrow \shuffle(\thez, \thezpk, (\vc{\epsilon}{i})_{i \in [n]}, \si{\mixer{1}}{\permfunc{1}}, \dots, \si{\mixer{m}}{\permfunc{m}})$ \\
			\enskip \mycomment{$\vc{\epsilon}{j}' = \thezre(\thezpk, \vc{\epsilon}{\pi(j)})$, where $\pi = \permfunc{m} \circ \dots \circ \permfunc{1}$} \\
			\enskip $({\vc{v}{j}'})_{j \in [n]} \leftarrow \ve\thezd((\vc{\epsilon}{j}')_{j \in [n]}, (\si{\mixer{k}}{\thezsk{k}})_{k \in [m]})$ \\       
			\enskip $(\mixer{k})_{k \in [m]}$: \\
			\enskip \quad $\vs{v}{k}\leftarrow \ve\pked(\pkesk{k}, \vs{\mathsf{ev}}{k})$;\quad   $\vs{r}{k}\leftarrow \ve\pked(\pkesk{k}, \vs{\mathsf{er}}{k})$ \\ 
			\enskip \textbf{output} $\ve{v'}, (\so{\mixer{k}}{\wit{k}:= (\permfunc{k}, \vs{v}{k}, \vs{r}{k})})_{k \in [m]}$ \\ 
			\\
			
			$\btrin(\mpk:=(\cdot, \thegpk, \cdot), \ve{c}:=(\cdot, \vc{\gamma}{i}, \cdot, \cdot, \cdot)_{i \in [n]},$ \\
			\enskip$\ve{v'} \in \Zq^n, I \subseteq [n], J \subseteq [n],$ \\
			$\enskip(\si{\mixer{k}}{\msk{k}:=(\cdot, \thegsk{k}, \cdot), \wit{k} := (\permfunc{k}, \vs{v}{k}, \vs{r}{k})})_{k \in [m]},$ \\
			$\underline{\quad\si{\querier}{}):\quad\quad\quad\quad\quad\quad\quad\quad\quad\quad\quad\quad\quad\quad\quad\quad\quad\quad\quad\quad\quad\quad\quad}$ \\
			\enskip $\so{\querier}{I^*} \leftarrow \zkpsm(\thegpk, \ve{\gamma}, \ve{v'}, I, J,$ \\
			\enskip \quad\quad$(\si{\mixer{k}}{\thegsk{k}, \permfunc{k}, \vs{v}{k}, \vs{r}{k}})_{k \in [m]}, \si{\querier}{})$ \mycomment{see Fig. \ref{fig:distr-set-membership}} \\
			\enskip $\so{\querier}{I^*_{\mathsf{c}}} \leftarrow \zkpsm(\thegpk, \ve{\gamma}, \ve{v'}, I, [n]\setminus J,$ \\
			\enskip \quad\quad$(\si{\mixer{k}}{\thegsk{k}, \permfunc{k}, \vs{v}{k}, \vs{r}{k}})_{k \in [m]}, \si{\querier}{})$ \mycomment{see Fig. \ref{fig:distr-set-membership}} \\			
			\enskip $\querier:$ \textbf{if $I^* \cup I^*_{\mathsf{c}} \neq I$:} \textbf{abort} \\
			\enskip \textbf{output} $\so{\querier}{\vc{c}{I^*}}$ \\
			\\
			
			$\btrout(\mpk:=(\cdot, \thegpk, \thezpk), \ve{c}:=(\cdot, \vc{\gamma}{i}, \cdot, \vc{\epsilon_r}{i}, \vc{\pokcomm}{i})_{i \in [n]},$ \\
			\enskip$\ve{v'} \in \Zq^n, I \subseteq [n], J \subseteq [n],$ \\
			$\enskip(\si{\mixer{k}}{\msk{k}:=(\cdot, \thegsk{k}, \thezsk{k}), \wit{k} := (\permfunc{k}, \vs{v}{k}, \vs{r}{k})})_{k \in [m]},$ \\
			$\underline{\enskip\si{\querier}{}):\quad\quad\quad\quad\quad\quad\quad\quad\quad\quad\quad\quad\quad\quad\quad\quad\quad\quad\quad\quad\quad\quad\quad\enskip}$ \\      
			\enskip $\so{\querier}{J^*} \leftarrow \zkprsm(\thegpk, \thezpk, \ve{\gamma}, \ve{\pokcomm}, \vc{\epsilon}{r}, \ve{v'}, I, J,$ \\
			\enskip $\quad\quad(\si{\mixer{k}}{\thegsk{k}, \thezsk{k}, \permfunc{k}, \vs{v}{k}, \vs{r}{k}})_{k \in [m]}, \si{\querier}{})$ \mycomment{see Fig. \ref{fig:distr-rev-set-membership}} \\ 
			\enskip $\so{\querier}{J^*_{\mathsf{c}}} \leftarrow \zkprsm(\thegpk, \thezpk, \ve{\gamma}, \ve{\pokcomm}, \vc{\epsilon}{r}, \ve{v'}, [n] \setminus I, J,$ \\
			\enskip $\quad\quad(\si{\mixer{k}}{\thegsk{k}, \thezsk{k}, \permfunc{k}, \vs{v}{k}, \vs{r}{k}})_{k \in [m]}, \si{\querier}{})$ \mycomment{see Fig. \ref{fig:distr-rev-set-membership}}\\
			\enskip $\querier:$ \textbf{if $J^* \cup J^*_{\mathsf{c}} \neq J$:} \textbf{abort} \\
			\enskip \textbf{output} $\so{\querier}{\vc{v}{J^*}'}$ \\
			\hline
		\end{tabular}
	}
	\caption{{\small Our construction of a traceable mixnet.}}
	\label{fig:construction}
\end{figure}

Figure \ref{fig:distr-set-membership} shows our $\zkpsm$ protocol. In stage 1, the querier $\querier$ publishes valid BB signatures $\vc{\sigma}{j}'$ on $\vc{v}{j}'$ for each $j \in J$ and invalid signatures (a fixed group element) for $j \not\in J$. These ``signatures'' are encrypted under $\theg$ and shuffled by the mix-servers in the reverse direction from $\mixer{m}$ to $\mixer{1}$, with each $\mixer{k}$ using permutation $(\permfunc{k})^{-1}$, to produce encrypted signatures $\vc{\epsilon}{\ve{\sigma}_i}$ on the value $\vc{v}{i}$ committed by $\vc{\gamma}{i}$ (the encrypted signature being valid only if $\vc{v}{i} \in \vc{v}{J}'$). The mix-servers then use the multiplicative homomorphism of $\theg$ to jointly obtain encryptions $\vc{\tilde{\epsilon}}{\ve{\sigma}_i} \leftarrow \vc{\epsilon}{\ve{\sigma}_i}^{\sum_{k \in [m]} \vcs{b}{i}{k}}$, where each $\mixer{k}$ contributes blinding factors $\vs{b}{k} \xleftarrow{\$} \Zq^n$. The plaintext blinded signatures $\vc{\tilde{\sigma}}{i}$ are finally obtained by threshold decryption of $\vc{\tilde{\epsilon}}{\ve{\sigma}_i}$ and published alongside $\vc{\gamma}{i}$. In stage 2, $\vc{\tilde{\sigma}}{i}$ for each $(\vc{\gamma}{i})_{i \in I}$ are looked up in $O(1)$ time. A DPK of a BB signature on the value committed by $\vc{\gamma}{i}$ is given by proving joint knowledge of the commitment openings and blinding factors to unblind $\vc{\tilde{\sigma}}{i}$ to a valid signature. This DPK has the format of Section \ref{sss:distr-pok} and can be given efficiently since each $\mixer{k}$ knows additive shares $\vcs{v}{i}{k}, \vcs{r}{i}{k}$ for the commitment openings and $\vcs{b}{i}{k}$ for the blinding factors. All indices for which the DPK passed are included in $\querier$'s output $I^*$. If $\querier$ is not corrupted, $(\mixer{k})_{k \in [m]}$ do not learn $I^*$ since they do not learn which DPK passed. The amortised complexity of the entire protocol is $O(n)$.

\begin{figure}[t]
\centering
\scalebox{0.8}{
    \begin{tabular}{@{}|l|@{}}
    \hline
    \multicolumn{1}{|l|}{\textbf{Participants:} Mix-servers $(\prover{k})_{k \in [m]}$, Querier $\querier$} \\
    \multicolumn{1}{|l|}{\textbf{Common~input: $(\thegpk$, $\ve{\gamma} = (\vc{\gamma}{i})_{i \in [n]}, \ve{v'}=(\vc{v}{j}')_{j \in [n]}, I, J)$ s.t.:}} \\
    \multicolumn{1}{|l|}{\quad \quad $\vc{\gamma}{i} \in \Gone, \vc{v}{j}' \in \Zq, I,J \subseteq [n]$} \\
    \multicolumn{1}{|l|}{\textbf{$\prover{k}$'s input: $(\thegsk{k}$, $\permfunc{k}, \vs{v}{k}, \vs{r}{k})$ s.t. letting}} \\
    \multicolumn{1}{|l|}{\textbf{\quad $\ve{v}:= \sum_{k \in [m]} \vs{v}{k}$, $\ve{r} := \sum_{k \in [m]} \vs{r}{k}$, $\pi:= \permfunc{m} \circ \dots \circ \permfunc{1}$:}} \\
    \multicolumn{1}{|l|}{\quad \quad 1) $\forall i \in [n]: \vc{\gamma}{i} = g_1^{\vc{v}{i} }h_1^{\vc{r}{i} }$} \\
    \multicolumn{1}{|l|}{\quad \quad 2) $\forall j \in [n]: \vc{v}{j}' :=  \vc{v}{\pi(j)}$} \\
    \multicolumn{1}{|l|}{\textbf{$\querier$'s output: $I^* := \{i \in I \mid \vc{v}{i} = \vc{v}{j}' \text{ for some } j \in J\}$ }} \\
    
	\hline

    \textbf{\underline{Stage 1:}} \\
    \underline{\emph{Signature generation}} \\
    \enskip $\querier$: \enskip $x \xleftarrow{\$} \Zq$; $y \leftarrow g_2^x$ \\
    \enskip \quad\quad $\ve{\sigma'} \leftarrow (\vc{\sigma}{j}')_{j \in [n]},$ where $\vc{\sigma}{j}' := \begin{cases}g_1^{\frac{1}{x+ \vc{v}{j}' }} & \text{if}~ j \in J \\ g_1  & \text{otherwise} \end{cases}$ \\
    \enskip \quad\quad $\ve{\epsilon}_{\sigma}' \leftarrow \ve{\thege}(\thegpk, \ve{\sigma'})$ \\
    \enskip \quad\quad publish $y, \ve{\sigma'}, \ve{\epsilon}_{\sigma}'$ \\ 
    
	\underline{\emph{Shuffling}} \\
	\enskip $\ve{\epsilon}_{\sigma} \leftarrow \shuffle(\theg, \thegpk, \ve{\epsilon}_{\sigma}', \si{\mixer{m}}{(\permfunc{m})^{-1}}, \dots, \si{\mixer{1}}{(\permfunc{1})^{-1}})$ \\
	\enskip \mycomment{$\vc{\epsilon}{\sigma_i} = \thegre(\thegpk, \vc{\sigma}{i})$; $\vc{\sigma}{i}:=g_1^{\frac{1}{x+ \vc{v}{i}}}$ if $\pi^{-1}(i) \in J$ else $g_1$} \\
	
	\underline{\emph{Homomorphic blinding}}	\\
 	\enskip $(\prover{k})_{k\in[m]}$: \\
    \enskip \quad $\share{\ve{b}}{k} \xleftarrow{\$} \Zq^n$\\
	\enskip \quad publish $\vs{\tilde{\epsilon}}{k}_{\sigma} \leftarrow \ve{\thegre}(\thegpk, {\ve{\epsilon}}^{\vs{b}{k}}_{\sigma})$ \quad \mycomment{$\vs{\tilde{\epsilon}}{k}_{\sigma_i} = \thege(\thegpk, \vc{\sigma}{i}^{\vcs{b}{i}{k}})$} \\
    \enskip $(\prover{k})_{k\in[m]}$: $\ve{\tilde{\epsilon}}_{\sigma} \leftarrow \prod_{k \in [m]} \vs{\tilde{\epsilon}}{k}_{\sigma}$ \quad \mycomment{$\ve{\tilde{\epsilon}}_{\sigma_i} = \thege(\thegpk, \vc{\sigma}{i}^{\sum_{k\in[m]}\vcs{b}{i}{k}})$}\\
    
    \underline{\emph{Threshold decryption}} \\
    \enskip $\ve{\tilde{\sigma}} \leftarrow \ve{\thegd}(\ve{\tilde{\epsilon}}_{\sigma}, (\si{\mixer{k}}{\thegsk{k}})_{k \in [m]})$ \\
    \enskip \mycomment{$\vc{\tilde{\sigma}}{i} = \vc{\sigma}{i}^{\vc{b}{i}} = g_1^{\frac{\vc{b}{i}}{x+ \vc{v}{i}}}$ if $\pi^{-1}(i) \in J$ else $g_1^{\vc{b}{i}}$} \\
    \enskip publish $\ve{\tilde{\sigma}}$ \\
     \\ 

    \textbf{\underline{Stage 2:}} \\
    \enskip $I^* \leftarrow \emptyset$ \\
    \enskip \textbf{for} $i \in I$: \\
    \enskip \enskip $\so{\querier}{res} \leftarrow \mathsf{DPK}((\ve{\gamma}_i, \vc{\tilde{\sigma}}{i}, y),p_{\mathsf{BB}_i},\si{\mixer{k}}{\vcs{v}{i}{k}, \vcs{r}{i}{k}, \vcs{b}{i}{k}}, \si{\querier}{})$ \\
    \enskip \enskip \enskip where $p_{\mathsf{BB}_i}:=(\ve{\gamma}_i = g_1^{\vc{v}{i}}h_1^{\vc{r}{i}} \wedge e(\vc{\tilde{\sigma}}{i}, y) = e(g_1, g_2)^{\vc{b}{i}} e(\vc{\tilde{\sigma}}{i}, g_2)^{-\vc{v}{i}})$ \\
    \enskip \enskip $\querier$: \enskip \textbf{if $res = 1$:} $I^* \leftarrow I^* \cup \{i\}$ \\
    \enskip \textbf{endfor} \\
    \enskip \textbf{output} $\so{\querier}{I^*}$ \\
    \hline
    \end{tabular}
}
\caption{{\small Protocol $\zkpsm$ (see summary in Figure \ref{fig:summary-zkpsm}).}}
\label{fig:distr-set-membership}
\end{figure}

\begin{figure}
\centering
\scalebox{0.8}{
    \begin{tabular}{@{}|l|@{}}
    \hline    
    \multicolumn{1}{|l|}{\textbf{Participants:} Mix-servers $(\prover{k})_{k \in [m]}$, Querier $\querier$} \\
    \multicolumn{1}{|l|}{\textbf{Common~input: $(\thegpk, \thezpk, \ve{\gamma} = (\vc{\gamma}{i})_{i \in [n]}, \ve{\pokcomm} = (\vc{\pokcomm}{i})_{i \in [n]},$}} \\
    \multicolumn{1}{|l|}{\textbf{\quad $\ve{\epsilon_r} = (\vc{\epsilon_r}{i})_{i \in [n]},\ve{v'}=(\vc{v}{j}')_{j \in [n]}, I, J)$ s.t.:}} \\
    \multicolumn{1}{|l|}{\quad \quad $\vc{\gamma}{i} \in \Gone, \vc{\pokcomm}{i} = \nizkpk{v,r}{\vc{\gamma}{i} = g_1^{v}h_1^r}, \vc{\epsilon}{r_i} \in \mathbb{C}(\thez)$} \\
    \multicolumn{1}{|l|}{\quad \quad $\vc{v}{j}' \in \Zq, I,J \subseteq [n]$} \\
    \multicolumn{1}{|l|}{\textbf{$\prover{k}$'s input:~~ $(\thegsk{k}, \thezsk{k}, \permfunc{k}, \vs{v}{k}, \vs{r}{k})$ s.t. letting}} \\
    \multicolumn{1}{|l|}{\textbf{\quad $\ve{v}:= \sum_{k \in [m]} \vs{v}{k}$, $\ve{r}:=\sum_{k \in [m]} \vs{r}{k}$, $\pi:= \permfunc{m} \circ \dots \circ \permfunc{1}$:}} \\
    \multicolumn{1}{|l|}{\quad \quad 1) $\forall i \in [n]: \vc{\gamma}{i} = g_1^{\vc{v}{i} }h_1^{\vc{r}{i}}\ \text{ and } \vc{\epsilon}{r_i} \leftarrow \theze(\thezpk, \vc{r}{i})$} \\
    \multicolumn{1}{|l|}{\quad \quad 2) $\forall j \in [n]: \vc{v}{j}' :=  \vc{v}{\pi(j)}$} \\
    \multicolumn{1}{|l|}{\textbf{$\querier$'s output: $J^* := \{j \in J \mid \vc{v}{j}' = \vc{v}{i} \text{ for some } i \in I\}$ }} \\
    \hline

    \textbf{\underline{Stage 1:}} \\    
    \underline{\emph{Signature generation}} \\
    \enskip $\querier$: \enskip \textbf{for each} $i \in I$: \textbf{abort} if $\nizkver(\vc{\gamma}{i}, \vc{\pokcomm}{i}) \neq 1$ \\ 
    \enskip \quad\quad  $x \xleftarrow{\$} \Zq$; $y \leftarrow f_2^x$ ; $\ve{c} \xleftarrow{\$} \Zq^n; \ve{\hat{r}} \xleftarrow{\$} \Zq^n $  \\ 
    \enskip \quad\quad $\ve{S} \leftarrow (\vc{S}{i})_{i\in[n]}$ where $\vc{S}{i}:=\begin{cases} (f_1 h_1^{\vc{\hat{r}}{i}} \vc{\gamma}{i})^{\frac{1}{x+\vc{c}{i}}} & \text{if}~ i \in I \\ f_1^0 & \text{otherwise}\end{cases}$ \\ 
    \enskip \quad\quad $(\ve{\epsilon_S}, \ve{\epsilon_c},\ve{\epsilon_{\hat{r}}}) \leftarrow (\ve{\thege}(\thegpk, \ve{S}), \ve{\theze}(\thezpk, \ve{c}),$ \\
    \enskip \quad\quad\quad $\ve{\theze}(\thezpk, \ve{\hat{r}}))$ \\
    \enskip \quad\quad publish $y, \ve{\hat{\sigma}}:=(\ve{S}, \ve{c}, \ve{\hat{r}}), \ve{\epsilon_{\hat{\sigma}}}:=(\ve{\epsilon_S}, \ve{\epsilon_c},\ve{\epsilon_{\hat{r}}})$ \\ 
    \enskip $(\mixer{k})_{k \in [m]}$: \quad $\ve{\epsilon_{\mathsf{r}}} \leftarrow  \ve{\epsilon_{\hat{r}}}\ve{\epsilon_r}$ \\
    \enskip \mycomment{$(\vc{\epsilon}{\vc{S}{i}}, \vc{\epsilon}{\vc{c}{i}},\vc{\epsilon}{\vc{{\mathsf{r}}}{i}})$ encrypt $((f_1g_1^{\vc{v}{i}}h_1^{\vc{\hat{r}}{i}+\vc{r}{i}})^{\frac{1}{x+\vc{c}{i}}}, \vc{c}{i}, \vc{\hat{r}}{i}+\vc{r}{i})$ if $i \in I$} \\
    \enskip \mycomment{else $(f_1^0,\vc{c}{i}, \vc{\hat{r}}{i}+\vc{r}{i})$} \\
    
    \underline{\emph{Shuffling}} \\   
    \enskip $(\ve{\epsilon_S}', \ve{\epsilon_c}', \ve{\epsilon_{\mathsf{r}}}') \leftarrow \shuffle((\theg,\thez,\thez), (\thegpk, \thezpk, \thezpk), (\ve{\epsilon_S}, \ve{\epsilon_c}, \ve{\epsilon_{\mathsf{r}}}),$ \\    
    \enskip \quad \quad $\si{\mixer{1}}{\permfunc{1}}, \dots, \si{\mixer{m}}{\permfunc{m}})$ \mycomment{shuffle all three components}\\
	
	\underline{\emph{Homomorphic blinding}} \\	
    \enskip $(\prover{k})_{k \in [m]}$:  \\
    \enskip \quad $(\vcs{b}{\ve{S}}{k}, \vcs{b}{\ve{c}}{k}, \vcs{b}{\ve{r}}{k}) \xleftarrow{\$} (\mathbb{Z}^{n}_{q} \times \mathbb{Z}^{n}_{q} \times \mathbb{Z}^{n}_{q})$; $\vcs{\chi}{\ve{c}}{k}, \vcs{\chi}{\ve{r}}{k} \xleftarrow{\$} \mathbb{Z}_{q-1}^{n}$ \\
    \enskip \quad publish $(\vs{\epsilon}{k}_{\ve{bS}}, \vs{\epsilon}{k}_{\ve{bc}}, \vs{\epsilon}{k}_{\ve{br}}) \leftarrow (\ve\thege(\thegpk, g_1^{\vs{b}{k}_{\ve{S}}}),$ \\
    \enskip \quad \quad $\ve\theze(\thezpk, \vcs{b}{\ve{c}}{k} + q\vcs{\chi}{\ve{c}}{k}), 
        \ve\theze(\thezpk, \vcs{b}{\ve{r}}{k} + q\vcs{\chi}{\ve{r}}{k}))$ \\ 
    \enskip $(\prover{k})_{k \in [m]}$: $(\ve{\tilde{\epsilon}_S}', \ve{\tilde{\epsilon}_c}', \ve{\tilde{\epsilon}_{\mathsf{r}}}') \leftarrow (\ve{\epsilon_S}' \prod\limits_{k \in [m]} \vs{\epsilon}{k}_{\ve{bS}},  \ve{\epsilon_c}' \prod\limits_{k \in [m]} \vs{\epsilon}{k}_{\ve{bc}},\ve{\epsilon_{\mathsf{r}}}' \prod\limits_{k \in [m]} \vs{\epsilon}{k}_{\ve{br}})$\\
    
    \underline{\emph{Threshold decryption}} \\     
    \enskip $\ve{\tilde{S}}' \leftarrow \ve\thegd(\ve{\tilde{\epsilon}_S}', (\si{\mixer{k}}{\thegsk{k}})_{k \in [m]})$ \\
    \enskip $\ve{\tilde{c}}'' \leftarrow \ve\thezd(\ve{\tilde{\epsilon}_c}', (\si{\mixer{k}}{\thezsk{k}})_{k \in [m]}); \quad \ve{\tilde{c}}' \leftarrow \ve{\tilde{c}}'' \ve{\bmod{\ q}}$ \\
    \enskip $\ve{\tilde{r}}'' \leftarrow  \ve\thezd(\ve{\tilde{\epsilon}_{\mathsf{r}}}', (\si{\mixer{k}}{\thezsk{k}})_{k \in [m]}); \quad \ve{\tilde{r}}' \leftarrow \ve{\tilde{r}}''  \ve{\bmod{\ q}}$\\
    \enskip publish $\ve{\tilde{\sigma}}':=(\ve{\tilde{S}}', \ve{\tilde{c}}', \ve{\tilde{r}}')$ \\
    \enskip \mycomment{$\vc{\tilde{S}}{j}' = (f_1g_1^{\vc{v}{\pi(j)}}h_1^{\vc{\hat{r}}{\pi(j)}+\vc{r}{\pi(j)}})^{\frac{1}{x+\vc{c}{\pi(j)}}}g_1^{\vc{b_S}{j}}$ if $\pi(j)\in I$ else $g_1^{\vc{b_S}{j}}$} \\
    \enskip \mycomment{$\quad\enskip = (f_1g_1^{\vc{v}{j}'}h_1^{\vc{\hat{r}}{\pi(j)}+\vc{r}{\pi(j)}})^{\frac{1}{x+\vc{c}{\pi(j)}}}g_1^{\vc{b_S}{j}}$ if $\pi(j)\in I$ else $g_1^{\vc{b_S}{j}}$} \\
    \enskip \mycomment{$\vc{\tilde{c}}{j}' = \vc{c}{\pi(j)} + \vc{b_c}{j}; \enskip \vc{\tilde{r}}{j}' = \vc{\hat{r}}{\pi(j)}+\vc{r}{\pi(j)} + \vc{b_r}{j}$} \\
    \enskip \mycomment{where $\vc{b}{\vc{S}{j}} := \sum\limits_{k \in [m]}\vcs{b}{\vc{S}{j}}{k}; \vc{b}{\vc{c}{j}} := \sum\limits_{k \in [m]}\vcs{b}{\vc{c}{j}}{k}; \vc{b}{\vc{r}{j}} := \sum\limits_{k \in [m]}\vcs{b}{\vc{r}{j}}{k}$} \\
    \\

    \textbf{\underline{Stage 2:}} \\
    \enskip $\querier$: \quad $J^* \leftarrow \emptyset$ \\ 
    \enskip $(\mixer{k})_{k \in [m]},\querier$: \quad $\mathfrak{h}_1 \leftarrow e(h_1, f_2)^{-1}$; $\mathfrak{h}_2 \leftarrow e(g_1, f_2)^{-1}$; $\mathfrak{h}_3 \leftarrow f_T$ \\
    \enskip \textbf{for} $j \in J$: \\
    \enskip \enskip $(\prover{k})_{k \in [m]}$: \\ 
    \enskip \enskip \quad $\share{\delta}{k}_{0} \xleftarrow{\$} \Zq$; $(\share{\delta}{k}_{1},\share{\delta}{k}_{2}) \leftarrow (\mathsf{Mult}(\vcs{b}{\vc{S}{j}}{k}, \vcs{b}{\vc{c}{j}}{k}), \mathsf{Mult}(\share{\delta}{k}_{0}, \vcs{b}{\vc{c}{j}}{k}))$ \\ 
    \enskip \enskip \quad publish $\share{\mathfrak{z}}{k}_1 \leftarrow \mathfrak{h}_2^{\vcs{b_S}{j}{k}}\mathfrak{h}_3^{\share{\delta_{0}}{k}}$ \\
    \enskip \enskip $(\mixer{k})_{k \in [m]},\querier$:  $\mathfrak{z}_1 \leftarrow \prod_{k \in [m]} \share{\mathfrak{z}}{k}_1$;\quad $\mathfrak{z}_2 \leftarrow e(\vc{\tilde{S}}{j}', yf_2^{\vc{\tilde{c}}{j}'}) / e(f_1 g_1^{\vc{v}{j}'} h_1^{\vc{\tilde{r}}{j}'}, f_2)$ \\
    \enskip \quad\quad\quad\quad\quad\quad\quad $\mathfrak{g}_1 \leftarrow e(\vc{\tilde{S}}{j}', f_2)$; \quad $\mathfrak{g}_2 \leftarrow e(g_1, yf_2^{\vc{\tilde{c}}{j}'})$   \\ 
    \enskip \enskip $\so{\querier}{res} \leftarrow \mathsf{DPK}((\mathfrak{g}_1,\mathfrak{g}_2,\mathfrak{h}_1,\mathfrak{h}_2,\mathfrak{h}_3, \mathfrak{z}_1,\mathfrak{z}_2),p_{\mathsf{BBS+}_j},$ \\
    \enskip \enskip \enskip $ (\si{\mixer{k}}{\vcs{b}{\vc{S}{j}}{k}, \vcs{b}{\vc{c}{j}}{k}, \vcs{b}{\vc{r}{j}}{k}, \share{\delta}{k}_0, \share{\delta}{k}_1, \share{\delta}{k}_2})_{k \in [m]}, \si{\querier}{})$ \\
    \enskip \enskip \quad where $p_{\mathsf{BBS+}_j}:=(\mathfrak{z}_1 = \mathfrak{h}_2^{\vc{b}{\vc{S}{j}}}\mathfrak{h}_3^{\delta_{0}} \wedge 1_{\GT} = \mathfrak{z}_1^{\vc{-b}{\vc{c}{j}}} \mathfrak{h}_2^{\delta_{1}} \mathfrak{h}_3^{\delta_{2}} \wedge$ \\
	\enskip \enskip \quad \enskip $\mathfrak{z}_2 = \mathfrak{g}_1^{\vc{b}{\vc{c}{j}}} \mathfrak{g}_2 ^{\vc{b}{\vc{S}{j}}} \mathfrak{h}_1 ^ {\vc{b}{\vc{r}{j}}} \mathfrak{h}_2 ^ {\delta_{1}})$ \\	    
    \enskip \enskip $\querier:$ \enskip \textbf{if $res = 1$:} $J^* \leftarrow J^* \cup \{j\}$ \\
    \enskip \textbf{endfor} \\
    \enskip \textbf{output} $\so{\querier}{J^*}$ \\
    \hline
    \end{tabular}
}
\caption{{\small Protocol $\zkprsm$ (see summary in Figure \ref{fig:summary-zkprsm}).}}
\label{fig:distr-rev-set-membership}
\end{figure}

Figure \ref{fig:distr-rev-set-membership} shows our $\zkprsm$ protocol. Note that to obtain BBS+ signatures on committed values, knowledge of commitment openings must be first shown. For this, $\querier$ checks the NIZKs $\vc{\pokcomm}{i}$ uploaded by the senders for each $i \in I$. $\querier$ then sends valid BBS+ quasi-signatures $\vc{\hat{\sigma}}{i}$ on $\vc{\gamma}{i}$ for $(\vc{\gamma}{i})_{i \in I}$ and invalid ones for $(\vc{\gamma}{i})_{i \not\in I}$. It encrypts each component $(\vc{S}{i},\vc{c}{i},\vc{\hat{r}}{i})$ of $\vc{\hat{\sigma}}{i}$ independently, using $\theg$ for $\vc{S}{i}$ and $\thez$ for $\vc{c}{i}$ and $\vc{\hat{r}}{i}$, to create encrypted quasi-signatures $\vc{\epsilon}{\vc{\hat{\sigma}}{i}}:=(\vc{\epsilon}{\vc{S}{i}}, \vc{\epsilon}{\vc{c}{i}},\vc{\epsilon}{\vc{\hat{r}}{i}})$. Using $\vc{\epsilon}{\vc{\hat{r}}{i}}$ and the sender-uploaded encryptions $\vc{\epsilon}{\vc{r}{i}}$, encrypted BBS+ signatures on the committed values are derived by homomorphically adding the commitment randomness $\vc{r}{i}$ to the signature's $\vc{\hat{r}}{i}$ component. Thus, encrypted (valid and invalid) BBS+ signatures $\vc{\epsilon}{\vc{\sigma}{i}}:=(\vc{\epsilon}{\vc{S}{i}}, \vc{\epsilon}{\vc{c}{i}}, \vc{\epsilon}{\vc{\mathsf{r}}{i}}:=\vc{\epsilon}{\vc{\hat{r}}{i}}\vc{\epsilon}{\vc{r}{i}})_{i \in [n]}$ are obtained next to each commitment $\vc{\gamma}{i}$ in the input list.

To obtain blinded BBS+ signatures on plaintext values in the permuted list $\ve{v}'$, each component $(\ve{\epsilon_S}, \ve{\epsilon_c}, \ve{\epsilon_{\mathsf{r}}})$ of the BBS+ signature is re-encrypted individually using encryption schemes $(\theg, \thez, \thez)$ respectively and shuffled in the forward direction. The permuted and re-encrypted signatures $(\ve{\epsilon_S}', \ve{\epsilon_c}', \ve{\epsilon_{\mathsf{r}}}')$ are individually blinded to obtain $(\ve{\tilde{\epsilon_S}}', \ve{\tilde{\epsilon_c}}', \ve{\tilde{\epsilon_{\mathsf{r}}}}')$, which are individually threshold-decrypted to obtain blinded BBS+ signatures $\vc{\tilde{\sigma}}{j}':=(\vc{\tilde{S}}{j}',\vc{\tilde{c}}{j}',\vc{\tilde{r}}{j}')$ alongside $\vc{v}{j}'$.

Note that $\thez$ used for encrypting $\ve{\epsilon_c}', \ve{\epsilon_{\mathsf{r}}}'$ is homomorphic in group $\mathbb{Z}_N$, which induces addition modulo $N$ in the plaintext space, not modulo $q$. Thus, blinding by blinding factors drawn from $\Zq$ would not be perfectly hiding. To circumvent this issue, we follow an approach similar to \citep{gennaro-threshold,gennaro-threshold2}. First, we ensure that $N$ is much larger than $q$ so that all additions remain integer additions and do not wrap around $N$ (this is anyway the case since $N$ is usually a 2048-bit modulus, $q$ is of $254$ bits and we perform a small number of homomorphic additions per ciphertext). Second, we pad blinding factors $\vc{b_c}{j},\vc{b_r}{j} \in \Zq$ for $\vc{\epsilon}{\vc{c}{j}}', \vc{\epsilon}{\vc{\mathsf{r}}{j}}'$ by much larger offsets $q\vc{\chi}{\vc{c}{j}}, q\vc{\chi}{\vc{r}{j}} \in \mathbb{Z}_{q(q-1)}$ so that the padded blinding factors $\vc{b_c}{j} + q\vc{\chi}{\vc{c}{j}}, \vc{b_r}{j} + q\vc{\chi}{\vc{r}{j}}$ are identically distributed to uniform samples from $\mathbb{Z}_{q^2}$ and provide an almost perfect integer blinding for the messages $\vc{c}{\pi(j)} \in \Zq$ and $\vc{\hat{r}}{\pi(j)}+\vc{r}{\pi(j)} \in \mathbb{Z}_{2q}$. The offsets are removed by reducing decryptions $\vc{\tilde{c}}{j}'',\vc{\tilde{r}}{j}''$ of $\vc{\epsilon}{\vc{c}{j}}', \vc{\epsilon}{\vc{\mathsf{r}}{j}}'$ modulo $q$.

In stage 2, for each $j \in J$, a DPK is given that proves that the mix-servers know shares of $\vc{b}{\vc{S}{j}}, \vc{b}{\vc{c}{j}}, \vc{b}{\vc{r}{j}}$ such that $(\vc{\tilde{S}}{j}'g_1^{-\vc{b}{\vc{S}{j}}},$ $\vc{\tilde{c}}{j}' - \vc{b}{\vc{c}{j}} \bmod q,$ $\vc{\tilde{r}}{j}' - \vc{b}{\vc{r}{j}} \bmod q)$ is a valid BBS+ signature on message $\vc{v}{j}'$ under $\querier$'s verification key. The actual DPK used in Figure \ref{fig:distr-rev-set-membership} is designed to follow the format of Section \ref{sss:distr-pok}. Since each mix-server knows shares of $\vc{b_S}{j}, \vc{b_c}{j}, \vc{b_r}{j}, \delta_{0}$ and gets shares of $\delta_{1}$ and $\delta_{2}$ via Beaver's $\mathsf{Mult}$ algorithm, these DPKs can be given efficiently.

\subsubsection{HBC to malicious security}
\label{sec:hbc-to-malicious}

Now we highlight steps to maintain secrecy even when the corrupted parties maliciously deviate from the protocol (see Appendix \ref{honest-but-curious-to-full-zk} for detailed steps). Any party that publishes an encryption - senders during $\enc$, the querier, or the mix-servers during $\mix$ or $\btrin$/$\btrout$ protocols - must provide proofs of knowledge of underlying plaintexts. This is to avoid attacks where re-encryptions of honest senders' ciphertexts are published to get them decrypted. Moreover, before participating in threshold decryption protocols, honest mix-servers must verify that all encryptions are correctly created and consistently shuffled using the same permutation in $\zkprsm$ as in $\mix$ and the inverse permutation in $\zkpsm$ (say, using techniques of \citep{commitment-consistent-proof-shuffle,terelius-restricted-shuffles}). Mix-servers should also verify that the querier's signatures are valid for all elements in the requested set and invalid for the complement set. The DPKs can be converted using Fiat-Shamir heuristic \citep{fiat-shamir}.

\section{Security analysis}
\label{analysis}

\begin{theorem}[Completeness]
	\label{thm:completeness}
    Let $\Pi_{\mathsf{TM}}$ be the protocol of Figure \ref{fig:construction}. $\Pi_{\mathsf{TM}}$ is complete (Definition \ref{def:compl}). (Proof in Appendix \ref{pf:completeness}).
\end{theorem}

\begin{theorem}[Soundness]
	\label{thm:soundness}
	Under the DL assumption in $\Gone$ and $n$-SDH assumption in $(\Gone, \Gtwo)$ \citep{boneh-boyen}, $\Pi_{\mathsf{TM}}$ is sound (Definition \ref{def:snd}).
\end{theorem}
\emph{Proof sketch (full proof in Appendix \ref{pf:soundness}):} If an adversary won $\snd$ then $\querier$ must have output either 1) a wrong $\vc{c}{I^*}$ in a $\btrin(I,J)$ query; or 2) a wrong $\vc{v}{J^*}'$ in a $\btrout(I,J)$ query. Case 1 implies that either $a)$ $\vc{c}{I^*}$ included a $\vc{c}{i} \in \vc{c}{I}$ not encrypting a plaintext in $\vc{v}{J}'$ or $b)$ it excluded a $\vc{c}{i} \in \vc{c}{I}$ encrypting a plaintext in $\vc{v}{J}'$. Case 1a directly reduces to breaking the soundness of DB-SM for $\vc{\gamma}{i}$ against set $\vc{v}{J}'$. Further, since each $\vc{c}{i} \in \vc{c}{I} \setminus \vc{c}{I^*}$ must be in $\vc{c}{I_{\mathsf{c}}^*}$ for $\querier$ to not abort, case 1b also reduces to breaking the soundness of DB-SM (because if $\vc{c}{i} \not\in \vc{c}{I^*}$ encrypts a plaintext in $\vc{v}{J}'$, it does not encrypt a plaintext in $\vc{v}{[n] \setminus J}'$ due to distinctness of $\vc{v}{j}'$s, but $\vc{c}{i} \in \vc{c}{I_{\mathsf{c}}^*}$). Case 2 similarly reduces to breaking the soundness of DB-RSM.
	
\emph{Soundness of $\zkpsm$ for $I,J$:} The DPK proves that $(\mixer{k})_{k \in [m]}$ know blinding factors to unblind $\vc{\tilde{\sigma}}{i}$ to a valid BB signature on the value $\vc{v}{i}$ committed by $\vc{\gamma}{i}$ under $\querier$'s fresh public key. Since $\querier$ issued valid signatures only for plaintexts in $\vc{v}{J}'$, passing the DPK if $\vc{v}{i} \not\in \vc{v}{J}'$ requires $(\mixer{k})_{k \in [m]}$ to forge a BB signature under $\querier$'s public key, which is hard under the stated assumptions.

\emph{Soundness of $\zkprsm$ for $I,J$:} The DPK proves that $(\mixer{k})_{k \in [m]}$ know blinding factors to unblind $\vc{\tilde{\sigma}}{j}'$ to a valid BBS+ signature on $\vc{v}{j}'$ under $\querier$'s fresh public key. Since $\querier$ issued valid signatures only for commitments in $\vc{\gamma}{I}$, deriving a signature on a plaintext not committed in any commitment in $\vc{\gamma}{I}$ reduces to breaking the soundness of the BBS+ scheme for obtaining signatures on committed values, which is hard under the stated assumptions.

\begin{theorem}[Secrecy]
	\label{thm:secrecy}
	Under the IND-CPA security of $\pke$, the DDH assumption in $\Gone$ and the DCR assumption \citep{paillier}, $\Pi_{\mathsf{TM}}$ protects secrecy (Definition \ref{def:secr}) against HBC adversaries in the random oracle model.
\end{theorem}
\emph{Proof sketch (full proof in Appendix \ref{pf:secrecy}):} We need to show that for any pair of values $v_{i_0},v_{i_1}$, no PPT adversary controlling $(\mixer{k})_{k \not\eq k^*}$, $(S_{i})_{i \not\in \{i_0,i_1\}}$ and $\querier$ can distinguish between world $0$ where $S_{i_0}$ sends $v_{i_0}$ and $S_{i_1}$ sends $v_{i_1}$ and world $1$ where this order is reversed, if the assert conditions in all $\otrin$ and $\otrout$ calls of $\secr$ are respected (and all adversarial parties are honest-but-curious). We do this by simplifying world $b$ for each $b \in \{0,1\}$ by the following sequence of indistinguishability arguments:
\begin{itemize}[leftmargin=*]
	\item[-] For each $i \in \{i_0,i_1\}$, NIZK $\vc{\pokcomm}{i}$ can be simulated, shares $\vcs{v}{i}{k},\vcs{r}{i}{k}$ given to $(\mixer{k})_{k \neq k^*}$ can be replaced by random elements drawn from $\Zq$, and encryptions $\vcs{\mathsf{ev}}{i}{k^*}$, $\vcs{\mathsf{er}}{i}{k^*}$ can be replaced by encryptions of $0$ by the IND-CPA security of $\pke$.
	\item[-] Threshold decryption protocols for $\theg$ and $\thez$ do not reveal any information beyond the decryption output, which can be obtained by permuting the list of input values, where $S_{i_0}$ and $S_{i_1}$'s input values are $v_{i_b}$ and $v_{i_{1-b}}$ respectively. With the decryption oracles now eliminated, all $\theg$ and $\thez$ encryptions/re-encryptions can be replaced by encryptions of dummy values, by the IND-CPA security of $\theg$ and $\thez$ under the DDH and DCR assumptions respectively.
	\item[-] The only information leaked during a DPK for $S_{i_0}/S_{i_1}$'s commitment/plaintext is whether the DPK passed or not, but the assert conditions ensure that this information is the same in both the worlds. Thus, these DPKs can be simulated by a ZK simulator that does not know which specific world it is in.
	\item[-] For $i \in \{i_0,i_1\}$, $\vcs{\mathsf{ev}}{i}{k^*}$, $\vcs{\mathsf{er}}{i}{k^*}$ can be replaced by encryptions of $0$ since the encrypted values $\vcs{v}{i}{k^*},\vcs{r}{i}{k^*}$ are not used anywhere anymore (they were used in the DPKs/encryptions earlier). Similarly, $\vc{\gamma}{i_0},\vc{\gamma}{i_1}$ can now be replaced with commitments of $0$.
	\item[-] The blinded signatures corresponding to $S_{i_0}$ or $S_{i_1}$'s commitments or plaintexts during $\zkpsm$ and $\zkprsm$ can now be replaced with random group elements because they are not used anywhere anymore and the blinding factors chosen by $\mixer{k^*}$ are random (note: for the $c,r$ components of the BBS+ signature, this holds because of flooding with large blinding factors). 
\end{itemize} 
The two worlds obtained now are indistinguishable because now only the mixnet output list $\ve{v}'$ depends on $v_{i_0}$ and $v_{i_1}$ and this list is identically distributed in the two worlds because of the uniformly chosen permutation $\permfunc{k^*}$ by $\mixer{k^*}$.

\begin{theorem}[Output secrecy]
	\label{thm:output-secrecy}
	Under the same assumptions as Theorem \ref{thm:secrecy}, $\Pi_{\mathsf{TM}}$ protects output secrecy (Definition \ref{def:secr-mixers}). (Proof in Appendix \ref{pf:output-secrecy}).
\end{theorem}

In Appendix \ref{pf:secrecy-general}, we also sketch a proof that after applying the additional steps mentioned in Section \ref{sec:hbc-to-malicious} and Appendix \ref{honest-but-curious-to-full-zk}, our construction protects secrecy against general malicious adversaries. 

\subsection{Privacy risk analysis of query outputs} 

Theorem \ref{thm:secrecy} guarantees that a traceable mixnet does not reveal any information beyond the output of the $\btrin/\btrout$ queries. To evaluate the privacy risk impact of these outputs themselves, we provide a mechanism in Appendix \ref{privacy-risk-analysis} to statically analyse information leaked by a given set of queries. Specifically, given a set $Q$ of proposed TraceIn/TraceOut queries in an application, the mechanism outputs information potentially leaked by them: $a)$ for each $i \in [n]$, the smallest set $J_{\mathsf{min}_i}$ representing potential plaintexts that $\vc{c}{i}$ might encrypt and $b)$ for each $j \in [n]$, the smallest set $I_{\mathsf{min}_j}$ representing potential ciphertexts that $\vc{v}{j}'$ might decrypt from ($J_{\mathsf{min}_i}=[n]$ denotes that queries in $Q$ reveal no additional information about $\vc{c}{i}$; likewise for $I_{\mathsf{min}_j}$). The $J_{\mathsf{min}}/I_{\mathsf{min}}$ information can then directly be used to analyse application-level security.

\section{Implementation and benchmarks}
\label{practicalities}

We implemented a proof-of-concept for our traceable mixnet construction, with the primary goal of evaluating its runtime query performance. Given this focus, we mainly implemented the $\zkpsm$ and $\zkprsm$ protocols. This allows us to directly estimate the total time for $\btrin$ or $\btrout$ queries as that of two $\zkpsm$ or $\zkprsm$ invocations. However, an optimisation where the $\zkpsm$ calls for $J$ and $[n]\setminus J$ (similarly for $\zkprsm$) are combined to a single call as follows considerably improves this estimate: $\querier$ generates signature key pairs $x,y$ for $J$ and $x_{\mathsf{c}}, y_{\mathsf{c}}$ for $[n]\setminus J$ and sends signatures using key $x_{\mathsf{c}}$ for $j \not\in J$ instead of invalid signatures, thus avoiding extraneous shuffling/decryption of invalid signatures.

We used standard threshold ElGamal encryption \citep{threshold-cryptography} for $\theg$ and Damg{\aa}rd et al.'s \cite{dj-generalisation} optimised threshold Paillier encryption for $\thez$.\footnote{We slightly simplified \cite{dj-generalisation} by skipping the factorial trick used for general $t$-out-of-$m$ threshold cryptography in RSA groups, focusing on the $m$-out-of-$m$ case.} For key generation, we implemented a simplified protocol where a trusted dealer distributes key shares to the mix-servers. Secure distributed key generation is trivial for $\theg$ but requires a special protocol \citep{damgaard-secure-keygen} for $\thez$. However, our simplification is justified as key generation is a one-time setup step that can be executed ahead of time and does not affect runtime query performance. 

We also implemented all the steps mentioned in Section \ref{sec:hbc-to-malicious} and Appendix \ref{honest-but-curious-to-full-zk} to evaluate performance in the realistic malicious setting. We used the following techniques to optimise these steps: 1) standard $\Sigma$-protocol techniques to let senders efficiently prove knowledge of their uploaded ciphertexts and mix-servers prove knowledge of blinding factors during homomorphic blinding; 2) batch-verification techniques of \citep{bellare-batch-verification,fast-batch-sigverif} to let each mix-server efficiently verify the querier's signatures/quasi-signatures and the correctness of other mix-servers' decryption shares during threshold decryption; and 3) \emph{permutation commitment}-based techniques of \citep{commitment-consistent-proof-shuffle,terelius-restricted-shuffles} to let each mix-server efficiently prove that they created their shuffles consistently during $\mix$, $\zkpsm$ and $\zkprsm$.

We implemented all sender proofs of knowledge required to achieve our soundness and secrecy requirements. However, we did not implement proofs of well-formedness of input ciphertexts (Section \ref{formalism}): proofs that $\vc{\gamma}{i}$ commit the same value as encrypted by $\vc{\epsilon}{i}$ and that $\vc{\epsilon}{i}$, $\vc{\epsilon}{r_i}$ encrypt values in the range $[0,q]$.\footnote{Senders do not need to prove distinctness of their encrypted messages because duplicate values can be detected at the mixnet output and the offending senders could be identified (by, say, a $\btrin$ query) on demand.} This is primarily due to our focus on runtime query performance and because these costs are incurred by individual senders and mix-servers as and when senders send their data. Our proofs of knowledge incur $\sim$0.15 seconds cost per-sender for both the sender and the mix-servers. With standard $\Sigma$-protocol techniques and FO commitments \citep{focomm} for proving equality of committed and encrypted values and \citep{boudot-rangeproof,lindell-rangeproof-impl} for range proofs, the total time is expected to remain <1 s per sender.

Finally, we did not implement the authenticated broadcast channel. Since our uploaded datasets are moderate in size and the round complexity is small, we do not expect this to be a bottleneck.

Our implementation \citep{tm-impl} is based on the Charm library \citep{charm-crypto} with the PBC backend \citep{pbc} for pairing operations. We chose the BN254 curve \citep{barreto-naehrig,johnson2021bn254} to instantiate pairing groups $(\Gone,$ $\Gtwo,$ $\GT)$, which gives a group order $q$ of 254 bits.\footnote{With NFS attacks \citep{bn254-attacks}, the security of BN254 curves has dropped to ~100-110 bits \citep{bn254-secdrop,bls12-381}. We are limited to BN254 because of our chosen Charm library, but we estimate that the overall performance hit in per-mix-server time on switching to a more secure curve BLS12-381 \citep{bls12-381} is <1.3x, given that BLS12-381 operations are <2x slower than BN254 \citep{comparison-bls12-381-bn254} and that curve operations predominantly only affect our stage 2 DPKs.} We attempted to minimise the number of pairing operations wherever possible and used pre-computation of powers of fixed bases to speed-up exponentiation operations.

\begin{figure}
	\centering
	\scalebox{0.8}{
			\begin{tabular}{l}
				\begin{tabular}{|@{ }l@{ }|c@{  }|c@{  }|c@{  }|c@{  }|c@{  }|c@{  }|}
					\hline
					\multicolumn{7}{|c|}{\textbf{$\mixer{k}$ and $\querier$ times (sec) in $\zkpsm$ and $\zkprsm$ for different $n$ and $m$}} \\
					\hline
												& \multicolumn{3}{c|}{$m=2$}    & \multicolumn{3}{c|}{$m=4$}    \\
					\hline							  
												& $n=10^3$ & $n=10^4$ & $n=10^5$ & $n=10^3$ & $n=10^4$ & $n=10^5$ \\
					\hline
					DB-SM - $\mixer{k}$         & 35       & 340      & 3400     & 49       & 490      & 4916     \\
					DB-SM (HBC) - $\mixer{k}$   & 23       & 225      & 2268     & 23       & 227      & 2274     \\     
					DB-SM - $\querier$          & 20       & 190      & 2000     & 20       & 200      & 2000     \\
					collab-zkSNARK-SM           & 4120     & 42960    & 446800   & 4120     & 42960    & 446800   \\
					
					\hline 
					DB-RSM - $\mixer{k}$        & 168      & 1600     & 20000    & 240      & 2400     & 32000    \\
					DB-RSM (HBC) - $\mixer{k}$  & 122      & 1188     & 11981    & 129      & 1291     & 12949    \\     
					DB-RSM - $\querier$         & 62       & 610      & 6200     & 62       & 620      & 6200     \\
					collab-zkSNARK-RSM          & 4160     & 43100    & 448920   & 4160     & 43100    & 448920   \\
				\hline
				\end{tabular} \\

				\\

				\begin{tabular}{ |@{ }m{8.8cm} l@{ }|  } 
					\hline
					\multicolumn{2}{|c|}{\textbf{Detailed breakdown for $n=10^4$ and $m=4$}} \\
					\hline
					- Size of $n$ input ciphertexts & 200 MB \\
					\hline
					- ($\mixer{k}$) Mixing (Fig. \ref{fig:construction}): & 343 s \\
					\hline
					$\zkpsm$ (Fig. \ref{fig:distr-set-membership}): &  \\
					- ($\querier$) Generating $n$ BB signatures/encryptions & 8.3 s \\
					- ($\mixer{k}$) Verifying $n$ BB signatures/encryptions & 15 s \\
					- ($\mixer{k}$) Re-encryption of encrypted signatures & 7.3 s \\
					- ($\mixer{k}$) Proof-of-shuffle of encrypted signatures & 139 s \\
					- ($\mixer{k}$) Homomorphic blinding of encrypted signatures & 130 s \\
					- ($\mixer{k}$) Threshold decryption of encrypted signatures & 22 s \\
					- ($\mixer{k}$) Generating $n$ DPK proofs for $p_{\mathsf{BB}}$ & 170 s \\
					- ($\querier$) Verifying $n$ DPK proofs for $p_{\mathsf{BB}}$ & 190 s \\ 
					- Size of $n$ BB signatures  & 0.3 MB \\
					- Size of $n$ DPK proofs for $p_{\mathsf{BB}}$ & 3.8 MB \\
					\hline
					$\zkprsm$ (Fig. \ref{fig:distr-rev-set-membership}): & \\
					- ($\querier$) Verifying $n$ PoKs of commitments & 7.9 s \\
					- ($\querier$) Generating $n$ BBS+ quasi-signatures & 23 s \\
					- ($\mixer{k}$) Verifying $n$ BBS+ quasi-signatures/encryptions & 27 s \\
					- ($\mixer{k}$) Re-encryption of encrypted signatures & 16 s \\
					- ($\mixer{k}$) Proof-of-shuffle of encrypted signatures & 349 s \\
					- ($\mixer{k}$) Homomorphic blinding of encrypted signatures & 880 s \\
					- ($\mixer{k}$) Threshold decryption of encrypted signatures & 441 s \\
					- ($\mixer{k}$) Generating $n$ DPK proofs for $p_{\mathsf{BBS}+}$ & 673 s \\
					- ($\querier$) Verifying $n$ DPK proofs for $p_{\mathsf{BBS}+}$ & 590 s \\
					- Size of $n$ BBS+ signatures & 0.9 MB \\
					- Size of $n$ DPK proofs for $p_{\mathsf{BBS}+}$ & 11 MB \\
					\hline
				\end{tabular}
			\end{tabular}
	}
	\caption{{\small Performance of $\zkpsm$ and $\zkprsm$ ($n$: number of input ciphertexts, $m$: number of mix-servers). $\mixer{k}$ and $\querier$ denote per-mix-server and querier times respectively; collab-zkSNARK denotes estimated per-prover times in collaborative zkSNARKs \citep{collaborative-zksnarks}.}}
	\label{fig:benchmarks}
\end{figure}

We ran all our benchmarks on an Intel(R) Xeon(R) W-1270 CPU @ 3.40GHz with 64 GB RAM on a single core. Figure \ref{fig:benchmarks} shows the performance of our $\zkpsm/\zkprsm$ implementation in the worst case when $I=J=[n]$, with the overall mix-server and querier times for different $n$ and $m$ at the top and the detailed breakdown for $n=10000$ ciphertexts and $m=4$ mix-servers at the bottom. All reported values are averages over 3 runs. The deviation from the average in any run is <1.5\%. We report per-mix-server times, which accurately capture real-world latencies as the heavy operations like proofs-of-shuffle, homomorphic blinding, threshold decryption and stage 2 DPKs can be run in parallel by the mix-servers. The only sequential operation is re-encryption, but it is a negligible fraction of other steps. We also report timings for the HBC case to highlight the overhead introduced by the malicious security steps.

Our ZKPs are practical for offline batch processing tasks, finishing within an hour and requiring moderate amount of data to be published for $10000$ ciphertexts. They scale linearly with $n$. The scale-up with $m$ is constant for the HBC mix-server time and for the querier, but in the malicious case, each mix-server needs to verify other mix-servers' output, which leads to a $\sim$1.4x increase from $m=2$ to $m=4$. Our main bottlenecks are expensive pairing and exponentiation computations in our DPKs and expensive Paillier operations in $\zkprsm$. The additional proofs for malicious security add an overhead of roughly 1.5-2x over the HBC case where these steps are skipped. 

There exists a high degree of task parallelism in our construction, since DPKs in stage 2 are independent of each other and stage 1 operations incur at most constant communication overhead. Thus, we expect significant speedups if each mix-server and the querier are given multiple cores. With a moderate parallel cluster of 100 nodes, recovery for an election with $10^6$ votes can thus be performed within a few hours. Further, if the signer in stage 1 of our ZKPs could be a separate trusted entity different than the querier and pre-sign all set entries, then our ZKPs become completely noninteractive and only incur stage 2 costs for verification. \\

\begin{paragraph}{Comparison}
	The only technique comparable to our distributed setting is collaborative zkSNARKs \citep{collaborative-zksnarks}. We indirectly estimate our performance against them by employing a thumbrule given by \citep{collaborative-zksnarks} that the per-prover time in a collaborative zkSNARK, assuming each prover already has a share of the SNARK witness, is $\sim$2x the prover time in the corresponding single-prover zkSNARK. Thus, we implemented zkSNARKs for $\rho_{\mathsf{SM}\text{-}\mathsf{Acc}}$ and $\rho_{\mathsf{RSM}\text{-}\mathsf{Acc}}$ via Merkle accumulators (see Section \ref{ss:review-sm}). We used the ZoKrates toolchain \citep{zokrates} and the Groth16 proof system \citep{groth16}. We used the Baby Jubjub curve \citep{babyjubjub}, which has similar order as BN254 and allows efficient computation of Merkle hashes for commitments for $\rho_{\mathsf{RSM}\text{-}\mathsf{Acc}}$. For creating Merkle hashes, we used the Poseidon hash function \citep{poseidon}. 

	Figure \ref{fig:benchmarks} (top) also shows the per-prover times in a collaborative zkSNARK approach estimated as above (averaged over 3 runs with deviation <1.1\%). We find the prover time for one $\rho_{\mathsf{SM}\text{-}\mathsf{Acc}}$ or $\rho_{\mathsf{RSM}\text{-}\mathsf{Acc}}$ proof against a set of size 10000 is $\sim$2.15 s, which when scaled to $10000$ commitments takes $\sim$21500 s. From this, the per-prover time in collaborative zkSNARKs is estimated to be $\sim$43000 s, as shown in the column for $n=10^4$ (for any $m$). This estimate is conservative as it does not count the time taken to securely distribute shares of the SNARK witness among the collaborating provers. This makes our $\zkpsm$ and $\zkprsm$ proofs $\sim$86x and $\sim$18x faster than collaborative zkSNARKs. We note that our verification times (200 s and 620 s for $n=10000$) are slower than zkSNARKs' $\sim$50 s, but the dominant prover times in zkSNARKs imply that our techniques still bring drastic overall improvements.

	We also ran Benarroch et al.'s \citep{modular-zkp-sm} official implementation \citep{cpsnarks-set-impl}, which proves $\rho_{\mathsf{SM}\text{-}\mathsf{Acc}}$ for a single prover (but not $\rho_{\mathsf{RSM}\text{-}\mathsf{Acc}}$). This takes $\sim$2200 s for $10000$ $\rho_{\mathsf{SM}\text{-}\mathsf{Acc}}$ proofs, excluding the $O(n^2)$ time taken to generate the RSA accumulator witnesses (see Section \ref{ss:review-sm}).	
\end{paragraph}

\section{Conclusion}

We introduced and formalised the notion of traceable mixnets, extending traditional mixnets to provably answer useful subset queries in zero knowledge. We also proposed a traceable mixnet construction using novel distributed ZKPs of set membership and reverse set membership, which are useful in other settings too. We implemented these ZKPs and showed that they are significantly faster than the state-of-the-art techniques. Nevertheless, our current implementation is practical only for offline batch processing such as recovery in elections. Constructing traceable mixnets for real-time privacy applications is a challenging open problem.

\begin{acks}
	We wish to thank Rohit Vaish, Kabir Tomer and Mahesh Sreekumar Rajasree for their helpful discussions and comments, and Aarav Varshney for assistance in setting up the benchmarks. Our gratitude also extends to the anonymous reviewers whose suggestions significantly improved the manuscript's presentation.

	This research was partially supported by the Pankaj Gupta Chair in Privacy and Decentralisation. Prashant Agrawal received additional support from the Pankaj Jalote Doctoral Grant, and Abhinav Nakarmi was supported by a research grant from the MPhasis F1 Foundation.
\end{acks}

\bibliography{traceable-mixnets}
\bibliographystyle{ACM-Reference-Format}

\appendix

\section{Single prover reverse set membership}
\label{appendix:rsm-single}

Figure \ref{fig:rev-set-membership} shows the ZKP of reverse set membership in the single prover case. Note that in this protocol, the prover knows commitment openings for each $\gamma \in \Phi$ (not only for the commitment committing $v$).

\begin{figure}[t]
	
	\centering
	
	\scalebox{0.8}{
		
		\begin{tabular}{|l|}
			\hline    
			\multicolumn{1}{|l|}{\textbf{Participants:} Prover $\mathcal{P}$, Verifier $\mathcal{V}$} \\			
			\multicolumn{1}{|l|}{\textbf{Common~input:} $\Phi \in 2^{\Gone} \text{(set of commitments)}, v \in \Zq$} \\
			\multicolumn{1}{|l|}{\textbf{$\mathcal{P}$'s input:}~~For each $\gamma_v \in \Phi$ ($\Phi$ indexed by $v$): $(v,r_v)$ s.t. $\gamma_v = g_1^{v}h_1^{r_{v}}$} \\
			\multicolumn{1}{|l|}{\textbf{$\mathcal{V}$'s output:} $1$ if $\mathcal{V}$ accepts else $0$} \\
			\hline
			
			\textbf{\underline{Stage 1:}} \\
			\enskip $\mathcal{P}$: \textbf{for each $\gamma \in \Phi$:}  $\rho_{\gamma} \leftarrow \nizkpk{v,r}{\gamma = g_1^{v}h_1^{r}}$ \\
			\enskip \quad \enskip send $\{\rho_{\gamma} \mid \gamma \in \Phi \}$ to $\verifier$ \\
			\enskip $\mathcal{V}$: \\
			\enskip \quad \textbf{for each} $\gamma \in \Phi$: \\
			\enskip \quad \quad \textbf{abort} \textbf{if} $\nizkver(\gamma, \rho_{\gamma}) \neq 1$\\ 
			\enskip \quad \quad $x \xleftarrow{\$} \Zq$; $y \leftarrow f_2^x$ \\
			\enskip \quad \quad $c_{\gamma} \xleftarrow{\$} \Zq; \hat{r}_{\gamma} \xleftarrow{\$} \Zq$; $S_{\gamma} \leftarrow (f_1 h_1^{\hat{r}_{\gamma}} \gamma)^{\frac{1}{x+{c}_{\gamma}}}$ \\
			\enskip \quad \textbf{endfor} \\ 
			\enskip \quad send $y, \hat{\sigma}_{\gamma}:=(S_{\gamma}, c_{\gamma}, \hat{r}_{\gamma})$ to $\mathcal{P}$\\ 
			\enskip $\mathcal{P}$: \\
			\enskip \quad Find $\gamma_v \in \Phi$ s.t. $\gamma_v = g_1^vh_1^{r_{v}}$ \mycomment{$O(1)$ look-up if $\Phi$ is indexed by $v$} \\
			\enskip \quad $\hat{\sigma}_v := (S_v, c_v, \hat{r}_v) := \hat{\sigma}_{\gamma_v}$ \\
			\enskip \quad $(b_S, b_c, b_r) \xleftarrow{\$} (\Zq \times \Zq \times \Zq)$ \\
			\enskip \quad $(\tilde{S}_v, \tilde{c}_v, \tilde{r}_v) \leftarrow (S_vg_1^{b_S}, c_v + b_c \bmod q, \hat{r}_{v} + r_v + b_r \bmod q)$ \\
			\enskip \quad send $\tilde{\sigma}_v:=(\tilde{S}_v, \tilde{c}_v, \tilde{r}_v)$ to $\mathcal{V}$ \\
			\textbf{Stage 2:} \\
			\mycomment{$\mathcal{P}$ proves (in ZK) knowledge of a BBS+ signature} \\
			\mycomment{$(\tilde{S}_vg_1^{-b_S}, \tilde{c}_v-b_c \bmod q, \tilde{r}_v-b_r \bmod q)$  on $v$, i.e., } \\
			\underline{\mycomment{$e(\tilde{S}_vg_1^{-b_S}, yf_2^{\tilde{c}_v-b_c}) = e(f_1 g_1^{v} h_1^{\tilde{r}_v-b_r}, f_2)$ (Sec. \ref{sss:bbsplus-sig}):}} \\
			\enskip $\mathcal{P}, \mathcal{V}:$\\
			\enskip \quad $\mathfrak{h}_1 \leftarrow e(h_1, f_2)^{-1}$; $\mathfrak{h}_2 \leftarrow e(g_1, f_2)^{-1}$; $\mathfrak{h}_3 \leftarrow f_T$  \\
			\enskip \quad $\mathfrak{z}_2 \leftarrow e(\tilde{S}_v, yf_2^{\tilde{c}_v}) / e(f_1 g_1^{v} h_1^{\tilde{r}_v}, f_2)$ \\
			\enskip \quad $\mathfrak{g}_1 \leftarrow e(\tilde{S}_v, f_2)$; $\mathfrak{g}_2 \leftarrow e(g_1, yf_2^{\tilde{c}_v})$ \\
			\enskip $\mathcal{P}$: \\
			\enskip \quad $\delta_{0} \xleftarrow{\$} \Zq$; $\delta_{1} \leftarrow b_Sb_c$; $\delta_{2} \leftarrow \delta_{0}b_c$ \\
			\enskip \quad send $\mathfrak{z}_1:=\mathfrak{h}_2^{b_S}\mathfrak{h}_3^{\delta_{0}}$ to $\mathcal{V}$ \\
			\enskip $\mathcal{P}, \mathcal{V}:$\\
			\enskip \quad  $res \leftarrow \mathsf{PK}\{(b_S, b_c, b_r, {\delta_{0}}, \delta_{1}, \delta_{2}): (\mathfrak{z}_1 = \mathfrak{h}_2^{b_S}\mathfrak{h}_3^{\delta_{0}}) \wedge$\\
			\enskip \quad \quad \quad \quad $ (1_{\GT} = \mathfrak{z}^{-b_c}_1 \mathfrak{h}_2^{\delta_{1}} \mathfrak{h}_3^{\delta_{2}}) \wedge (\mathfrak{z}_2 = \mathfrak{g}_1^{b_c} \mathfrak{g}_2 ^{b_S} \mathfrak{h}_1 ^ {b_r} \mathfrak{h}_2 ^ {\delta_{1}}) \}$ \\
			\enskip $\mathcal{V}:$\\
			\enskip \quad \textbf{output} $res$ \\
			\hline
		\end{tabular}
		
	}
	
	\caption{ZKP of reverse set membership $\rho_{\mathsf{RSM}}(\Phi, v):=\pk{r}{\gamma = g_1^vh_1^r \wedge \gamma \in \Phi}$, if the prover knows openings for each $\gamma \in \Phi$.}
	
	\label{fig:rev-set-membership}
	
\end{figure}

\section{From honest-but-curious to malicious model}
\label{honest-but-curious-to-full-zk}

In this section, we mention steps required to derive secrecy (Definition \ref{def:secr}) in the general case when the adversary allows the corrupted parties (all-but-two senders, all-but-one mix-servers and the querier) to deviate from the protocol:
\begin{itemize}[leftmargin=*]
	\item[-] Each sender $S_i$ must attach NIZK proofs of knowledge $(\vc{\rho}{\epsilon_i}$, $\vc{\rho}{\epsilon_{r_i}}$, $(\vs{\rho}{k}_{\mathsf{ev}_i},\vs{\rho}{k}_{\mathsf{er}_i})_{k \in [m]})$ of plaintexts encrypted by encryptions $(\vc{\epsilon}{i}$, $\vc{\epsilon}{r_i}$, $(\vs{\mathsf{ev}}{k}_{i}, \vs{\mathsf{er}}{k}_{i})_{k \in [m]})$ sent by it. Each $(\mixer{k})_{k \in [m]}$ must verify these proofs before processing anything.
	
	\item[-] For encryptions $\ve{\epsilon}_{\ve{\sigma}}'$ in $\zkpsm$ and $\ve{\epsilon_{S}},\ve{\epsilon_{c}},\ve{\epsilon_{\hat{r}}}$ in $\zkprsm$, the querier must publish their randomnesses and each $(\mixer{k})_{k \in [m]}$ must verify that they were created correctly. Each $(\mixer{k})_{k \in [m]}$ must also verify that $\querier$ gave valid signatures/quasi-signatures for each element in the requested set and invalid ones for its complement. Note that this also involves verifying that $\ve{c},\ve{\hat{r}}$ in $\zkprsm$ contain only elements in the range $[0,q]$.

	\item[-] Each $(\mixer{k})_{k \in [m]}$ must provide proofs of correct shuffle in all the shuffle protocols. In a proof of shuffle, $\mixer{k}$ for a given input ciphertext list $\ve{\epsilon}$ and output ciphertext list $\ve{\epsilon}'$ proves that $\ve{\epsilon}'$ is a permutation and re-encryption of $\ve{\epsilon}$ under a permutation that is consistent across $\mix$, $\zkpsm$ and $\zkprsm$ protocols (i.e., the permutation used during $\zkprsm$ is the same as that used during $\mix$ and the permutation used during $\zkpsm$ is the inverse of it). Such proofs can be given efficiently using the permutation-commitment based techniques of \citep{commitment-consistent-proof-shuffle,terelius-restricted-shuffles}. Each $(\mixer{k})_{k \in [m]}$ must verify proofs given by other mix-servers before participating in the corresponding threshold decryption protocols.
	
	\item[-] Each $(\mixer{k})_{k \in [m]}$ must provide proofs of knowledge of the blinding factors of homomorphically blinded signatures. During $\zkpsm$, this involves giving proofs $\vcs{\rho}{\tilde{\sigma}_i}{k}$ for each $\vcs{\tilde{\epsilon}}{{\sigma}_i}{k}$ published by $\mixer{k}$ as a blinding of $\vc{\epsilon}{{\sigma}_i}$ with blinding factor $\vcs{b}{i}{k}$. Note that $\vcs{\rho}{\tilde{\sigma}_i}{k}:=\nizkpk{r,\vcs{b}{i}{k}}{\tilde{c}_0 = g_1^r c_0^{\vcs{b}{i}{k}} \wedge \tilde{c}_1 = \thegpk^r c_1^{\vcs{b}{i}{k}}}$, where $\vc{\epsilon_{\sigma}}{i}$ and $\vcs{\tilde{\epsilon}_{\sigma}}{i}{k}$ are parsed as ElGamal ciphertexts $(c_0,c_1)$ and $(\tilde{c}_0, \tilde{c}_1)$ respectively. During $\zkprsm$, this involves giving NIZK proofs $(\vcs{\rho}{bS_j}{k},$ $\vcs{\rho}{bc_j}{k},$ $\vcs{\rho}{br_j}{k})$ for each $(\vcs{\epsilon_{bS}}{j}{k},$ $\vcs{\epsilon_{bc}}{j}{k},$ $\vcs{\epsilon_{br}}{j}{k})$, which are encryptions of $g_1^{\vcs{b_S}{j}{k}}$ under $\theg$ and of $\vcs{b_c'}{j}{k}:=\vcs{b_c}{j}{k} + q \vcs{\chi_c}{j}{k}$ and $\vcs{b_r'}{j}{k}:=\vcs{b_r}{j}{k} + q \vcs{\chi_r}{j}{k}$ under $\thez$ respectively. Here, $\vcs{\rho}{bS_j}{k}:=\nizkpk{r,\vcs{b_S}{j}{k}}{ \tilde{c}_0 = g_1^r \wedge \tilde{c}_1 = g_1^{\vcs{b_S}{j}{k}} \thegpk^r }$, where $\vcs{\epsilon_{bS}}{j}{k}$ is parsed as the ElGamal ciphertext $(\tilde{c}_0, \tilde{c}_1)$ and $\vcs{\rho}{bc_j}{k}$, $\vcs{\rho}{br_j}{k}$ are proofs of knowledge of the plaintexts encrypted by $\vcs{\epsilon_{bc}}{j}{k}$, $\vcs{\epsilon_{br}}{j}{k}$ respectively.

	\item[-] During the $\thegd$ and $\thezd$ protocols, each $(\mixer{k})_{k \in [m]}$ must provide proofs that they produced correct decryption shares. Each $(\mixer{k})_{k \in [m]}$ should proceed with stage 2 only if these proofs pass.

	\item[-] During the $\zkprsm$ protocol, each $(\mixer{k})_{k \in [m]}$ must provide a proof of knowledge of the opening of $\share{\mathfrak{z}_1}{k^*}$. Each $(\mixer{k})_{k \in [m]}$ should proceed only if the proof passes.
\end{itemize}

Note that the DPKs in $\zkpsm$ and $\zkprsm$ already employ the Fiat-Shamir heuristic \citep{fiat-shamir} which makes them general ZKPs in the random oracle model (see Section \ref{sss:distr-pok}).

In Section \ref{pf:secrecy-general}, we sketch a proof that our construction with the above steps protects secrecy (Definition \ref{def:secr}) against general malicious adversaries.

\section{Proofs}
\label{appendix:proofs}

\subsection{Proof for Theorem \ref{thm:completeness}}
\label{pf:completeness}

It can be inspected that when all the parties are honest, inputs to protocols $\zkpsm$ and  $\zkprsm$ satisfy the preconditions mentioned in Figures \ref{fig:distr-set-membership} and \ref{fig:distr-rev-set-membership}, respectively. This is also true for the reruns of $\zkpsm$ and $\zkprsm$ in the $\btrin/\btrout$ calls against the complement sets. Thus, by Lemma \ref{lem:compl-zkpsm}, $I^*$ and $I^*_{\mathsf{c}}$ obtained by $\querier$ in a $\btrin$ call satisfy $I^* = \{ i \in I \mid \vc{v}{i} \in \vc{v}{J}' \}$ and $I^*_{\mathsf{c}} = \{ i \in I \mid \vc{v}{i} \in \vc{v}{[n]\setminus J}' \}$. Also, the correctness of $\mix$ implies that for each $i \in I \subseteq [n]$, $\vc{v}{i} \in \{\vc{v}{j}' \mid j \in [n]\}$, i.e., $\vc{v}{i} \in \vc{v}{J}' \cup \vc{v}{[n]\setminus J}'$ for any $J \subseteq [n]$. Thus, $I^* \cup I^*_{\mathsf{c}} = I$, which implies that $\querier$ does not abort and outputs a $\vc{c}{I^*}$ that satisfies the first condition of $\cmpl$ (Figure \ref{exp-completeness}). By a similar argument using Lemma \ref{lem:compl-zkprsm}, it follows that $\querier$ does not abort in a $\btrout$ call and outputs a $\vc{v}{J^*}'$ that satisfies the second condition of $\cmpl$.
	
\begin{lemma} 
	\label{lem:compl-zkpsm}
	If inputs $(\thegpk,$ $\ve{\gamma},$ $\ve{v}',$ $I,$ $J,$ $(\si{\mixer{k}}{\thegsk{k},$ $\permfunc{k},$ $\vs{v}{k},$ $\vs{r}{k}})_{k \in [m]})$ of a $\zkpsm$ invocation satisfy the preconditions mentioned in Figure \ref{fig:distr-set-membership} and all the parties are honest then $\querier$ outputs $I^* = \{ i \in I \mid \exists j \in J: \sum_{k \in [m]} \vcs{v}{i}{k} = \vc{v}{j}' \}$.
\end{lemma}
\begin{proof}
	\label{proof-compl-zkpsm}
	Note that for honest mix-servers, the $i^{\text{th}}$ DPK passes iff $\mixer{k}$ uses $(\vcs{v}{i}{k}, \vcs{r}{i}{k}, \vcs{b}{i}{k})$ such that $(\vc{v}{i}, \vc{r}{i}, \vc{b}{i}):=(\sum_{k \in [m]} \vcs{v}{i}{k},$ $\sum_{k \in [m]}\vcs{r}{i}{k},$ $\sum_{k \in [m]}\vcs{b}{i}{k})$ satisfy the predicate $p_{\mathsf{BB}_i}$. The equation $\vc{\gamma}{i} = g_1^{\vc{v}{i}}h_1^{\vc{r}{i}}$ of $p_{\mathsf{BB}_i}$ is trivially satisfied by the correctness of $(\vcs{v}{i}{k},\vcs{r}{i}{k})$. Next, note that  $\vc{\tilde{\sigma}}{i} = (\vc{\sigma'}{\pi^{-1}(i)})^{{\sum_{k \in [m]} \vcs{b}{i}{k}}}$, by the homomorphism of $\theg$. Let $j \in [n]$ be the index to which index $i$ is mapped after permutation, i.e., $i = \pi(j)$ or equivalently $\pi^{-1}(i)=j$. Correctness of input conditions implies $\vc{v}{j}' = \sum_{k \in [m]} \vcs{v}{\pi(j)}{k} = \sum_{k \in [m]} \vcs{v}{i}{k}$. Thus, $\vc{\tilde{\sigma}}{i}= (\vc{\sigma}{j}')^{\sum_{k \in [m]} \vcs{b}{i}{k}}$, which equals $g_1^{\frac{\sum_{k \in [m]} \vcs{b}{i}{k}}{x + \vc{v}{j}' }}$ $=$ $g_1^{\frac{\sum_{k \in [m]} \vcs{b}{i}{k}}{x + \sum_{k \in [m]} \vcs{v}{i}{k} }}$ if $j \in J$ and $g_1^{\sum_{k \in [m]} \vcs{b}{i}{k}}$ if $j \not\in J$. In the first case, the second equation of $p_{\mathsf{BB}_i}$ passes; in the second case, it fails. Since the DPK is run only for $i \in I$, $I^*$ is exactly as claimed.
\end{proof}

\begin{lemma}
	\label{lem:compl-zkprsm}
	If inputs $(\thegpk,$ $\thezpk,$ $\ve{\gamma},$ $\ve{\pokcomm},$ $\ve{\epsilon_r},$ $\ve{v}',$ $I,$ $J,$ $(\si{\mixer{k}}{\thegsk{k},$ $\thezsk{k}, \permfunc{k}, \vs{v}{k},$ $\vs{r}{k}})_{k \in [m]})$ of a $\zkprsm$ invocation satisfy the preconditions mentioned in Figure \ref{fig:distr-rev-set-membership} and all the parties are honest then $\querier$ outputs $J^* = \{ j \in J \mid \exists i \in I: \sum_{k \in [m]} \vcs{v}{i}{k} = \vc{v}{j}' \}$.
\end{lemma}
\begin{proof}
	\label{proof-compl-zkprsm}
	As in Lemma \ref{lem:compl-zkpsm}, let $i$ be s.t. $i = \pi(j)$. By correctness of inputs, $\vc{v}{j}' = \sum_{k \in [m]} \vcs{v}{i}{k}$, $\vc{\epsilon}{r_i} \leftarrow \theze(\thezpk, \sum_{k \in [m]} \vcs{r}{i}{k})$. Then, it can be inspected that the following equalities hold:
	\begin{itemize}[leftmargin=*]
		\item[-] $\vc{\tilde{S}}{j}'$ $=$ $(f_1 g_1^{ \sum_{k \in [m]}  \vcs{v}{i}{k} }$ $h_1^{ \vc{\hat{r}}{i} + \sum_{k \in [m]}  \vcs{r}{i}{k} })^{\frac{1}{x +  \vc{c}{i}}}$ $g_1^{\sum_{k \in [m]}  \vcs{b_S}{j}{k} }$ if $i \in I$ else $f_1^0$,
		\item[-] $\vc{\tilde{c}}{j}' = \vc{c}{i} + \sum_{k \in [m]} \vcs{b_c}{j}{k}$, 
		\item[-] $\vc{\tilde{r}}{j}' = \vc{\hat{r}}{i} + \sum_{k \in [m]} \vcs{r}{i}{k} + \sum_{k \in [m]} \vcs{b_r}{j}{k}$
	\end{itemize} 
	Therefore, $(\vc{\tilde{S}}{j}'g_1^{-\vc{b}{\vc{S}{j}}},\vc{\tilde{c}}{j}' - \vc{b}{\vc{c}{j}}, \vc{\tilde{r}}{j}' - \vc{b}{\vc{r}{j}})$ satisfies the BBS+ verification equation $e(\vc{\tilde{S}}{j}'g_1^{-\vc{b}{\vc{S}{j}}}, yf_2^{\vc{\tilde{c}}{j}' - \vc{b}{\vc{c}{j}}}) = e(f_1g_1^{\vc{v}{j}'} h_1^{\vc{\tilde{r}}{j}' - \vc{b}{\vc{r}{j}}},f_2)$ if $i \in I$, where $(\vc{b_S}{j}, \vc{b_c}{j}, \vc{b_r}{j})$ $:=$ $(\sum_{k \in [m]} \vcs{b_S}{j}{k},$ $\sum_{k \in [m]} \vcs{b_c}{j}{k},$ $\sum_{k \in [m]} \vcs{b_r}{j}{k})$. So: 
	\begin{align}
	\begin{split}
	\label{eq:dpk2}
	& \quad\quad e(\vc{\tilde{S}}{j}'g_1^{-\vc{b_S}{j}}, yf_2^{\vc{\tilde{c}}{j}'-\vc{b_c}{j}}) = e(f_1 g_1^{\vc{v}{j}'} h_1^{\vc{\tilde{r}}{j}'-\vc{b_r}{j}}, f_2) \\
	& \Leftrightarrow \frac{ e(\vc{\tilde{S}}{j}', yf_2^{\vc{\tilde{c}}{j}'-\vc{b_c}{j}}) }{ e(g_1^{\vc{b_S}{j}}, yf_2^{\vc{\tilde{c}}{j}'-\vc{b_c}{j}}) } = e(f_1 g_1^{\vc{v}{j}'} h_1^{\vc{\tilde{r}}{j}'-\vc{b_r}{j}}, f_2)  \\
	& \Leftrightarrow \frac{ e(\vc{\tilde{S}}{j}', yf_2^{\vc{\tilde{c}}{j}'}) }{ e(f_1 g_1^{\vc{v}{j}'} h_1^{\vc{\tilde{r}}{j}'}, f_2) } = \\
	& \quad\quad e(\vc{\tilde{S}}{j}', f_2)^{\vc{b_c}{j}} e(g_1,  yf_2^{\vc{\tilde{c}}{j}'})^{\vc{b_S}{j}}  e(h_1,  f_2)^{-\vc{b_r}{j}}e(g_1, f_2)^{-\vc{b_S}{j}\vc{b_c}{j}} \\
	& \Leftrightarrow \mathfrak{z}_2 = \mathfrak{g}_1^{ \vc{b_c}{j} } \mathfrak{g}_2 ^{\vc{b_S}{j}} \mathfrak{h}_1 ^ {\vc{b_r}{j}} \mathfrak{h}_2 ^ {\vc{b_S}{j}\vc{b_c}{j}}
	\end{split}
	\end{align}
	Thus, with the mix-servers holding shares of $(\vc{b_S}{j}, \vc{b_c}{j}, \vc{b_r}{j})$ and $\delta_1 := \vc{b_S}{j}\vc{b_c}{j}$ and $\delta_2 := \delta_0 \delta_1$, all equations of the predicate $p_{\mathsf{BBS}+_j}$ of the DPK are satisfied if $i \in I$. If $i \not\in I$, the last equation of $p_{\mathsf{BBS}+_j}$  is not satisfied. The claim follows.
\end{proof}

\subsection{Proof for Theorem \ref{thm:soundness}}
\label{pf:soundness}

Suppose for contradiction that there is a PPT adversary $\adv$ such that $\snd$ (Figure \ref{exp-soundness}) outputs $1$ with non-negligible probability. Note first that $\querier$ always outputs $I^* \subseteq I$ in $\zkpsm$ and $J^* \subseteq J$ in $\zkprsm$. Thus, $\vc{c}{I^*} \subseteq \vc{c}{I}$ and $\vc{v}{J^*}' \subseteq \vc{v}{J}'$. We now consider the following cases:

\noindent \emph{Case 1: $\vc{c}{I^*} \neq \{ \vc{c}{i} \in \vc{c}{I} \mid \vc{v}{i} \in \vc{v}{J}'\}.$} This leads to the following sub-cases:
\begin{itemize}[leftmargin=*]
	\item[-] \emph{Case 1.1: $\exists \vc{c}{i} \in \vc{c}{I^*}: \vc{v}{i} \not\in \vc{v}{J}'.$} As per Figure \ref{fig:construction}, $\vc{c}{i}$ contains $\vc{\gamma}{i} = g_1^{\vc{v}{i}}h_1^{\vc{r}{i}}$ for some $\vc{r}{i} \in \Zq$. Since $\vc{c}{i} \in \vc{c}{I^*}$, $\vc{\gamma}{i} \in \vc{\gamma}{I^*}$. Thus, by Lemma \ref{lem:snd-zkpsm}, a PPT extractor can extract a tuple $(j^*,r^*)$ such that $\vc{\gamma}{i} = g_1^{\vc{v}{j^*}'}h_1^{r^*}$ and $\vc{v}{j^*}' \in \vc{v}{J}'$. The requirement $\vc{v}{i} \not\in \vc{v}{J}'$ thus implies that $\vc{v}{i} \neq \vc{v}{j^*}'$. This allows producing two different openings $(\vc{v}{j^*}', r^*)$ and $(\vc{v}{i}, \vc{r}{i})$ for Pedersen commitment $\vc{\gamma}{i}$, which is a contradiction under the discrete logarithm assumption in $\Gone$.
	\item[-] \emph{Case 1.2: $\exists \vc{c}{i} \in \vc{c}{I\setminus I^*}: \vc{v}{i} \in \vc{v}{J}'.$} Note that since $\querier$ produces a $\vc{c}{I^*}$ and does not abort, it must be that $I^* \cup I^*_{\mathsf{c}} = I$ during the $\btrin$ call. Thus, $I^*_{\mathsf{c}} = I \setminus I^*$. Further, since all $\vc{v}{i}'$s are distinct, $\vc{v}{i} \in \vc{v}{J}' \implies \vc{v}{i} \not\in \vc{v}{[n] \setminus J}'$. Thus, this case can be restated as follows: $\exists \vc{c}{i} \in \vc{c}{I^*_{\mathsf{c}}}: \vc{v}{i} \not\in \vc{v}{[n] \setminus J}'$. Thus, by applying Lemma \ref{lem:snd-zkpsm} for the second $\zkpsm$ call in $\btrin$ and proceeding as the previous case, we conclude that this case is not possible.
\end{itemize}

\noindent \emph{Case 2: $\vc{v}{J^*}' \neq \{ \vc{v}{j}' \in \vc{v}{J}' \mid \vc{v}{j}' \in \vc{v}{I}\}.$} This leads to the following sub-cases:
\begin{itemize}[leftmargin=*]
	\item[-] \emph{Case 2.1: $\exists \vc{v}{j}' \in \vc{v}{J^*}': \vc{v}{j}' \not\in \vc{v}{I}.$} Since $\vc{v}{j}' \in \vc{v}{J^*}'$, by Lemma \ref{lem:snd-zkprsm}, a PPT extractor can extract a tuple $(i^*,r^*)$ such that $\vc{\gamma}{i^*} = g_1^{\vc{v}{j}'}h_1^{r^*}$ and $i^* \in I$. Since $i^* \in I$, the requirement $\vc{v}{j}' \not\in \vc{v}{I}$ implies $\vc{v}{j}' \neq \vc{v}{i^*}$. As per Figure \ref{fig:construction}, $\vc{\gamma}{i^*} \leftarrow g_1^{\vc{v}{i^*}}h_1^{\vc{r}{i^*}}$ for some $\vc{r}{i^*} \in \Zq$.  This allows producing two different openings $(\vc{v}{j}', r^*)$ and $(\vc{v}{i^*}, \vc{r}{i^*})$ for $\vc{\gamma}{i^*}$, which leads to a contradiction.
	\item[-] \emph{Case 2.2: $\exists \vc{v}{j}' \in \vc{v}{J \setminus J^*}': \vc{v}{j}' \in \vc{v}{I}.$} Note that since $\querier$ produces a $\vc{v}{J^*}'$ and does not abort, it must be that $J^* \cup J^*_{\mathsf{c}} = J$ during the $\btrin$ call. Thus, $J^*_{\mathsf{c}} = J \setminus J^*$. Further, since all $\vc{v}{i}$s are distinct, $\vc{v}{j}' \in \vc{v}{I} \implies \vc{v}{j}' \not\in \vc{v}{[n] \setminus I}$. Thus, this case can be restated as follows: $\exists \vc{v}{j}' \in \vc{v}{J^*_{\mathsf{c}}}: \vc{v}{j}' \not\in \vc{v}{[n] \setminus I}$. Thus, by applying Lemma \ref{lem:snd-zkprsm} for the second $\zkprsm$ call in $\btrout$ and proceeding as the previous case, we conclude that this case is not possible.
\end{itemize}

\begin{lemma}
	\label{lem:snd-zkpsm}
	If $\querier$ participates in the $\zkpsm$ protocol with common input $(\thegpk, \ve{\gamma}, \ve{v}', I, J)$ and outputs $I^*$, then for all PPT adversaries $\adv$ controlling $(\prover{k})_{k \in [m]}$ and for all $\vc{\gamma}{i} \in \vc{\gamma}{I^*}$, there exists a PPT extractor $\mathcal{E}$ that outputs a $(j^*, r^*)$ such that $\vc{\gamma}{i} = g_1^{ \vc{v}{j^*}' } h_1^{ r^* }  \wedge \vc{v}{j^*}' \in \vc{v}{J}'$.
\end{lemma}
\begin{proof}
	Note that the verification steps of $\querier$ in $\zkpsm$ for a given $\vc{\gamma}{i} \in \vc{\gamma}{I}$ are exactly those of the verifier of the single-prover ZKP of set membership \cite{Camenisch} for commitment $\vc{\gamma}{i}$ against set $\phi := \vc{v}{J}'$ when extended to our asymmetric pairing setting. Thus, for each $\vc{\gamma}{i} \in \vc{\gamma}{I^*}$, by the special soundness of the single-prover ZKP \citep{Camenisch} under the $n$-Strong Diffie-Hellman assumption in $(\Gone,\Gtwo)$, a PPT extractor $\mathcal{E}'$ can extract $v^*, r^*$ such that $\vc{\gamma}{i} = g_1^{v^*}h_1^{r^*} \wedge v^* \in \phi$. $\mathcal{E}$ simply runs $\mathcal{E}'$, finds $j^*$ such that $v^* = \vc{v}{j^*}'$ and outputs $(j^*, r^*)$.
\end{proof}

\begin{lemma}
	\label{lem:snd-zkprsm}
	If $\querier$ participates in the $\zkprsm$ protocol with common input $(\thegpk, \thezpk, \ve{\gamma}, \ve{\pokcomm}, \ve{\epsilon_r}, \ve{v}', I, J)$ and outputs $J^*$, then for all PPT adversaries $\adv$ controlling $(\prover{k})_{k \in [m]}$ and for all $\vc{v}{j}' \in \vc{v}{J^*}'$, there exists a PPT extractor $\mathcal{E}$ that outputs an $(i^*, r^*)$ such that $\vc{\gamma}{i^*} = g_1^{ \vc{v}{j}' } h_1^{ r^* }  \wedge i^* \in I$.
\end{lemma}
\begin{proof}
	\begin{figure}
		\scalebox{0.8}{
			\begin{tabular}{rll}
				1 & \multicolumn{2}{l}{\underline{$\mathcal{E}(\thegpk, \thezpk, \ve{\gamma}, \ve{\pokcomm}, \ve{v}', I, J)$:}} \\
				2 & $\mathcal{E}$: & \textbf{for each $i \in I$:} $(\vc{v}{i}, \vc{r}{i}) \leftarrow \mathcal{E}_1(\vc{\gamma}{i}, \vc{\pokcomm}{i})$ \\
				3 & $\mathcal{E} \rightarrow \mathcal{C}$: & $(\vc{v}{i})_{i \in I}$ (BBS+ signature queries) \\
				4 & $\mathcal{C} \rightarrow \mathcal{E}$: & $y, (\vc{\sigma}{i})_{i \in I}$ \\
				5 & $\mathcal{E}$: & \textbf{for each $i \in [n]$:} \\
				6 &	& \quad \textbf{if $i \in I$:} $(\vc{S}{i}, \vc{c}{i}, \vc{\mathsf{r}}{i}) := \vc{\sigma}{i}$; \quad $\vc{\hat{r}}{i} \leftarrow \vc{\mathsf{r}}{i} - \vc{r}{i}$ \\
				7 &	& \quad \textbf{else:} $\vc{S}{i} \leftarrow f_1^0$; $\vc{c}{i} \xleftarrow{\$} \Zq$; $\vc{\hat{r}}{i} \xleftarrow{\$} \Zq$ \\
				8 &	& \textbf{endfor} \\
				9 & & $(\ve{\epsilon}_S, \ve{\epsilon}_c, \ve{\epsilon}_{\hat{r}}) \leftarrow (\thege(\thegpk, \ve{S}), \theze(\thezpk, \ve{c}),$ \\ 
				& & \quad \quad $\theze(\thezpk, \ve{\hat{r}}))$ \\
				10 & $\mathcal{E} \rightarrow \adv$: & $y, \ve{\hat{\sigma}}:=(\ve{S}, \ve{c}, \ve{\hat{r}}), \ve{\epsilon}_{\hat{\sigma}}:=(\ve{\epsilon}_S, \ve{\epsilon}_c, \ve{\epsilon}_{\hat{r}})$ \\
				11 & $\adv \rightarrow \mathcal{E}$: & $\ve{\tilde{\sigma}}' := (\ve{\tilde{S}}', \ve{\tilde{c}}', \ve{\tilde{r}}')$ \\ 
				12 & $\mathcal{E}$: & \textbf{for each $j \in J$:} \\ 
				13 & & \quad Compute $\mathfrak{g}_1, \mathfrak{g}_2, \mathfrak{h}_1, \mathfrak{h}_2, \mathfrak{h}_3, \mathfrak{z}_1,\mathfrak{z}_2$ from $\vc{\tilde{\sigma}}{j}':=(\vc{\tilde{S}}{j}', \vc{\tilde{c}}{j}', \vc{\tilde{r}}{j}')$ \\
				14 & & \quad $(\vc{b_S}{j}, \vc{b_c}{j}, \vc{b_r}{j}, \delta_0, \delta_1, \delta_2) \leftarrow$ $\mathcal{E}_{2}(\mathfrak{g}_1, \mathfrak{g}_2, \mathfrak{h}_1, \mathfrak{h}_2, \mathfrak{h}_3, \mathfrak{z}_1, \mathfrak{z}_2)$ \\
				15 & & \quad $\vc{\sigma}{j}':=(\vc{\tilde{S}}{j}' g_1^{ -\vc{b_S}{j} }, \vc{\tilde{c}}{j}' - \vc{b_c}{j}, \vc{\tilde{r}}{j}' - \vc{b_r}{j})$ \\
				16 & & \quad \textbf{for each $i \in I$:} \\
				17 & & \quad \quad \textbf{if $(\vc{v}{i} = \vc{v}{j}')$:} output $(i,\vc{r}{i})$ \\
				18 & & \quad \textbf{endfor} \\
				19 & & \quad output $(\vc{v}{j}', \vc{\sigma}{j}')$ \\
				20 & & \textbf{endfor}
			\end{tabular}
		}
		\caption{Extractor $\mathcal{E}$.}
		\label{fig:extractor-rev-set-membership}
	\end{figure}
	
	We construct an algorithm $\mathcal{E}$ that for each $\vc{v}{j}' \in \vc{v}{J^*}'$ either outputs a desired tuple $(i^*,r^*)$ or forges a BBS+ signature, using extractors $\mathcal{E}_1$ for the NIZK proofs of knowledge of commitment openings $\vc{\rho}{\gamma_i}$ and $\mathcal{E}_2$ for the DPKs for $p_{\mathsf{BBS}+_j}$. On one end, $\mathcal{E}$ interacts with adversary $\adv$ controlling $(\mixer{k})_{k \in [m]}$ in the $\zkprsm$ protocol; on the other, with the challenger $\mathcal{C}$ of the BBS+ signature unforgeability game. See Figure \ref{fig:extractor-rev-set-membership}.
	
	For each $i \in I$, $\mathcal{E}$ first extracts opening $(\vc{v}{i}, \vc{r}{i})$ of commitment $\vc{\gamma}{i}$ from the NIZK proof $\vc{\pokcomm}{i}$ using extractor $\mathcal{E}_1$ (line 2). It then obtains BBS+ public key $y$ and signatures for each $(\vc{v}{i})_{i \in I}$ from the BBS+ challenger $\mathcal{C}$ (lines 3-4), derives quasi-signatures from them using $\vc{r}{i}$ (line 6) and forwards the quasi-signatures and their encrypted versions to $\adv$, along with invalid quasi-signatures for $i\not\in I$ --- similar to $\querier$ (lines 5-10). $\adv$ responds with blinded permuted signatures $\ve{\tilde{\sigma}}'$ (line 11), as $(\mixer{k})_{k \in [m]}$ do in the real protocol at the end of stage 1. In stage 2, for each $j \in J$, $\mathcal{E}$ extracts blinding factors for the blinded signature $\vc{\tilde{\sigma}}{j}'$ using extractor $\mathcal{E}_2$ (lines 13-14), from which it obtains an unblinded signature $\vc{\sigma}{j}'$ (line 15). $\mathcal{E}$ then attempts to find an opening $\vc{v}{i}$ extracted by $\mathcal{E}_1$ for some commitment $\vc{\gamma}{i} \in \vc{\gamma}{I}$ such that $\vc{v}{i}=\vc{v}{j}'$, and returns the corresponding tuple $(i,\vc{r}{i})$ (lines 16-18). If no such $\vc{v}{i}$ exists, it outputs the message-signature tuple $(\vc{v}{j}', \vc{\sigma}{j}')$ (line 19). 
	
	Note that $\querier$ produced some $J^*$ and did not abort. In this case, the view produced by $\mathcal{E}$ to $\adv$ is identical to that produced by $\querier$ to $(\mixer{k})_{k \in [m]}$. Further, for some $\vc{v}{j}' \in \vc{v}{J^*}'$, if $\vc{v}{i}=\vc{v}{j}'$ for some $i\in I$, $(i,\vc{r}{i})$ output in line 17 is a desired tuple since $\vc{\gamma}{i}=g_1^{\vc{v}{i}}h_1^{\vc{r}{i}} = g_1^{\vc{v}{j}'}h_1^{\vc{r}{i}}$ and $i \in I$ (the first equality follows by the soundness of the proof of knowledge of commitment openings; the second because $\vc{v}{i}=\vc{v}{j}'$). If no such $\vc{v}{i}$ exists, we show that the tuple $(\vc{v}{j}', \vc{\sigma}{j}')$ output in line 19 is a valid BBS+ signature forgery:
	\begin{itemize}[leftmargin=*]
		\item[-] Since for all $i\in I$, $\vc{v}{i} \neq \vc{v}{j}'$, a BBS+ signature for $\vc{v}{j}'$ was not queried from $\mathcal{C}$ in line 3. \label{enum:fresh-sig}
		\item[-] Since $\vc{v}{j}' \in \vc{v}{J^*}'$, the DPK for $p_{\mathsf{BBS}+_j}$ must have passed. Thus, by the soundness of DPKs: \label{enum:valid-sig}
		\begin{align}
		\mathfrak{z}_1 &= \mathfrak{h}_2^{ \vc{b_S}{j} }\mathfrak{h}_3^{\delta_{0}} \label{eq:z} \\
		\mathfrak{z}^{ \vc{b_c}{j} }_1 &= \mathfrak{h}_2^{\delta_{1}} \mathfrak{h}_3^{\delta_{2}}, \label{eq:z-power-c} \\
		\mathfrak{z}_2 &= \mathfrak{g}_1^{ \vc{b_c}{j} } \mathfrak{g}_2 ^{ \vc{b_S}{j} } \mathfrak{h}_1 ^ { \vc{b_r}{j} } \mathfrak{h}_2 ^ {\delta_{1}} \label{eq:y}
		\end{align}
		From Equations \ref{eq:z} and \ref{eq:z-power-c}, $\mathfrak{h}_2^{ \vc{b_S}{j} \vc{b_c}{j} }\mathfrak{h}_3^{\delta_{0} \vc{b_c}{j} }  = \mathfrak{h}_2^{\delta_{1}} \mathfrak{h}_3^{\delta_{2}}$. It must be that $\delta_{1} = \vc{b_S}{j}\vc{b_c}{j}$, otherwise two different openings $(\delta_{1}, \delta_{2})$ and $(\vc{b_S}{j}\vc{b_c}{j}, \delta_{0}\vc{b_c}{j})$ for the Pedersen commitment $\mathfrak{z}^{ \vc{b_c}{j} }_1$ can be produced. Equation \ref{eq:y} thus implies $\mathfrak{z}_2 = \mathfrak{g}_1^{ \vc{b_c}{j} } \mathfrak{g}_2 ^{\vc{b_S}{j}} \mathfrak{h}_1 ^ {\vc{b_r}{j}} \mathfrak{h}_2 ^ {\vc{b_S}{j}\vc{b_c}{j}}$, which implies that $\vc{\sigma}{j}':=(\vc{\tilde{S}}{j}'g_1^{-\vc{b_S}{j}}, \vc{\tilde{c}}{j}'-\vc{b_c}{j}, \vc{\tilde{r}}{j}'-\vc{b_r}{j})$ satisfies the BBS+ signature verification equation on message $\vc{v}{j}'$ under public key $y$: $e(\vc{\tilde{S}}{j}'g_1^{-\vc{b_S}{j}}, yf_2^{\vc{\tilde{c}}{j}'-\vc{b_c}{j}}) = e(f_1 g_1^{\vc{v}{j}'} h_1^{\vc{\tilde{r}}{j}'-\vc{b_r}{j}}, f_2)$ (see Equation \ref{eq:dpk2}; Lemma \ref{lem:compl-zkprsm}).
	\end{itemize}
	Since at most $n$ signature queries could have been made, forging a BBS+ signature is not possible under the $n$-Strong Diffie Hellman assumption in $(\Gone,\Gtwo)$ \cite{bbsplus-sig}. Thus, for each $\vc{v}{j}' \in \vc{v}{J^*}'$, some $i \in I$ such that $\vc{v}{i}=\vc{v}{j}'$ must exist and a desired tuple $(i,\vc{r}{i})$ must have been produced.
\end{proof}

\subsection{Proof for Theorem \ref{thm:secrecy}}
\label{pf:secrecy}

We prove the theorem by considering the following sequence of hybrid experiments (the complete hybrid experiments are shown in Figures \ref{fig:e1} to \ref{fig:e19}; solid boxes denote lines changed in a given hybrid and dashed boxes denote lines that change in the next hybrid):
\begin{itemize}[leftmargin=*]
	\item[-] $E_1$ (Figure \ref{fig:e1}): $E_1$ is the original secrecy game $\secr$ instantiated for our protocol.
	
	\item[-] $E_2$ (Figure \ref{fig:e2}): In $E_2$, NIZKPKs $\vc{\pokcomm}{i_0}$ and $\vc{\pokcomm}{i_1}$ are simulated. $E_2$ is indistinguishable from $E_1$ in the random oracle model because of the ZK property of the NIZK proof.
	
	\item[-] $E_3$ (Figure \ref{fig:e3}): In $E_3$, shares of $\vc{v}{i}, \vc{r}{i}$ for $i \in \{i_0,i_1\}$ are drawn as $(\vcs{v}{i}{k},\vcs{r}{i}{k})_{k \neq k^*} \xleftarrow{\$} \Zq$ and $\vcs{v}{i}{k^*} \leftarrow \vc{v}{i} - \sum_{k \neq k^*} \vcs{v}{i}{k}$, $\vcs{r}{i}{k^*} \leftarrow \vc{r}{i} - \sum_{k \neq k^*} \vcs{r}{i}{k}$. $E_3$ is indistinguishable from $E_2$ because the additive secret sharing is information-theoretically secure.
	
	\item[-] $E_4$ (Figure \ref{fig:e4}): In $E_4$, instead of decrypting $\vcs{\mathsf{ev}}{i}{k^*}$, $\vcs{\mathsf{er}}{i}{k^*}$ for $i \in \{i_0,i_1\}$ during $\mix$, $\vcs{v}{i}{k^*}$, $\vcs{r}{i}{k^*}$ are obtained by directly using their corresponding values in the $\enc$ call for $i_0/i_1$. $E_4$ is indistinguishable from $E_3$ because of correctness of decryption.
	
	\item[-] $E_5$ (Figure \ref{fig:e5}): In $E_5$, instead of decrypting $\vcs{\mathsf{ev}}{i}{k^*}$, $\vcs{\mathsf{er}}{i}{k^*}$ for $i \in [n] \setminus \{i_0,i_1\}$ during $\mix$, $\vcs{v}{i}{k^*}$, $\vcs{r}{i}{k^*}$ are obtained by rerunning the $\enc$ algorithm on input $\vc{v}{i}$ for sender $S_i$ using the randomnesses from the random tape issued to $\adv$. $E_5$ is indistinguishable from $E_4$ because of correctness of decryption.
	
	\item[-] $E_6$ (Figure \ref{fig:e6}): In $E_6$, encrypted shares $\vcs{\mathsf{ev}}{i_0}{k^*},\vcs{\mathsf{er}}{i_0}{k^*},\vcs{\mathsf{ev}}{i_1}{k^*},\vcs{\mathsf{er}}{i_1}{k^*}$ in the $\enc$ calls for $i_0,i_1$ are replaced by encryptions of $0$. $E_6$ is indistinguishable from $E_5$ by the IND-CPA security of $\pke$.
	
	\item[-] $E_7$ (Figure \ref{fig:e7}): In $E_7$, responses of $\mixer{k}^*$ in $\thegd$ and $\thezd$ protocols are simulated by first obtaining the correct decryption $m$ of the given ciphertext $c$ using ideal decryption oracles $\thegfunc$, $\thezfunc$ and then simulating $\mixer{k}^*$'s responses using $c,m$ and secret keys for $(\mixer{k})_{k\neq k^*}$. $E_7$ is indistinguishable from $E_6$ by the security of the threshold decryption protocols of $\theg$ and $\thez$. For $\theg$, it follows straightforwardly because given $c=(c_0,c_1) \in \mathbb{C}(\theg)$ and $m=c_1 / \prod_{k \in [m]} c_0^{\thegsk{k}}$ obtained from $\thegfunc$, the simulator can output $\frac{c_1}{m\prod\limits_{k \neq k^*}c_0^{\thegsk{k}}}$ $=$ $\frac{c_1\prod\limits_{k \in [m]}c_0^{\thegsk{k}}}{c_1 \prod\limits_{k \neq k^*}c_0^{\thegsk{k}}}$ $=$ $c_0^{\thegsk{k^*}}$, which is exactly the decryption share output by the real $\mixer{k^*}$. For $\thez$, Damg{\aa}rd et al. provide a proof (see Theorem 4; \citep{dj-generalisation}).
    
    \item[-] $E_8$ (Figure \ref{fig:e8}): In $E_8$, instead of using $\ve{v}'$ given by $\thezfunc$ during $\mix$, it is set as $\ve{v}' \leftarrow (\vc{v}{\pi(j)})_{j \in [n]}$. Here, $\pi \leftarrow \permfunc{m} \circ \dots \circ \permfunc{1}$ is obtained using the experimenter-selected $\permfunc{k^*}$ and $(\permfunc{k})_{k \neq k^*}$ obtained from the random tape issued to $\adv$, $(\vc{v}{i})_{i \in \{i_0,i_1\}}$ are obtained from their values in $\enc$ calls for $i_0,i_1$, and $(\vc{v}{i})_{i \in [n]\setminus \{i_0,i_1\}}$ are obtained from $\adv$'s input for senders $(S_i)_{i \in [n] \setminus \{i_0,i_1\}}$. $E_8$ is indistinguishable from $E_7$ by the correctness of $\shuffle$ and $\thezd$ protocols.
    
    \item[-] $E_9$ (Figure \ref{fig:e9}): In $E_9$, instead of using $\ve{\tilde{\sigma}}$ given by $\thegfunc$ during $\zkpsm$, it is set as $\ve{\tilde{\sigma}} \leftarrow ((\vc{\sigma}{\pi^{-1}(i)}')^{\vc{b}{i}})_{i \in [n]}$. Here $\ve{\sigma}'$ denotes BB signatures sent by $\adv$ at the beginning of $\zkpsm$, $\pi$ is as obtained in $E_8$ and $\vc{b}{i} := \sum_{k \in [m]} \vcs{b}{i}{k}$ is obtained using the experimenter-selected $\vcs{b}{i}{k^*}$ and $(\vcs{b}{i}{k})_{k \neq k^*}$ obtained from $\adv$'s random tape. $E_9$ is indistinguishable from $E_8$ by the correctness of $\shuffle$ and $\thegd$ protocols.
    
    \item[-] $E_{10}$ (Figure \ref{fig:e10}): In $E_{10}$, instead of using $\ve{\tilde{S}}',\ve{\tilde{c}}'', \ve{\tilde{r}}''$ given by $\thegfunc$ and $\thezfunc$ during $\zkprsm$, they are set as 
	$\ve{\tilde{S}}'$ $\leftarrow$ $(\vc{S}{\pi(j)}g_1^{\vc{b_S}{j}})_{j \in [n]}$, $\ve{\tilde{c}}''$ $\leftarrow$ $(\vc{c}{\pi(j)} + \vc{b_c}{j}')_{j \in [n]}$, $\ve{\tilde{r}}'' \leftarrow (\vc{\hat{r}}{\pi(j)} + \vc{r}{\pi(j)} + \vc{b_r}{j}')_{j \in [n]}$. Here, $(\ve{S}, \ve{c}, \ve{\hat{r}})$ denote BBS+ quasi-signatures sent by $\adv$ at the beginning of $\zkprsm$, $\pi$ is as obtained in $E_8$, 
	$\vc{b_S}{j} \leftarrow \sum_{k \in [m]} \vcs{b_S}{j}{k}$, $\vc{b_c}{j}' \leftarrow \sum_{k \in [m]} \vcs{b_c'}{j}{k} \bmod N$, 
	$\vc{b_r}{j}' \leftarrow \sum_{k \in [m]} \vcs{b_r'}{j}{k} \bmod N$ are obtained using experimenter-selected 
	$\vcs{b_S}{j}{k^*}, \vcs{b_c'}{j}{k^*}:=\vcs{b_c}{j}{k^*} + q \vcs{\chi_c}{j}{k^*}, \vcs{b_r'}{j}{k^*}:=\vcs{b_r}{j}{k^*} + q \vcs{\chi_r}{j}{k^*}$ and 
	$(\vcs{b_S}{j}{k}, \vcs{b_c'}{j}{k}:= \vcs{b_c}{j}{k} + q \vcs{\chi_c}{j}{k}, \vcs{b_r'}{j}{k}:=\vcs{b_r}{j}{k} + q \vcs{\chi_r}{j}{k})_{k \neq k^*}$ computed using $\adv$'s random tape, 
	$(\vc{r}{i})_{i \in \{i_0,i_1\}}$ are obtained from the $\enc$ calls for $i_0,i_1$ and $(\vc{r}{i})_{i \in [n]\setminus \{i_0,i_1\}}$ are obtained from the 
	random tape given to $\adv$. $E_{10}$ is indistinguishable from $E_9$ by the correctness of $\shuffle$, $\thegd$ and $\thezd$ protocols.

	\item[-] $E_{11}$ (Figure \ref{fig:e11}): In $E_{11}$, all encryptions/re-encryptions under $\thez$ are replaced by encryptions of $0$. $E_{11}$ is indistinguishable to $E_{10}$ by the IND-CPA security of $\thez$ under the DCR assumption \citep{paillier}.
	
	\item[-] $E_{12}$ (Figure \ref{fig:e12}): In $E_{12}$, all encryptions/re-encryptions under $\theg$ are replaced by encryptions of $g_1^0$. $E_{12}$ is indistinguishable to $E_{11}$ by the IND-CPA security of $\theg$ under the DDH assumption.
	
	\item[-] $E_{13}$ (Figure \ref{fig:e13}): In $E_{13}$, $\mixer{k^*}$'s responses during the DPKs in $\zkpsm$ are simulated using shares of $(\mixer{k})_{k\neq k^*}$, as follows:
	\begin{itemize}[leftmargin=*]
		\item[-] For $i \in \{i_0,i_1\}$, $\mixer{k^*}$'s responses are simulated completely if $\pi^{-1}(i_0) \in J$; otherwise only the responses corresponding to the first equation $\vc{\gamma}{i} = g_1^{\sum_{k \in [m]} \vcs{v}{i}{k}}h_1^{\sum_{k \in [m]} \vcs{r}{i}{k}}$ are simulated. For simulating the DPK completely, $(\vcs{v}{i}{k},\vcs{r}{i}{k})_{k \neq k^*}$ are obtained from their values in the $\enc$ calls for $i_0/i_1$ and $(\vcs{b}{i}{k})_{k \neq k^*}$ are chosen as the ones obtained in $E_9$. For simulating only the first component of the DPK, $\mixer{k^*}$'s responses corresponding to the first equation are simulated using $(\vcs{v}{i}{k}, \vcs{r}{i}{k})_{k \neq k^*}$ obtained from the $\enc$ calls for $i_0/i_1$, but $\mixer{k^*}$'s strategy is followed for the second component except that a freshly sampled $\vcs{b}{i}{k^*}$ is used instead.
		\item[-] For $i \in [n] \setminus \{i_0,i_1\}$, the DPK is not simulated, i.e., values $\vcs{v}{i}{k^*},\vcs{r}{i}{k^*},\vcs{b}{i}{k^*}$ are used exactly as $\mixer{k^*}$ does.
	\end{itemize}
	$E_{13}$ is indistinguishable from $E_{12}$ because $a)$ by the correctness of the blinded signatures, the DPK for $i_0$ passes iff $j_0 = \pi^{-1}(i_0) \in J$ (otherwise the blinded signature is invalid), $b)$ by the assert condition in the $\otrin$ call, the DPK for $i_1$ passes iff the DPK for $i_0$ passes and $c)$ even when these DPKs do not pass, their first component always passes.

	\item[-] $E_{14}$ (Figure \ref{fig:e14}): In $E_{14}$, $\mixer{k^*}$'s responses during the DPKs in $\zkprsm$ are simulated using shares of $(\mixer{k})_{k\neq k^*}$, as follows:
	\begin{itemize}[leftmargin=*]
		\item[-] For $j \in \{j_0,j_1\}$ where $j_0,j_1$ are s.t. $v_0 = \vc{v}{j_0}'$ and $v_1 = \vc{v}{j_1}'$, the DPK is simulated completely if $\pi(j_0) \in I$; otherwise only its first two components are simulated. For simulating the DPK completely, $(\vcs{b_S}{j}{k},$ $\vcs{b_c}{j}{k},$ $\vcs{b_r}{j}{k})_{k \neq k^*}$  are chosen as the ones obtained in $E_{10}$, $(\share{\delta}{k}_0)_{k \neq k^*}$ are obtained from the random tape issued to $\adv$ and $(\share{\delta}{k}_1, \share{\delta}{k}_2)_{k \neq k^*}$ are obtained by applying the $\mult$ algorithm to inputs $(\share{\delta}{k}_0, \vcs{b_S}{j}{k})$ and $(\share{\delta}{k}_0, \vcs{b_c}{j}{k})$ respectively. For simulating only the first two components, $\mixer{k^*}$'s responses corresponding to the first two equations are simulated using $(\vcs{b_S}{j}{k},$ $\vcs{b_c}{j}{k}, \share{\delta}{k}_0, \share{\delta}{k}_1, \share{\delta}{k}_2)_{k \neq k^*}$ obtained as above whereas for the third equation, $\mixer{k^*}$'s strategy is followed except with a freshly sampled $\vcs{b_r}{j}{k^*}$.
		\item[-] For $j \in [n] \setminus \{j_0,j_1\}$, the DPK is not simulated, i.e., values $(\vcs{b_S}{j}{k^*},$ $\vcs{b_c}{j}{k^*},$ $\vcs{b_r}{j}{k^*},\share{\delta}{k^*}_0, \share{\delta}{k^*}_1, \share{\delta}{k^*}_2)$ are used exactly as $\mixer{k^*}$ does.
	\end{itemize}
	$E_{14}$ is indistinguishable from $E_{13}$ because $a)$ by the correctness of the blinded signatures, the DPK for $j_0$ passes iff $i_0 = \pi(j_0) \in I$ (otherwise the blinded signature is invalid), $b)$ by the assert condition in the $\otrout$ call, the DPK for $j_1$ passes iff the DPK for $j_0$ passes, and $c)$ even when these DPKs do not pass, their first two equations do pass.

	\item[-] $E_{15}$ (Figure \ref{fig:e15}): In $E_{15}$, $\share{\mathfrak{z}}{k^*}_1$ in $\zkprsm$ is chosen uniformly at random from $\Zq$. $E_{15}$ is indistinguishable from $E_{14}$ because of perfect blinding using $\share{\delta}{k^*}_0$.
	
	\item[-] $E_{16}$ (Figure \ref{fig:e16}): In $E_{16}$, shares $\vcs{v}{i_0}{k^*},\vcs{r}{i_0}{k^*},\vcs{v}{i_1}{k^*},\vcs{r}{i_1}{k^*}$ during the $\enc$ calls for $i_0,i_1$ are set as $0$. $E_{16}$ is indistinguishable from $E_{15}$ because these shares are not used anymore.

	\item[-] $E_{17}$ (Figure \ref{fig:e17}): In $E_{17}$, commitments $\vc{\gamma}{i_0},\vc{\gamma}{i_1}$ in above $\enc$ calls commit $0$. $E_{17}$ is indistinguishable from $E_{16}$ because Pedersen commitments are perfectly hiding and the committed values are not used anywhere.

	\item[-] $E_{18}$ (Figure \ref{fig:e18}): In $E_{18}$, $\vc{\tilde{\sigma}}{i_0}$, $\vc{\tilde{\sigma}}{i_1}$ during $\zkpsm$ are replaced by randomly drawn elements from $\Gone$. $E_{18}$ is indistinguishable from 
	$E_{17}$ because $\vcs{b}{i_0}{k^*}, \vcs{b}{i_1}{k^*}$ used in computing $\vc{\tilde{\sigma}}{i_0}$, $\vc{\tilde{\sigma}}{i_1}$ are chosen uniformly at random from $\Zq$.

	\item[-] $E_{19}$ (Figure \ref{fig:e19}): In $E_{19}$, $\vc{\tilde{S}}{j_0}'$, $\vc{\tilde{S}}{j_1}'$ during $\zkprsm$ are replaced by randomly drawn elements from $\Gone$; $\vc{\tilde{c}}{j_0}''$, $\vc{\tilde{c}}{j_1}''$ are computed as $\vc{\tilde{c}}{j_0}''' + \sum_{k \neq k^*} \vcs{b_c'}{j_0}{k}$, $\vc{\tilde{c}}{j_1}''' + \sum_{k \neq k^*} \vcs{b_c'}{j_1}{k}$ where $\vc{\tilde{c}}{j_0}''',\vc{\tilde{c}}{j_1}''' \xleftarrow{\$} \mathbb{Z}_{q+q^2}$ and $(\vcs{b_c'}{j_0}{k},\vcs{b_c'}{j_1}{k})_{k \neq k^*}$ are the ones obtained in $E_{10}$; and $\vc{\tilde{r}}{j_0}''$, $\vc{\tilde{r}}{j_1}''$ are computed as $\vc{\tilde{r}}{j_0}''' + \sum_{k \neq k^*} \vcs{b_r'}{j_0}{k}, \vc{\tilde{r}}{j_1}''' + \sum_{k \neq k^*} \vcs{b_r'}{j_1}{k}$ where $\vc{\tilde{r}}{j_0}''',\vc{\tilde{r}}{j_1}''' \xleftarrow{\$} \mathbb{Z}_{2q+q^2}$ and $(\vcs{b_r'}{j_0}{k},\vcs{b_r'}{j_1}{k})_{k \neq k^*}$ are the ones obtained in $E_{10}$. $E_{19}$ is indistinguishable from $E_{18}$ because $\vcs{b_S}{j_0}{k^*},\vcs{b_S}{j_1}{k^*}$ used in computing $\vc{\tilde{S}}{j_0}'$, $\vc{\tilde{S}}{j_1}'$ are chosen uniformly at random from $\Zq$ and $\vcs{b_c'}{j_0}{k^*},\vcs{b_c'}{j_1}{k^*}$ (resp. $\vcs{b_r'}{j_0}{k^*},\vcs{b_r'}{j_1}{k^*}$) used in computing $\vc{\tilde{c}}{j_0}''$, $\vc{\tilde{c}}{j_1}''$ (resp. $\vc{\tilde{r}}{j_0}''$, $\vc{\tilde{r}}{j_1}''$) are chosen uniformly at random from an exponentially larger space than the corresponding messages $\vc{c}{\pi(j_0)},\vc{c}{\pi(j_1)}$ (resp. 
	 $\vc{\mathsf{r}}{\pi(j_0)},\vc{\mathsf{r}}{\pi(j_1)}$).
\end{itemize}
In $E_{19}$, $\adv$'s view is identical for both $b=0$ and $b=1$ because in this experiment only $\ve{v}'$ sent at the end of the $\mix$ protocol depends on $v_{i_0},v_{i_1}$ and $\ve{v}'$ is identically distributed for any choice of $b$ because $\permfunc{k^*}$ is uniformly distributed over $\perm{n}$. Thus, by a standard hybrid argument, the theorem holds.

\subsection{Proof for Theorem \ref{thm:output-secrecy}}
\label{pf:output-secrecy}

In this case, the proof proceeds exactly as above except that in experiment $E_{13}$, DPKs for both $i_0,i_1$ are always simulated completely and in experiment $E_{14}$, DPKs for both $j_0,j_1$ are always simulated completely. $E_{13}$ and $E_{14}$ are indistinguishable from their previous experiments because in this case, even though $\adv$ is not restricted by the constraints at $\otrin/\otrout$ calls, $\adv$ does not control $\querier$ anymore and thus does not even learn whether a DPK passed or not by the property of the DPKs (see Section \ref{sss:distr-pok}).

\subsection{Secrecy in the malicious model}
\label{pf:secrecy-general}

In this section, we sketch a proof that when the additional steps of Appendix \ref{honest-but-curious-to-full-zk} are applied then our construction protects secrecy (Definition \ref{def:secr}) even against general malicious adversaries. We assume that secure MPC protocols are used for Beaver triples generation and for $\theg$ and $\thez$ distributed key generation. 

The proof essentially follows the structure of the proof in Appendix \ref{pf:secrecy} with the following differences:
\begin{itemize}[leftmargin=*]
	\item[-] In $E_5$, $\vcs{v}{i}{k^*}$, $\vcs{r}{i}{k^*}$ for $i \in [n] \setminus \{i_0,i_1\}$ are obtained by extracting them from proofs $(\vs{\rho}{k}_{\mathsf{ev}_i},\vs{\rho}{k}_{\mathsf{er}_i})_{k \in [m]}$; the corresponding proofs for $i \in \{i_0,i_1\}$ are simulated. $E_5$ is indistinguishable from $E_4$ by knowledge soundness of $\vs{\rho}{k}_{\mathsf{ev}_i},\vs{\rho}{k}_{\mathsf{er}_i}$ for $i \in [n] \setminus \{i_0,i_1\}$ and their zero-knowledgeness for $i \in \{i_0,i_1\}$.
	
	\item[-] In $E_8$, $(\permfunc{k})_{k \neq k^*}$ are obtained by extracting them from proofs of shuffle produced by $(\mixer{k})_{k \neq k^*}$ and $(\vc{v}{i})_{i \in [n]\setminus \{i_0,i_1\}}$ are obtained by extracting them from proofs $\vc{\rho}{\epsilon_i}$ for $i \in [n] \setminus \{i_0,i_1\}$. Proofs of shuffle for $\mixer{k^*}$ and proofs $\vc{\rho}{\epsilon_i}$ for $i \in \{i_0,i_1\}$ are simulated. $E_8$ is indistinguishable to $E_7$ by $a)$ knowledge soundness of proofs of shuffles for $(\mixer{k})_{k \neq k^*}$, $b)$ knowledge soundness of $\vc{\rho}{\epsilon_i}$ for $i \in [n]\setminus \{i_0,i_1\}$, $c)$ zero-knowledgeness of proofs of shuffle for $\mixer{k^*}$, $d)$ zero-knowledgeness of $\vc{\rho}{\epsilon_i}$ for $i \in \{i_0,i_1\}$, and $e)$ correctness of threshold decryption protocol $\thezd$.

	\item[-] In $E_9$, $(\vcs{b}{i}{k})_{k \neq k^*}$ are obtained by extracting them from $(\vcs{\rho}{\tilde{\sigma}_i}{k})_{k \neq k^*}$. The corresponding proofs for $k=k^*$ are simulated. $E_9$ is indistinguishable from $E_8$ by $a)$ verification of correct construction of $\vc{\epsilon}{\sigma}'$ from $\ve{\sigma}'$, $b)$ knowledge soundness of proofs of shuffle that prove that each $(\mixer{k})_{k \neq k^*}$ shuffled encrypted signatures in this phase correctly using the inverse of the permutation $\permfunc{k}$ it applied during mixing, $c)$ knowledge soundness of $(\vcs{\rho}{\tilde{\sigma}_i}{k})_{k \neq k^*}$, $d)$ zero-knowledgeness of proofs of shuffle for $k = k^*$, $e)$ zero-knowledgeness of $\vcs{\rho}{\tilde{\sigma}_i}{k^*}$, and $f)$ correctness of threshold decryption protocol $\thegd$.   

	\item[-] In $E_{10}$, $(\vcs{b_S}{j}{k},\vcs{b_c'}{j}{k},\vcs{b_r'}{j}{k})_{k \neq k^*}$ are obtained by extracting them from $(\vcs{\rho}{bS_j}{k},$ $\vcs{\rho}{bc_j}{k},$ $\vcs{\rho}{br_j}{k})$ respectively and $\vc{r}{\pi(j)}$ is obtained by extracting them from $\vc{\rho}{\epsilon_{r_{\pi(j)}}}$. Proofs $(\vcs{\rho}{bS_j}{k^*},$ $\vcs{\rho}{bc_j}{k^*},$ $\vcs{\rho}{br_j}{k^*})$ and $\vc{\rho}{\epsilon_{r_{i_0}}},\vc{\rho}{\epsilon_{r_{i_1}}}$ are simulated. $E_{10}$ is indistinguishable from $E_9$ by $a)$ verification of correct construction of $(\ve{\epsilon_S}, \ve{\epsilon_c}, \ve{\epsilon_{\hat{r}}})$ from $(\ve{S}, \ve{c}, \ve{\hat{r}})$, $b)$ knowledge soundness of the proofs of shuffle that prove that each $\mixer{k}$ shuffled encrypted signatures in this phase correctly using $\permfunc{k}$ applied during mixing, $c)$ knowledge soundness of $(\vcs{\rho}{bS_j}{k},$ $\vcs{\rho}{bc_j}{k},$ $\vcs{\rho}{br_j}{k})$, $d)$ zero-knowledgeness of proofs of shuffle for $k=k^*$, $e)$ zero-knowledgeness of $\vcs{\rho}{bc_j}{k^*},$ $\vcs{\rho}{br_j}{k^*}$ and $\vc{\rho}{\epsilon_{r_{i_0}}},\vc{\rho}{\epsilon_{r_{i_1}}}$, and $f)$ correctness of threshold decryption protocols $\thegd$ and $\thezd$.

	\item[-] In $E_{13}$, the correctness of blinded signatures is ensured through the additional verification that signatures $\vc{\sigma}{j}'$ for each $j \in J$ are valid BB signatures on $\vc{v}{j}'$ and are invalid signatures for $j \not\in J$.

	\item[-] In $E_{14}$, the correctness of blinded signatures is ensured through the additional verification that quasi-signatures $(\vc{S}{i},\vc{c}{i},\vc{\hat{r}}{i})$ for each $i \in I$ are valid BBS+ quasi signatures using commitment $\vc{\gamma}{i}$ and invalid for $i \not\in I$. Note that correctness of quasi-signatures implies correctness of full BBS+ signatures for $i \in \{i_0,i_1\}$ because encryptions $\vc{\epsilon}{r_i}$ are generated by honest senders. Further, in this case, $(\share{\delta}{k}_0, \share{\delta}{k}_1, \share{\delta}{k}_2)_{k \neq k^*}$ are obtained by extracting $\share{\delta}{k}_0$ from proofs of knowledge of the opening of $(\share{\mathfrak{z}}{k}_1)_{k \neq k^*}$ and computing $(\share{\delta}{k}_1, \share{\delta}{k}_2)_{k \neq k^*}$ using the $\mult$ algorithm on inputs $(\share{\delta}{k}_0, \vcs{b_S}{j}{k})$ and $(\share{\delta}{k}_0, \vcs{b_c}{j}{k})$ respectively, where $\vcs{b_S}{j}{k},\vcs{b_c}{j}{k}:=(\vcs{b_c'}{j}{k} \bmod q)$ are the ones obtained in $E_{10}$ (note that this requires Beaver triples held by $(\mixer{k})_{k \neq k^*}$ which are provided by the ideal functionality for the MPC protocol for Beaver triple generation).
\end{itemize}

\section{Privacy risk analysis of TraceIn/TraceOut queries}
\label{privacy-risk-analysis}

In this section, we provide formal rules to analyse the privacy risk impact of allowing a given set of TraceIn/TraceOut queries to a querier in an application. Since the goal of the analysis is to apriori decide whether to allow the queries or not, we assume that we only have query inputs and not their outputs. Further, we assume that no query's input depends on other queries' outputs. Given this, we adopt a static analysis approach where we conservatively estimate the information \emph{potentially} leaked by the queries. If the actual information leaked by a query depends on its output, we conservatively assume that the query potentially leaks the information corresponding to both the outputs.

Consider a mixnet with input ciphertexts $(\vc{c}{i})_{i \in [n]}$ and output plaintexts $(\vc{v}{j}')_{j \in [n]}$. Let $Q$ denote the set of proposed allowed queries containing elements of the form $(\mathsf{TraceIn},$ $i,$ $J)$ and $(\mathsf{TraceOut},$ $I,$ $j)$ for $i,j \in [n]$ and $I,J \subseteq [n]$ representing the TraceIn($i,J$) and TraceOut($I,j$) queries respectively ($\btrin$/$\btrout$ queries can be compiled to this format). Let $K(i, J)$ denote the fact that the querier potentially knows that ciphertext $\vc{c}{i}$ encrypts a plaintext in set $\vc{v}{J}'$. Let $K(I, j)$ (with an index set as the first argument and an index as the second argument) denote the fact that the querier potentially knows that plaintext $\vc{v}{j}'$ is encrypted in some ciphertext in set $\vc{c}{I}$. We therefore have the following rules:
\begin{itemize}[leftmargin=*]
	\item[-] \textbf{R0.1:} $\forall i \in [n]: K(i, [n])$.
	\item[-] \textbf{R0.2:} $\forall j \in [n]: K([n], j)$.
	\item[-] \textbf{R1.1:} $\forall i \in [n], J \subseteq [n]:$ $(\mathsf{TraceIn}, i, J) \in Q$ $\implies$ $K(i, J)$ $\wedge$ $K(i, [n] \setminus J)$.
	\item[-] \textbf{R1.2:} $\forall j \in [n], I \subseteq [n]:$ $(\mathsf{TraceOut}, I, j) \in Q$ $\implies$ $K(I, j)$ $\wedge$ $K([n] \setminus I, j)$.
	\item[-] \textbf{R2.1:} $\forall i \in [n], J,J' \subseteq [n]: K(i, J) \wedge K(i, J') \implies K(i, J \cap J')$.
	\item[-] \textbf{R2.2:} $\forall j \in [n], I,I' \subseteq [n]: K(I, j) \wedge K(I', j) \implies K(I \cap I', j)$.
	\item[-] \textbf{R3.1:} $\forall i,j \in [n], I,J \subseteq [n]: K(i, J) \wedge K(I, j) \wedge i \not\in I \wedge j \in J \implies K(i,  J \setminus \{j\})$.
	\item[-] \textbf{R3.2:} $\forall i,j \in [n], I,J \subseteq [n]: K(i, J) \wedge K(I, j) \wedge i \in I \wedge j \not\in J \implies K(I \setminus \{i\},  j)$.
\end{itemize}

Rules R0.1 and R0.2 encode the fact that the querier initially knows that each input ciphertext encrypts a plaintext in the output list and each output plaintext is encrypted in a ciphertext in the input list. R1.1 (symmetrically R1.2) encodes that a TraceIn($i, J$) query leads to either the querier learning that $\vc{c}{i}$ encrypts a plaintext in $\vc{v}{J}'$ (if the query passed) or that $\vc{c}{i}$ encrypts a plaintext in $\vc{v}{[n] \setminus J}'$ (if the query failed). Thus, potentially, the querier may learn both. R2.1 (symmetrically R2.2) encodes that if the querier knows that $\vc{c}{i}$ encrypts a value in $\vc{v}{J}'$ and in $\vc{v}{J'}'$ then it knows that $\vc{c}{i}$ encrypts a value in $\vc{v}{J \cap J'}'$. R3.1 (symmetrically R3.2) encodes that if the querier knows that $\vc{c}{i}$ encrypts a value in $\vc{v}{J}'$ but there exists a $j \in J$ that is known to be encrypted in a ciphertext in a set $I$ that does not include $i$, then it knows that $\vc{c}{i}$ cannot encrypt $\vc{v}{j}'$ because of distinctness of encrypted values.

Given these rules, the goal of our analysis is to output $a)$ for each $i \in [n]$, the smallest set $J_{\text{min}_i}$ such that $K(i,J_{\text{min}_i})$ holds; and $b)$ for each $j \in [n]$, the smallest set $I_{\text{min}_j}$ such that $K(I_{\text{min}_j}, i)$ holds. Note that if for some $i$, no TraceIn($i,J$) query exists in $Q$ then $J_{\text{min}_i}=[n]$ (similarly for $I_{\text{min}_j}$). The $J_{\text{min}},I_{\text{min}}$ outputs tell exactly what information is potentially leaked by the queries in $Q$. This information can directly feed to the application-level security analysis. Since each implication rule monotonically decreases the size of either the $J$ set for $K(i,J)$ or the $I$ set for $K(I,j)$, $J_{\mathsf{min}_i}$ and $I_{\mathsf{min}_j}$ can both be obtained in finite time for a finite set $Q$ using a fixed-point algorithm.

\begin{figure*}
	\centering
	\scalebox{0.8}{

	}
	\caption{$E_1$: The original $\secr$ experiment instantiated for our construction.}
	\label{fig:e1}
	\end{figure*}

	\begin{figure*}
		\centering
		\scalebox{0.8}{
			%
		}
		\caption{$E_2$: Simulating $\vc{\pokcomm}{i_0}$ and $\vc{\pokcomm}{i_1}$.}
		\label{fig:e2}
	\end{figure*}

	\begin{figure*}
		\centering
		\scalebox{0.8}{
			%
		}
		\caption{$E_3$: Send random values as shares to $(\mixer{k})_{k \neq k^*}$.}
		\label{fig:e3}
	\end{figure*}
	
	\begin{figure*}
		\centering
		\scalebox{0.8}{
			%
		}
		\caption{$E_4$: Do not use $\pkesk{k^*}$ for decrypting $\vcs{\mathsf{ev}}{i}{k^*}$, $\vcs{\mathsf{er}}{i}{k^*}$ for $i \in \{i_0,i_1\}$.}
		\label{fig:e4}
	\end{figure*}

	\begin{figure*}
		\centering
		\scalebox{0.8}{
			%
		}
		\caption{$E_5$: Do not use $\pkesk{k^*}$ for decrypting $\vcs{\mathsf{ev}}{i}{k^*}$, $\vcs{\mathsf{er}}{i}{k^*}$ for $i \in [n] \setminus \{i_0,i_1\}$.}
		\label{fig:e5}
	\end{figure*}

	\begin{figure*}
		\centering
		\scalebox{0.8}{
			%
		}
		\caption{$E_6$: Replace $\vcs{\mathsf{ev}}{i}{k^*}$, $\vcs{\mathsf{er}}{i}{k^*}$ for $i \in \{i_0,i_1\}$ by encryptions of zeros.}
		\label{fig:e6}
	\end{figure*}

	\begin{figure*}
		\centering
		\scalebox{0.8}{
			%
		}
		\caption{$E_7$: Simulate threshold decryption protocols using ideal decryption oracles.}
		\label{fig:e7}
	\end{figure*}

	\begin{figure*}
		\centering
		\scalebox{0.8}{
			%
		}
		\caption{$E_8$: Obtain $\ve{v}'$ without decryption.}
		\label{fig:e8}
	\end{figure*}

	\begin{figure*}
		\centering
		\scalebox{0.8}{
			%
		}
		\caption{$E_9$: Obtain $\ve{\tilde{\sigma}}$ during $\zkpsm$ without decryption.}
		\label{fig:e9}
	\end{figure*}

	\begin{figure*}
		\centering
		\scalebox{0.8}{
			%
		}
		\caption{$E_{10}$: Obtain $\ve{\tilde{S}}', \ve{\tilde{c}}'', \ve{\tilde{r}}''$ during $\zkprsm$ without decryption.}
		\label{fig:e10}
	\end{figure*}
	
	\begin{figure*}
		\centering
		\scalebox{0.8}{
			%
		}
		\caption{$E_{11}$: Send encryptions of zeros in place of all $\thez$ encryptions.}
		\label{fig:e11}
	\end{figure*}

	\begin{figure*}
		\centering
		\scalebox{0.8}{
			%
		}
		\caption{$E_{12}$: Send encryptions of $g_1^0$ in place of all $\theg$ encryptions.}
		\label{fig:e12}
	\end{figure*}
	
	\begin{figure*}
		\centering
		\scalebox{0.8}{
			%
		}
		\caption{$E_{13}$: Simulate DPKs in $\zkpsm$ call.}
		\label{fig:e13}
	\end{figure*}
	
	\begin{figure*}
		\centering
		\scalebox{0.8}{
			%
		}
		\caption{Same as $E_{13}$ but DPK simulation steps during $\zkpsm$ are abstracted out.}
		\label{fig:e13.5}
	\end{figure*}

	\begin{figure*}
		\centering
		\scalebox{0.8}{
			%
		}
		\caption{$E_{14}$: Simulate DPKs during $\zkprsm$.}
		\label{fig:e14}
	\end{figure*}

	\begin{figure*}
		\centering
		\scalebox{0.8}{
			%
		}
		\caption{$E_{15}$: Use random group element for $\share{\mathfrak{z}}{k^*}_1$.}
		\label{fig:e15}
	\end{figure*}

	\begin{figure*}
		\centering
		\scalebox{0.8}{
			%
		}
		\caption{$E_{16}$: Set $\vcs{v}{i_0}{k^*},\vcs{r}{i_0}{k^*},\vcs{v}{i_1}{k^*},\vcs{r}{i_1}{k^*}$ as zeros (also abstracted out the DPK simulation steps during $\zkprsm$).}
		\label{fig:e16}
	\end{figure*}

	\begin{figure*}
		\centering
		\scalebox{0.8}{
			%
		}
		\caption{$E_{17}$: Set $\vc{\gamma}{i_0},\vc{\gamma}{i_1}$ as commitments of zeros.}
		\label{fig:e17}
	\end{figure*}

	\begin{figure*}
		\centering
		\scalebox{0.8}{
			%
		}
		\caption{$E_{18}$: Replace $\vc{\tilde{\sigma}}{i_0}, \vc{\tilde{\sigma}}{i_1}$ by elements randomly drawn from $\Gone$.}
		\label{fig:e18}
	\end{figure*}

    \begin{figure*}
		\centering
		\scalebox{0.8}{
			%
		
			\end{tabular}
		}
		\caption{$E_{19}$: Replace $\vc{\tilde{S}'}{j}, \vc{\tilde{c}''}{j}, \vc{\tilde{r}''}{j}$ for $j \in \{\pi^{-1}(i_0), \pi^{-1}(i_1)\}$ by randomly drawn elements.}
		\label{fig:e19}
	\end{figure*}

\end{document}